\documentclass[acmsmall,screen,nonacm]{acmart}

\setcopyright{none}
\usepackage{xparse}
\usepackage{ebproof} 
\usepackage{setspace}
\usepackage{stackengine}
\usepackage{mathpartir}
\usepackage{pifont}
\usepackage{thm-restate}
\usepackage{enumitem}
\usepackage{cleveref}
\usepackage{algorithm}
\usepackage{algpseudocode}
\usepackage{tikz}
\usepackage{graphicx}
\usepackage{longtable}
\usepackage{xcolor}
\usetikzlibrary{positioning,calc,matrix,fit,arrows.meta}
\usepackage{subcaption}
\usepackage{listings}
\usepackage[most]{tcolorbox}
\usepackage{misc/mpst-macros} 
\usepackage{xr}

\title{Top-down $=$ Bottom-up}
\subtitle{Sound and Complete Characterisations of Liveness by Multiparty Global Protocols}

\author{Kai Pischke}
\affiliation{%
  \institution{University of Oxford}
  \department{Department of Computer Science}
  \city{Oxford}
  \country{United Kingdom}
}
\email{kai.pischke@cs.ox.ac.uk}
\orcid{0000-0002-6435-9456}

\author{Nobuko Yoshida}
\affiliation{%
  \institution{University of Oxford}
  \department{Department of Computer Science}
  \city{Oxford}
  \country{United Kingdom}
}
\email{nobuko.yoshida@cs.ox.ac.uk}
\orcid{0000-0002-3925-8557}

\begin{document}
\begin{abstract}
    \emph{Multiparty session types} (MPST) are a type discipline for concurrent and distributed systems, 
    designed to ensure not only type safety and deadlock-freedom, but also \emph{liveness} of typed communicating processes. 
    Two main MPST methodologies, top-down and bottom-up, have been proposed and are integrated into a wide range of programming languages and tools. 
    The \emph{top-down} strategy starts by specifying the overall choreography of the protocol (called a global type), 
    from which a set of local types that satisfy safety and liveness are generated by endpoint projection (EPP).
    Once each participant is type-checked against a generated local type, liveness of the set of typed processes is automatically ensured \emph{by construction}. 
    The \emph{bottom-up} strategy directly checks whether local types inferred from processes satisfy liveness in order to enforce liveness of processes. 
    Since the top-down strategy depends on global types and the EPP algorithms, 
    it has often been considered that the top-down system offers \emph{strictly less typability} than the bottom-up system. 
    Our paper negates this belief. We prove that, using the precise subtyping for the subsumption rule, 
    the top-down strategy offers \emph{exactly the same typability} as the bottom-up system. 
    More precisely, a multiparty session $M$ is typable and verified to be live by the bottom-up typing system if and only if $M$ is typable by the top-down typing system. 
    The key to the proof is the development of a \emph{principal global type
inference algorithm} 
which builds a \emph{principal} global type 
from an arbitrary set of live local types. 
    We have implemented the global type inference algorithm together with
    projection, process type checking and local type
    inference algorithms, and built a toolchain
for both the top-down and bottom-up strategies.
We evaluated our toolchain with representative examples 
from the literature, confirming that the top-down approach is more efficient
than the bottom-up approach. 
\end{abstract}

\ccsdesc[500]{Theory of computation~Concurrency}
\ccsdesc[500]{Theory of computation~Type theory}
\ccsdesc[500]{Theory of computation~Process calculi}
\ccsdesc[500]{Theory of computation~Distributed computing models}

\maketitle
\nocite{*}

\section{Introduction}
\label{introduction}
\emph{Multiparty session types} (MPST) \cite{HYC2016,Honda2008} are a
type framework for distributed services and agents, 
designed to describe and verify communication protocols between 
\emph{multiple} distributed participants. 
The original theory takes 
the \emph{top-down} strategy (Figure~\ref{fig:topdown}):  
a protocol designer first prescribes the overall 
choreography (who sends what to whom) 
as a \emph{global type} $\G$.  
Then an endpoint projection (EPP) algorithm
uses this global type to generate a set of local types $\{\T[i]\}_{i\in I}$
which constrain how each individual participant $\PP[i]$ can behave. 
Type-checking guarantees that typed processes $\{\PP[i]\}_{i\in I}$
satisfy \emph{communication safety} (there are no unexpected type and label mismatches), 
\emph{deadlock-freedom} and \emph{session fidelity} (each process adheres to the declared global protocol). 
In the top-down system, 
typability can ensure a stronger property, 
\emph{liveness}, i.e., ``something good eventually happens''.  
Here ``something good'' means that a request eventually gets a response, 
and each participant in the protocol is able to make \emph{progress} 
without stalling indefinitely. 
The top-down strategy has been integrated into 
many programming languages, taking  
various forms 
such as code generation, API generation, API inference,
event-driven programming, 
static type checking, 
protocol compliance, real-time systems and  
runtime monitoring.  
Some notable examples include integration with 
languages such as Rust~\cite{CYV2022,DBLP:conf/ecoop/LagaillardieNY22,DBLP:conf/ecoop/HouLY24,DBLP:conf/ecoop/VassorY24}, 
Java~\cite{FASE16EndpointAPI,HY2017,DBLP:journals/scp/KouzapasDPG18,DBLP:conf/tacas/BoumaGJ23},
Scala~\cite{SDHY2017,OOPSLA21FaultTolerantMPST,ECOOP22MPSTScala,BHYZ2023,DBLP:conf/issta/FerreiraJ23,DBLP:conf/tacas/Jongmans25}, 
Go~\cite{DBLP:journals/pacmpl/CastroHJNY19}, 
TypeScript~\cite{MFYZ2021,DBLP:conf/ecoop/GheriLSTY22},
PureScript~\cite{DBLP:journals/corr/abs-1904-01287}, 
OCaml~\cite{YZF2021,INYY2020},
MPI-C~\cite{NCY2015,LMMNSVY2015},
F$\star$~\cite{ZFHNY2020},
F$\sharp$~\cite{NHYA2018},
Python~\cite{DBLP:conf/rv/NeykovaYH13,NY2017Actor,DBLP:journals/fmsd/DemangeonHHNY15},
Erlang~\cite{DBLP:journals/corr/Fowler16,NY2017,DBLP:conf/icsoft/EgidiGV22,BHVT2026,FH2026},
C~\cite{NYH2012,NYNTL2012}, 
C$\sharp$~\cite{DBLP:journals/corr/abs-2004-01325},
and domain-specific languages~\cite{larsen_et_al:LIPIcs.ECOOP.2026.17,DBLP:journals/corr/abs-1208-4632,DBLP:journals/scp/CharalambidesDA16,DBLP:conf/ecoop/00020DG21,Fowler2026}.    
See~\cite{Y2024} for a survey.
Mechanisations of the top-down system have been 
developed across different frameworks to ensure correctness of the
top-down system~\cite{DBLP:conf/pldi/Castro-Perez0GY21,DBLP:conf/ecoop/TiroreBC25,DBLP:journals/jar/TiroreBC25,DBLP:journals/pacmpl/JacobsBK22a,EKY2025,DBLP:conf/itp/LiW25,Castro2026,KYG2026}. 

The shortcoming of the top-down approach is that \emph{typability is 
limited} since it depends on the EPP algorithm 
being used. If $\G$ is not projectable by the chosen EPP algorithm, 
either $\G$ is not \emph{implementable} (there exists no correct set 
of local types) or this particular EPP algorithm rejects some implementable $\G$. 
If the EPP is too limited, 
many safe-and-live processes become untypable.
As expected, computationally cheaper EPP algorithms tend to  
give fewer typable processes than more expensive ones 
\cite{thien-nobuko-popl-25}.  

\begin{figure}[!t]
\centering
\begin{subfigure}[b]{0.24\textwidth}
\centering
\begin{tikzpicture}[
    every node/.style={align=center, font=\scriptsize},
    level/.style={minimum width=0.5cm, minimum height=0.25cm, inner sep=1pt},
    arrow/.style={->, semithick, black!90},
    groupbox/.style={draw, gray!40, rounded corners=2pt, inner sep=4pt, line width=1pt},
    annotation/.style={font=\tiny, text=blue!60}
]
\node[level] (G1) at (0.225,0.5) {$\G$};
\node[level] (T1)  at (-0.55,-0.25) {$\T[1\text{~pro}]$};
\node[level] (T2)  at (0.05,-0.25) {$\T[2\text{~pro}]$};
\node[level] (T3)  at (0.55,-0.25) {$\cdots$};
\node[level] (Tn)  at (1.0,-0.25) {$\T[n\text{~pro}]$};
\node[level] (S1) at (-0.55,-0.55) {\rotatebox[origin=c]{90}{$\subt$}};
\node[level] (S2) at (0.05,-0.55) {\rotatebox[origin=c]{90}{$\subt$}};
\node[level] (S3) at (0.55,-0.55) {};
\node[level] (Sn) at (1.0,-0.55) {\rotatebox[origin=c]{90}{$\subt$}};
\node[level] (T11) at (-0.55,-0.85) {$\T[1]$};
\node[level] (T12) at (0.05,-0.85) {$\T[2]$};
\node[level] (T13) at (0.55,-0.85) {$\cdots$};
\node[level] (T1n) at (1.0,-0.85) {$\T[n]$};
\node[level] (P11) at (-0.55,-2.05) {$\PP[1]$};
\node[level] (P12) at (0.05,-2.05) {$\PP[2]$};
\node[level] (P13) at (0.55,-2.05) {$\cdots$};
\node[level] (P1n) at (1.0,-2.05) {$\PP[n]$};
\draw[->] (G1) -- (T1);
\draw[->] (G1) -- (T2);
\draw[->] (G1) -- (Tn);
\draw[arrow] (T11) -- (P11);
\draw[arrow] (T12) -- (P12);
\draw[arrow] (T1n) -- (P1n);
\end{tikzpicture}
\caption{\small top-down}
\label{fig:topdown}
\end{subfigure}%
\hspace{1mm}%
\begin{subfigure}[b]{0.24\textwidth}
\centering
\begin{tikzpicture}[
    every node/.style={align=center, font=\scriptsize},
    level/.style={minimum width=0.5cm, minimum height=0.25cm, inner sep=1pt},
    arrow/.style={->, semithick, black!90},
    groupbox/.style={draw, gray!40, rounded corners=2pt, inner sep=4pt, line width=1pt},
    annotation/.style={font=\tiny, text=blue!60}
]
\coordinate (top) at (0.225,0.3);
\node[level] (Tp21) at (-0.55,-1.4) {$\T[1\text{~min}]$};
\node[level] (Tp22) at (0.05,-1.4) {$\T[2\text{~min}]$};
\node[level] (Tp23) at (0.55,-1.4) {$\cdots$};
\node[level] (Tp2n) at (1.0,-1.4) {$\T[n\text{~min}]$};
\node[level] (P21) at (-0.55,-2.05) {$\PP[1]$};
\node[level] (P22) at (0.05,-2.05) {$\PP[2]$};
\node[level] (P23) at (0.55,-2.05) {$\cdots$};
\node[level] (P2n) at (1.0,-2.05) {$\PP[n]$};
\node[groupbox, fit=(Tp21)(Tp2n)] (Tbox2) {};
\node[annotation, fill=white, inner sep=1.5pt, anchor=south west, xshift=2pt, yshift=2pt] at (Tbox2.north west) {$\safe + \live$};
\draw[->] (P21) -- (Tp21);
\draw[->] (P22) -- (Tp22);
\draw[->] (P2n) -- (Tp2n);
\end{tikzpicture}
\caption{\small bottom-up}
\label{fig:bottomup}
\end{subfigure}%
\hspace{1mm}%
\begin{subfigure}[b]{0.24\textwidth}
\centering
\begin{tikzpicture}[
    every node/.style={align=center, font=\scriptsize},
    level/.style={minimum width=0.5cm, minimum height=0.25cm, inner sep=1pt},
    arrow/.style={->, semithick, black!90},
    groupbox/.style={draw, gray!40, rounded corners=2pt, inner sep=4pt, line width=1pt},
    annotation/.style={font=\tiny, text=blue!60}
]
\node[level] (G3) at (0.225,0) {$\G[\text{inf}]$};
\node[level] (T31) at (-0.55,-0.85) {$\T[1]$};
\node[level] (T32) at (0.05,-0.85) {$\T[2]$};
\node[level] (T33) at (0.55,-0.85) {$\cdots$};
\node[level] (T3n) at (1.0,-0.85) {$\T[n]$};
\node[level] (P31) at (-0.55,-2.05) {$\PP[1]$};
\node[level] (P32) at (0.05,-2.05) {$\PP[2]$};
\node[level] (P33) at (0.55,-2.05) {$\cdots$};
\node[level] (P3n) at (1.0,-2.05) {$\PP[n]$};
\node[groupbox, fit=(T31)(T3n)] (Tbox3) {};
\node[annotation, fill=white, inner sep=1.5pt, anchor=south west, xshift=2pt, yshift=2pt] at (Tbox3.north west) {$\safe$\\$+\,\live$};
\draw[->, thick, red!80] (T31) -- (G3);
\draw[->, thick, red!80] (T32) -- (G3);
\draw[->, thick, red!80] (T3n) -- (G3);
\draw[arrow] (T31) -- (P31);
\draw[arrow] (T32) -- (P32);
\draw[arrow] (T3n) -- (P3n);
\end{tikzpicture}
\caption{global type inference}
\label{fig:topdown3}
\end{subfigure}%
\hfill%
\begin{subfigure}[b]{0.24\textwidth}
\centering
\begin{tikzpicture}[
    every node/.style={align=center, font=\scriptsize},
    level/.style={minimum width=0.5cm, minimum height=0.25cm, inner sep=1pt},
    arrow/.style={->, semithick, black!90},
    groupbox/.style={draw, gray!40, rounded corners=2pt, inner sep=4pt, line width=1pt},
    annotation/.style={font=\tiny, text=blue!60}
]
\node[level] (G4) at (0.225,-0.4) {$\G[\text{min}]$};
\node[level] (Tp41) at (-0.55,-1.4) {$\T[1\text{~min}]$};
\node[level] (Tp42) at (0.05,-1.4) {$\T[2\text{~min}]$};
\node[level] (Tp43) at (0.55,-1.4) {$\cdots$};
\node[level] (Tp4n) at (1.0,-1.4) {$\T[n\text{~min}]$};
\node[level] (P41) at (-0.55,-2.05) {$\PP[1]$};
\node[level] (P42) at (0.05,-2.05) {$\PP[2]$};
\node[level] (P43) at (0.55,-2.05) {$\cdots$};
\node[level] (P4n) at (1.0,-2.05) {$\PP[n]$};
\node[groupbox, fit=(Tp41)(Tp4n)] (Tbox4) {};
\node[annotation, fill=white, inner sep=1.5pt, anchor=south west, xshift=2pt, yshift=2pt] at (Tbox4.north west) {$\safe$\\$+\,\live$};
\draw[->] (P41) -- (Tp41);
\draw[->] (P42) -- (Tp42);
\draw[->] (P4n) -- (Tp4n);
\draw[->, thick, red!80] (Tp41) -- (G4);
\draw[->, thick, red!80] (Tp42) -- (G4);
\draw[->, thick, red!80] (Tp4n) -- (G4);
\end{tikzpicture}
\caption{principal type inf.}
\label{fig:bottomup4}
\end{subfigure}
\vspace*{-2mm}
\caption{Multiparty session type frameworks.
We have $\G[\text{min}] \subseteq \G[\text{inf}] \subseteq \G$, where $\subseteq$ denotes trace inclusion.}
\label{fig:methodologies}
\vspace*{-5mm}
\end{figure}
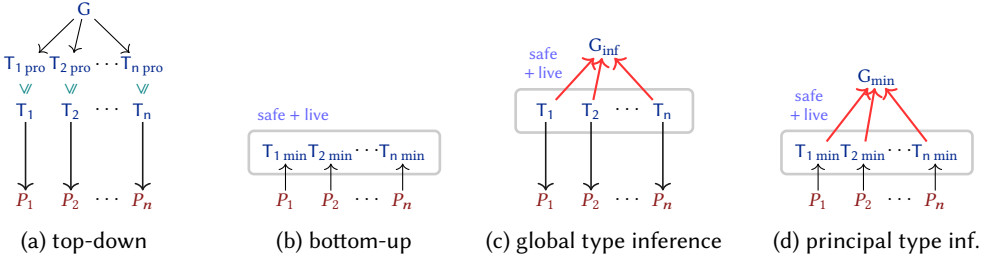

The bottom-up strategy proposed in \cite{Scalas2019} aims to solve this issue. 
It directly verifies that a set of local types (either prescribed by the programmer 
or inferred from a set of processes) 
satisfies some desired property $\varphi$ (Figure~\ref{fig:bottomup}). 
The benefit of this approach is the ability to verify a wide variety of properties 
of typed processes. For instance, it is possible to 
check whether processes satisfy type safety but not liveness.  
Since local types resemble processes (such as CCS \cite{MilnerR:comc} 
and CSP \cite{Hoare:1985:CSP:3921}) 
or state machines, this approach 
can directly use off-the-shelf model checkers
(such as mCRL2 \cite{mcrl2} or SPIN \cite{spin}) to verify properties. 
The bottom-up system is integrated into Scala \cite{BSYZ2022,SYB2019}, 
OCaml \cite{DBLP:conf/tacas/ImaiLN22}, 
Go \cite{NY2016,LNTY2018} 
and Rust \cite{CYV2022,DBLP:conf/ecoop/LagaillardieNY22}, 
and certified using Rocq \cite{DBLP:journals/pacmpl/HinrichsenJK24,DBLP:conf/itp/EkiciY24,BY2026,HQB2026} and Why3 \cite{DBLP:conf/esop/GiuntiY25}. 

The shortcoming of this approach is that it is \emph{computationally
expensive}.
Verifying safety, deadlock-freedom and liveness of a set of local types
is \PSPACE-complete
\cite{thien-nobuko-popl-25} in the synchronous
semantics, and is \emph{undecidable} in the asynchronous
semantics \cite{Brand1983}.
By contrast, the top-down approach requires only projection and
subtyping, both of which are polynomial in the size of the types
\cite{thien-nobuko-popl-25}.

These two approaches represent a trade-off. The top-down
approach demands more effort as $\G$ must be supplied upfront;
however, this can offer substantially cheaper verification and a
human-readable specification, while the bottom-up approach makes
no design-time demand but pays at verification time, with known
hard lower bounds (\PSPACE-complete synchronously,
undecidable asynchronously). Our main result concerns 
the respective expressiveness of these two methodologies
($\Ptop = \Pbot$), not the benefits and drawbacks of each workflow. 
In particular, we show the two methodologies accept exactly the same processes, 
and hence the practical choice should be motivated by other factors.

\textbf{This paper} gives results with unexpected conclusions
related to typability of the top-down
and bottom-up systems. Our focus is the \emph{liveness} property.
To explain this, let us assume a process $\PP[i]$ that acts as 
participant $\pp[i]$ (denoted by $\proctag{\pp[i]}{\PP[i]}$)  
has session type $\T[i]$. A \emph{multiparty session} $\M$
is a parallel composition of participants 
($\M=\proctag{\pp[1]}{\PP[1]} \ | \  \cdots\ | \ \proctag{\pp[n]}{\PP[n]}$).
The bottom-up typing judgement 
is given as 
$\provesbot \M:\ctx$ where
$\ctx$ is a \emph{safe typing context} mapping each participant $\pp[i]$ to its type $\T[i]$.
``Safety'' on $\ctx$ is a minimum requirement on the typing
judgement $\provesbot$ to ensure that $\M$ is type safe.

We define \emph{a set of typable live sessions},   
$\Pbot$, as a set of multiparty sessions 
which are typable by the bottom-up system:\\[1mm]
\centerline{$\Pbot \ = \ \{\ \M \ | \ \provesbot \M:\ctx \quad \text{such that} \ \ctx \ \text{is live}\}
$}\\[1mm] 
If $\M$ is typed by live $\ctx$, then $\M$ satisfies liveness
\cite{Scalas2019,revisited}. This set is closed under reduction, i.e.,
if $\M \in \Pbot$ and 
$\M$ reduces to $\Mi$, then we have
$\Mi \in \Pbot$.
Thus $\Pbot$ gives a \textbf{complete characterisation of liveness}, i.e., $\M$ is typable and live if
$\M\in \Pbot$ and vice versa, and $\Pbot$ is closed under reductions.

In the top-down system, 
each participant is 
typed by $\T[i]$, which is a projection of global type $\G$ on
participant $\role{p_i}$ 
(denoted by $\G \cofullproj{p_i} \T[i]$). 
Assuming $\PP[i]$ is typed by $\T[i]$, we define:\\[1mm] 
\centerline{$
\Ptop \ = \ \{\ \M \ | \ \provestop \M: 
\ctx \quad \text{such that}\ \G \cofullproj{p_i}\T[i] \quad
\ctx=\set{\AT{\role{p_i}}{\T[i]}}_{i\in I}\ \text{for some balanced} \ \G\}
$}\\[1mm] 
This implies: 
if $\M\in \Ptop$, then 
$\M\in \Pbot$, i.e., 
$\Ptop \, \subseteq \, \Pbot$. 
Depending on the EPP, we may have
$\Ptop \, \subsetneq \, \Pbot$, i.e., $\M\in \Pbot$ but 
$\M\not\in \Ptop$. 
This is because
$\G \cofullproj{p_i} \T[i]$ might be undefined even if
$\G$ is in fact implementable by safe processes satisfying liveness.

\textbf{This paper} shows the typability given by the top-down system 
is the same as that given by the bottom-up system, i.e., 
$\Ptop \, = \, \Pbot$, against 
the common belief that the top-down strategy is incomplete.  
\textbf{The top-down strategy is complete with respect to liveness} if 
three conditions are met. 
First, the typing rule for a process 
includes the subsumption rule defined by  
\emph{precise synchronous subtyping} $\subt$ 
\cite{synchronous-subtyping}. Secondly,
$\G$ is \emph{balanced} 
(i.e., any participant
appearing in a subterm of $\G$ is unavoidable);  
and thirdly, the EPP algorithm
we use is the most
expressive and computationally expensive EPP algorithm among 
the four EPP algorithms proposed in the MPST literature 
\cite{thien-nobuko-popl-25}, 
namely the \emph{coinductive full-merging projection}.

To prove this, we develop a \emph{principal global type inference
algorithm} that synthesises a balanced global type from any safe and
live typing context $\ctx$ (Figure~\ref{fig:topdown3}). The key idea
is to equip contexts with behavioural semantics and follow a
deterministic fair schedule, so that synthesis extracts a
representative interleaving while turning repeated semantic states
into syntactic $\toc{\sfmu}$-binders. We then show that the resulting
global type is projectable by the coinductive full-merging projection,
preserves exactly the traces of $\ctx$, and is therefore a principal
type for $\ctx$.

A recent line of work establishes \emph{complete} EPPs in CFSM-based settings
with respect to realisability and deadlock-freedom
\cite{Li2023,DiGiustoUrsoLozesPPDP2025,Stutz2023PowerOfChoice}.
We make use of such EPPs for a different completeness result, 
namely the completeness of the top-down system with respect to
bottom-up \emph{typability of processes}.
Notice that these are different notions of completeness: 
the bottom-up to top-down typability question does not reduce to
EPP-completeness. Instead,
we give a synthesis procedure which introduces two coinductive
devices absent from the inductive, syntactic projection literature, a
combined projection-and-branch-subtyping relation
(Lemma~\ref{lemma:soundness-combined-relation}) and 
a fair, round-robin scheduling discipline (\SEC{sec:semantics}).

We implement a toolchain (\SEC{sec:implementation}) which
consists of (1) \textbf{the top-down system} with the four projection
algorithms in \cite{thien-nobuko-popl-25}, subtyping checking
and process type checking; and (2) \textbf{the bottom-up system}
with the minimum local type inference from a given process,
safety and liveness checking of a typing context
and the principal global type inference algorithm.
We evaluate the toolchain on a range of examples from the literature,
including particularly challenging global types from
\cite{Scalas2019}, \cite{Li2023}, and \cite{Castro2026},
all of which can be projected top-down and whose global types can also be inferred bottom-up.
We further show that the top-down strategy can often be more efficient
than off-the-shelf bottom-up model checkers such as {\sf mpstk}~\cite{Scalas2019}.

\myparagraph{Outline}
\SEC{sec:overview}
gives an overview of our framework and results.  
\SEC{sec:calculus}
summarises the calculus and defines 
the properties needed in this paper. 
\SEC{sec:types} defines global and local types and projection.
\SEC{sec:semantics} equips types with labelled transition semantics and fair scheduling.
\SEC{sec:building}
is the main section:
it presents the inference algorithm and proves its correctness with respect to liveness. 
\SEC{sec:typing}  
proves the completeness of the top-down typing system with respect to
live processes, proving
$\Ptop=\Pbot$. 
\SEC{sec:implementation}  
outlines the toolchain and benchmarks
it on case studies from the literature. 
\SEC{sec:related} discusses related work and concludes the
paper. 
\ifnotsplit{The appendix contains
the full proofs and more examples.}{The full
version~\cite{fullversion} contains the full proofs and more
examples.}
Information about the toolchain can be found in 
the data availability statement (\SEC{sec:dataavailability}). 

\section{Overview}
\label{sec:overview}
As a running example, we consider a \emph{Monte Carlo protocol}
for estimating~$\pi$, a variation of the MapReduce protocol
from~\cite{Scalas2019}. Our goal is to prove
that any safe-and-live typing context
$\ctx$ has a corresponding global type $\G$.
For this, we build $\G$ from $\ctx$ which is
\emph{inferred} from a live session $\M$.

We consider four participants:
\emph{manager} $\role{m}$, \emph{workers} $\role{w_1}$, $\role{w_2}$,
and
\emph{reducer} $\role{r}$.
In each round, the manager selects \labname{map} for each
worker and sends a random float in $[0,1]$; each worker computes
$4\sqrt{1 - x^2}$ and reports the result to the reducer, who 
decides whether to continue or stop.  If it stops, the manager
selects \labname{stop} for both workers and the protocol terminates.
The average of all reported values estimates~$\pi$.
Figure~\ref{fig:protocol-steps} visualises one round of this workflow.

\begin{figure}[tb]
\centering
\begin{subcaptionblock}{0.3\linewidth}
\centering
\begin{tikzpicture}[
  >=Stealth,
  role/.style={draw, rounded corners=3pt, fill=gray!10, minimum width=1cm, minimum height=0.45cm, font=\scriptsize, align=center},
  idle/.style={draw, rounded corners=3pt, fill=white, minimum width=1cm, minimum height=0.45cm, font=\scriptsize, align=center, text=black!30, draw=black!30},
  msg/.style={->, semithick}
]
  \node[role] (m) at (0,1.6) {$\role{m}$};
  \node[role] (w1) at (-1.3,0.55) {$\role{w_1}$};
  \node[role] (w2) at (1.3,0.55) {$\role{w_2}$};
  \node[idle] (r) at (0,-0.55) {$\role{r}$};
  \draw[msg] (m) -- node[left, font=\tiny] {$\pv{x}$} (w1);
  \draw[msg] (m) -- node[right, font=\tiny] {$\pv{x}$} (w2);
\end{tikzpicture}
\caption{Dispatch}
\label{fig:step-dispatch}
\end{subcaptionblock}%
\hfill
\begin{subcaptionblock}{0.3\linewidth}
\centering
\begin{tikzpicture}[
  >=Stealth,
  role/.style={draw, rounded corners=3pt, fill=gray!10, minimum width=1cm, minimum height=0.45cm, font=\scriptsize, align=center},
  idle/.style={draw, rounded corners=3pt, fill=white, minimum width=1cm, minimum height=0.45cm, font=\scriptsize, align=center, text=black!30, draw=black!30},
  msg/.style={->, semithick}
]
  \node[idle] (m) at (0,1.6) {$\role{m}$};
  \node[role] (w1) at (-1.3,0.55) {$\role{w_1}$};
  \node[role] (w2) at (1.3,0.55) {$\role{w_2}$};
  \node[role] (r) at (0,-0.55) {$\role{r}$};
  \draw[msg] (w1) -- node[left, font=\tiny] {$\pv{4\sqrt{1-x^2}}$} (r);
  \draw[msg] (w2) -- node[right, font=\tiny] {$\pv{4\sqrt{1-x^2}}$} (r);
\end{tikzpicture}
\caption{Compute}
\label{fig:step-compute}
\end{subcaptionblock}%
\hfill
\begin{subcaptionblock}{0.3\linewidth}
\centering
\begin{tikzpicture}[
  >=Stealth,
  role/.style={draw, rounded corners=3pt, fill=gray!10, minimum width=1cm, minimum height=0.45cm, font=\scriptsize, align=center},
  idle/.style={draw, rounded corners=3pt, fill=white, minimum width=1cm, minimum height=0.45cm, font=\scriptsize, align=center, text=black!30, draw=black!30},
  msg/.style={->, semithick},
  stopmsg/.style={->, semithick, dashed, red!70!black}
]
  \node[role] (m) at (0,1.6) {$\role{m}$};
  \node[role] (w1) at (-1.3,0.55) {$\role{w_1}$};
  \node[role] (w2) at (1.3,0.55) {$\role{w_2}$};
  \node[role] (r) at (0,-0.55) {$\role{r}$};
  \draw[msg] (r) -- node[sloped, above, font=\tiny] {\labname{cont}/\labname{stop}} (m);
  \draw[stopmsg] (m) -- node[left, font=\tiny] {\labname{stop}} (w1);
  \draw[stopmsg] (m) -- node[right, font=\tiny] {\labname{stop}} (w2);
\end{tikzpicture}
\caption{Control}
\label{fig:step-control}
\end{subcaptionblock}
\vspace{-2mm}
\caption{One round of the Monte Carlo $\pi$-estimation protocol.
  (a)~Manager sends random floats to workers.
  (b)~Workers compute $4\sqrt{1-x^2}$ and send results to reducer.
  (c)~Reducer sends \labname{cont} or \labname{stop} to manager; on \labname{stop}, manager terminates both workers (dashed).}
\Description{Three subfigures showing one round of the protocol: (a) manager sends floats to both workers, (b) workers send computed results to reducer, (c) reducer sends continue or stop to manager, and on stop the manager sends stop to both workers.}
\label{fig:protocol-steps}
\vspace{-5mm}
\end{figure}
To make the role of subsumption explicit, we let the manager be
\emph{slightly more permissive} than the protocol shown in
Figure~\ref{fig:protocol-steps}: besides \labname{cont} and
\labname{stop}, it is prepared to receive an extra label
\labname{crash} from the reducer. The reducer never emits
\labname{crash}, so this branch does not affect the actual behaviour of
the session, but it does appear in the inferred local type of the
manager.
A multiparty session $\M$ of a Monte Carlo protocol consists of
four processes defined as:\\[1mm]
\centerline{
$\small
\begin{array}{r@{~~}l}
\PP[\role{w_i}] =& \pfix{X} \procbra{\role{m}}{\labname{map} \gc
\procin{\role{m}}{x}{\procout{\role{r}}{\pv{4\sqrt{1-x^2}}}{\poc{X}}},\;
\labname{stop} \gc \inact} \qquad (i=1,2)\\[1mm]
\PP[\role{r}] = & \pfix{Y} \procin{\role{w_1}}{x}{\procin{\role{w_2}}{y}{\cond{g(x,y)}{\procsel{\role{m}}{\labname{cont}}\, \poc{Y}}{\procsel{\role{m}}{\labname{stop}}\, \inact}}}\\[1mm]
\PP[\role{m}] = & \pfix{X} \procsel{\role{w_1}}{\labname{map}}\,
\procout{\role{w_1}}{\pv{\mathtt{rand}()}}{\procsel{\role{w_2}}{\labname{map}}\,
  \procout{\role{w_2}}{\pv{\mathtt{rand}()}}{\procbra{\role{r}}{\labname{cont} \gc
      \poc{X},\; \labname{crash} \gc \inact,\; \labname{stop} \gc
      \procsel{\role{w_1}}{\labname{stop}}\,
      \procsel{\role{w_2}}{\labname{stop}}\, \inact}}}
\end{array}
$
}  
\\[1mm]
To read this example, it is enough to know that $\pfix{X}\PP$
denotes guarded recursion, $\procbra{\role{m}}{\ldots}$ branches on a
label received from~$\role{m}$, and
$\procsel{\role{q}}{\lab}\,\PP$ is the dual selection of label
$\lab$ at~$\role{q}$; inputs, outputs, and termination are written
$\procin{\role{m}}{x}{\PP}$, $\procout{\role{r}}{\e}{\PP}$, and
$\inact$.
Thus the reducer receives both worker results and chooses between
\labname{cont} and \labname{stop}, while the manager dispatches work
to the workers and follows the reducer's verdict.
Together these form the session
$\M = \proctag{\role{m}}{\PP[\role{m}]} \pc \proctag{\role{w_1}}{\PP[\role{w_1}]}
\pc \proctag{\role{w_2}}{\PP[\role{w_2}]} \pc \proctag{\role{r}}{\PP[\role{r}]}$.

A \emph{local type} abstracts a process to its communication
structure, replacing concrete values with base types: for instance,
the worker's receive $\role{m}\,\poc{?(}\pv{x}\poc{)}$ of a float
becomes input type $\tin[]{\role{m}}{\tfloat}$.
From $\M$, we can infer 
typing context $\ctx = \set{\ptag{\role{m}}{\T[\role{m}]}, \ptag{\role{w_1}}{\T[\role{w_1}]},
\ptag{\role{w_2}}{\T[\role{w_2}]}, \ptag{\role{r}}{\T[\role{r}]}}$
where:\\[1mm]
\centerline{
\small
$\begin{array}{r@{~~}l}
\T[\role{w_i}] =& \tfix{\ty} \tbra{\role{m}}{\labname{map} \gc
  \tin{\role{m}}{\tfloat} \tout{\role{r}}{\tfloat} \ty,\;
  \labname{stop} \gc \tend}
\qquad (i=1,2)
\\[2pt]
\T[\role{r}] = &
\tfix{\ty} \tin{\role{w_1}}{\tfloat}
\tin{\role{w_2}}{\tfloat}
\tsel{\role{m}}{\labname{cont} \gc \ty,\;
  \labname{stop} \gc \tend}
\\[2pt]
\T[\role{m}] = & \tfix{\ty} \tsel{\role{w_1}}{\labname{map} \gc \tout{\role{w_1}}{\tfloat} \tsel{\role{w_2}}{\labname{map} \gc \tout{\role{w_2}}{\tfloat}{}\\
& \quad \quad \quad \quad \quad \tbra{\role{r}}{\labname{cont} \gc \ty,\; \labname{crash} \gc \tend,\; \labname{stop} \gc \tsel{\role{w_1}}{\labname{stop} \gc \tsel{\role{w_2}}{\labname{stop} \gc \tend}}}}}
\end{array}$}\\[1mm]
$\tsel{}$/$\tbra{}$ denote the selection/branching types,
and we can check $\ctx$ is safe and live. 

The extra \labname{crash} branch is exactly where subsumption enters.
The reducer type $\T[\role{r}]$ can select only \labname{cont} or
\labname{stop}, so the intended global protocol need not mention
\labname{crash}. Let $\T[\role{m},\text{pro}]$ be the type obtained from
$\T[\role{m}]$ by deleting the \labname{crash} branch. Then
$\T[\role{m}] \subt \T[\role{m},\text{pro}]$, since a branching type
with more offered labels is a subtype. Writing
\[
\ctx[\text{pro}] =
\set{\ptag{\role{m}}{\T[\role{m},\text{pro}]},\,
\ptag{\role{w_1}}{\T[\role{w_1}]},\,
\ptag{\role{w_2}}{\T[\role{w_2}]},\,
\ptag{\role{r}}{\T[\role{r}]}}
\]
the original session $\M$ is typable by $\ctx[\text{pro}]$ via
subsumption.

\emph{Global types} represent  
the entire choreography from a bird's-eye view. 
We now build \emph{global type} $\G[\text{map}]$
whose projection is 
$\ctx[\text{pro}]$:\\[1mm]
\centerline{
$\small\begin{aligned}
\G[\text{map}] = \toc{\sfmu} \ty .\;&
\role{m}\gar\role{w_1}\;\goc{\Big\{}\labname{map} \gc
\Gvt{\role{m}}{\role{w_1}}{\tfloat}
\Gvt{\role{w_1}}{\role{r}}{\tfloat} {} \\
&
\role{m}\gar\role{w_2}\;\goc{\Big\{}\labname{map} \gc
\Gvt{\role{m}}{\role{w_2}}{\tfloat}
\Gvt{\role{w_2}}{\role{r}}{\tfloat} {} \\
&\quad
\role{r}\gar\role{m}\;\goc{\big\{}\labname{cont} \gc \ty,\;
\labname{stop} \gc {}
\; \GvtPairs{\role{m}}{\role{w_1}}{\labname{stop} \gc
\GvtPairs{\role{m}}{\role{w_2}}{\labname{stop} \gc \gend}}\goc{\big\}}\goc{\Big\}}\goc{\Big\}}
\end{aligned}$}\\[1mm]
where
$\role{p} \gar \role{q}\; \goc{\langle} \B \goc{\rangle}$ denotes
a message with type $\B$ from participant $\role{p}$ to participant $\role{q}$, 
and $\role{p} \gar \role{q}\; \goc{\{}\labname{l_i} \gc \ldots\goc{\}}$ a labelled choice.

Recall our aim is to show $\M$ is typable by the top-down system.
For this running example, the projected context is slightly smaller
than the inferred one: it omits the manager's unused
\labname{crash} branch. This is precisely why we need subsumption in the
top-down system. The global type captures only the actual
\labname{cont}/\labname{stop} choreography, while the process typing is
still discharged because $\ctx \subt \ctx[\text{pro}]$
(Figure~\ref{fig:topdown}).

\myparagraph{Summary.} 
Figure~\ref{fig:methodologies} gives the overall picture. 
The \emph{top-down} approach (Figure~\ref{fig:topdown}) starts with a global type $\G$. 
Then the projection generates local types $\T[1\,\text{pro}], 
\ldots, \T[n\,\text{pro}]$ for each participant, 
guaranteed to be safe and live.
We type processes $\PP[i]$ against  
$\T[i]$ where $\T[i] \, \subt \, \T[i\,\,\text{pro}]$ using the
subsumption rule.  
Then 
$\M=\proctag{\pp[1]}{\PP[1]}\,|\, \cdots \,|\, \proctag{\pp[n]}{\PP[n]}$ 
will be live (\emph{correctness by construction}).

In the \emph{bottom-up} approach (Figure~\ref{fig:bottomup}),
 we start with processes 
$\PP[1], \ldots, \PP[n]$ 
and infer their principal local types 
$\T[1\,\text{min}], \ldots, \T[n\,\text{min}]$ (by the type inference 
algorithm in \cite{thien-nobuko-popl-25}). 
We then check whether 
a set of local types satisfies safety and liveness.

The diagram in Figure~\ref{fig:topdown3} 
outlines our \emph{inference} algorithm, which can
construct a global type $\G[\text{inf}]$ from any safe-and-live
context.
Applying it to this running example, the inference starts from
$\ctx = \set{\ptag{\role{m}}{\T[\role{m}]},
\ptag{\role{w_1}}{\T[\role{w_1}]},
\ptag{\role{w_2}}{\T[\role{w_2}]},
\ptag{\role{r}}{\T[\role{r}]}}$.
It first follows the deterministic prefix
$\role{m}\gar\role{w_1}$, $\role{w_1}\gar\role{r}$,
$\role{m}\gar\role{w_2}$, $\role{w_2}\gar\role{r}$.
Note that some of these interactions involve disjoint participant
pairs (e.g.\ $\role{w_1}\gar\role{r}$ and $\role{m}\gar\role{w_2}$),
so the synthesis must choose an order for them.
As we shall see in \SEC{sec:semantics}, the global type semantics
allow independent interactions to commute, so this choice does
not affect the set of traces; the algorithm uses a fair schedule
to ensure every enabled pair is eventually selected.
Then it detects the choice $\role{r}\gar\role{m}$ between
\labname{cont} and \labname{stop}. The \labname{cont} branch
returns to the initial context; the \labname{stop} branch has $\role{m}$
select \labname{stop} for $\role{w_1}$ and $\role{w_2}$ and
terminates.  When the initial context is revisited, synthesis binds it
with recursion variable $\ty$, yielding $\G[\text{inf}]=\G[\text{map}]$.

If $\T[i]$ is a type of $\PP[i]$ as in 
Figure~\ref{fig:topdown}, then  
$\G[\text{inf}] \subseteq \G$ (i.e., $\G$ has at least all traces of $\G[\text{inf}]$). We call
$\G[\text{inf}]$ a \emph{principal global type} of the
typing context $\ctx=\set{\AT{\pp[1]}{\T[1]}, \ldots, \AT{\pp[n]}{\T[n]}}$. 

Figure~\ref{fig:bottomup4} shows
that type $\T[i\,\text{min}]$ inferred
from $\PP[i]$ is the minimum (principal) type of $\PP[i]$ \cite[Sec.~5.2]{thien-nobuko-popl-25}.
Hence a global type obtained by our inference 
is the \emph{principal (minimum) global type} which can type $\M$.
Since $\ctx$ is the principal typing context for $\M$,
$\G[\text{inf}]=\G[\text{map}]$ is a principal type of $\M$. 

To illustrate the role of subtyping in more detail, 
consider replacing the reducer with a variant
$\PPi[\role{r}]$ that stops after one round with its
principal type $\Ti[\role{r}\,\text{min}]$: 
\[
\small 
\PPi[\role{r}] =
\procin{\role{w_1}}{x_1}{\procin{\role{w_2}}{y_1}{\procsel{\role{m}}{\labname{stop}}\,
    \inact}}
\qquad\text{with}\quad 
\Ti[\role{r}\,\text{min}]
 = 
\tin{\role{w_1}}{\tfloat}
\tin{\role{w_2}}{\tfloat}
\tsel{\role{m}}{\labname{stop} \gc \tend}
\]
Then we have 
$\Ti[\role{r}\,\text{min}] \subt \T[\role{r}]$.
Therefore $\PPi[\role{r}]$ is typable by the original
context~$\ctx$ and by $\G$.
However, because $\PPi[\role{r}]$ has strictly fewer behaviours
than $\PP[\role{r}]$, the session $\Mi=
\proctag{\role{m}}{\PP[\role{m}]} \pc \proctag{\role{w_1}}{\PP[\role{w_1}]}
\pc \proctag{\role{w_2}}{\PP[\role{w_2}]} \pc \proctag{\role{r}}{\PPi[\role{r}]}$ has fewer traces,
yielding a smaller principal global type $\G[\text{min}]$ such that
($\subseteq$ here denotes trace inclusion; see Definition~\ref{def:traces})
$\G[\text{min}] \subseteq \G[\text{inf}] \subseteq \G$:\\[1mm]
\centerline{
$\small\begin{aligned}
\G[\text{min}] = \;&
\role{m}\gar\role{w_1}\;\goc{\Big\{}\labname{map} \gc
\Gvt{\role{m}}{\role{w_1}}{\tfloat}{}
\role{m}\gar\role{w_2}\;\goc{\Big\{}\labname{map} \gc
\Gvt{\role{m}}{\role{w_2}}{\tfloat} {}\\
&\quad
\Gvt{\role{w_1}}{\role{r}}{\tfloat} \;
\Gvt{\role{w_2}}{\role{r}}{\tfloat}
\role{r}\gar\role{m}\;\goc{\big\{}
\labname{stop} \gc {}
\GvtPairs{\role{m}}{\role{w_1}}{\labname{stop} \gc
\GvtPairs{\role{m}}{\role{w_2}}{\labname{stop} \gc \gend}}\goc{\big\}}\goc{\Big\}}\goc{\Big\}}
\end{aligned}$}

\smallskip\noindent
The remainder of the paper formalises this development.
We introduce the session calculus and its properties
(\SEC{sec:calculus}), then the type language and projection
(\SEC{sec:types}), followed by the operational semantics of types
(\SEC{sec:semantics}).
The core technical contribution is the synthesis algorithm and its
correctness proof (\SEC{sec:building}), which together yield the
completeness result (\SEC{sec:typing}).

\section{Multiparty Session Calculus}
\label{sec:calculus}
This section introduces the multiparty synchronous session calculus
from \cite{synchronous-subtyping}, which suffices for our purposes.

\myparagraph{Syntax.}
\emph{Expressions} ($\e,\ei,\ldots$) are variables ($\pv{x},\pv{y},\pv{z},\ldots$), a value
$\val$ (either a natural $\valnat$, an integer $\valint$, a float $\pv{f}$, or a boolean
$\true/\false$), or are built by operators such as $\pv{\neg}$, $\pv{\sqrt{~}}$, $\pv{\oplus}$,
$\pv{\vee}$, $\pv{+}$, $\pv{-}$, $\pv{*}$, $\pv{/}$, or $\pv{\le}$.
The nullary $\pv{\mathtt{rand}()}$ returns a random float in $[0,1]$.
The operator $\pv{\oplus}$ models non-determinism: $\e[1] \mathbin{\pv{\oplus}} \e[2]$ may
yield either $\e[1]$ or $\e[2]$.
\[
  \e \;::=\; \true \SEP \false \SEP \valnat \SEP \valint \SEP \pv{f} \SEP \pv{x}
  \SEP \pv{\mathtt{rand}()}
  \SEP \e \pv{\vee} \e \SEP \pv{\neg} \e \SEP \pv{\sqrt{\e}} \SEP \e \pv{+} \e \SEP \e \pv{-} \e \SEP \e \pv{*} \e \SEP \e \pv{/} \e \SEP \e \mathbin{\pv{\oplus}} \e \SEP \e \pv{\le} \e
\]
\emph{Processes} ($\PP,\PQ,\ldots$) and \emph{multiparty sessions}
($\M,\Mi,\ldots$) are:
\[
\begin{array}{r@{\;}l@{\quad}l}
\PP ::= & \inact
  \SEP \procout{\pp}{\e}{\PP}
  \SEP \procin{\pp}{x}{\PP}
  \SEP \procsel{\pp}{\lab}\,\PP  
  \SEP \procbrasub{\pp}{\,\lab[i]:\PP[i]\,}{i\in I}
\SEP  \cond{\e}{\PP}{\PQ}
  \SEP \pfix{X}\PP
  \SEP \poc{X}
  &
\\[2pt]
\M ::= & \proctag{\pp}{\PP} \SEP \M \pc \Mi
\end{array}
\]
$\pp,\pq,\pr,\role{s},\ldots \in \partset = \set{\pp[0], \ldots, \pp[n-1]}$ are \emph{participants};
$\lab,\labi,\ldots \in \labset$ are \emph{labels}.
Nil $\inact$ denotes termination; 
output $\procout{\pp}{\e}{\PQ}$ sends the value of $\e$ to $\pp$;
input $\procin{\pp}{x}{\PQ}$ receives from $\pp$;
selection $\procsel{\pp}{\lab}\,\PQ$ selects $\lab$ at $\pp$ (internal choice);
branching $\procbrasub{\pp}{\lab[i]:\PP[i]}{i\in I}$ ($\size{I}\geq 1$) accepts any $\lab[i]$ from $\pp$ with
continuation $\PP[i]$ (external choice).
Conditionals $\cond{\e}{\PP}{\PQ}$ and recursion $\pfix{X}\PP$ are standard.
Recursion is \emph{guarded} (e.g.\ $\pfix{X}\poc{X}$ is invalid, but
$\pfix{X}\procout{\pp}{\e}{\poc{X}}$ is valid).
A process $\PP$ playing role $\pp$ is written $\proctag{\pp}{\PP}$.
We write
$\prod_{i\in I}\proctag{\pp[i]}{\PP[i]}$ for the composition of processes into a multiparty session.

\myparagraph{Reductions.}
The value $\val$ of $\e$ (notation $\eval{\e}{\val}$) is computed by
a standard big-step semantics. The full rules are given in
\appref{app:expr-eval}.
The internal choice $\e[1] \mathbin{\pv{\oplus}} \e[2]$ evaluates to either branch (see \cite{synchronous-subtyping}, Table~1).
The \emph{reduction rules of multiparty sessions} are in
Figure~\ref{fig:reduction_sessions}.

\begin{figure}
  \begin{center}
\small
$\begin{array}{c}
\begin{array}{@{}rll}
\rulename{r-comm} &
\proctag{\pp}{\procin{\pq}{x}{\PP}} \pc \proctag{\pq}{\procout{\pp}{\e}{\PQ}}  \; \pc\;\M \red
\proctag{\pp}{{\PP}\sub{\val}{\pv{x}}}\; \pc \;\proctag{\pq}{\PQ} \; \pc\;\M
& (\eval{\e}{\val})\\
\rulename{r-bra} &
\proctag{\pp}{\procsel{\pq}{\lab[k]}{\PP}} \pc \proctag{\pq}{\procbrasub{\pp}{\lab[i]: \PP[i]}{i\in I}}
\; \pc\;\M
\red
\proctag{\pp}{\PP} \pc \proctag{\pq}{\PP[k]}
\; \pc\;\M
& (k\in I)\\
\rulename{t-cond}
&{
    \proctag{\pp}{\cond{\e}{\PP}{\PQ}} \; \pc \;  \M \red \proctag{\pp}{\PP}\; \pc \;  \M
   }
& (\eval{\e}{\true}) \\
\rulename{f-cond}
&{
    \proctag{\pp}{\cond{\e}{\PP}{\PQ}} \; \pc \;  \M \red \proctag{\pp}{\PQ}\; \pc \;  \M
   }
& (\eval{\e}{\false}) \\
\rulename{r-str}
&
\M[1]\prestruct \M[2] \quad \M[2]\red \M[3] \quad \M[3] \prestruct \M[4]
 \ \Longrightarrow \
\M[1] \red \M[4]\\
\rulename{v-err}
&
\proctag{\pp}{\cond{\e}{\PP}{\PQ}} \; \pc \; \M \red \error
& (\eval{\e}{\val} \text{ and } \val\not\in \{\true,\false\}) \\
\rulename{c-err}
&
\proctag{\pp}{\procsel{\pq}{\lab[k]}{\PP}} \pc \proctag{\pq}{\procbrasub{\pp}{\lab[i]: \PP[i]}{i\in I}}
\; \pc\;\M
\red \error & (k \notin I)\\
\end{array}
\end{array}
$
\vspace{-3mm}
\end{center}
  \caption{Reductions of sessions}
\vspace{-3mm}  
\label{fig:reduction_sessions}
\end{figure}

Rule \rulename{r-comm} is communication between an output and an input;
\rulename{r-bra} is selection/branching. Rules \rulename{t-cond,f-cond} are standard, and
\rulename{r-str} closes reductions under \emph{structural precongruence}
$\prestruct$, defined by:\\[1mm]
$\small
\proctag{\pp}{\pfix{X}\PP}\pc \M \;\prestruct\; \proctag{\pp}{\PP\sub{\pfix{X}\PP}{\poc{X}}}\pc \M
\qquad
\prod_{i\in I}\proctag{\pp[i]}{\PP[i]} \;\prestruct\;
\prod_{j\in J}\proctag{\pp[j]}{\PP[j]}
\ \text{ ($I$ is a permutation of $J$)}
$\\[1mm]
\rulename{v-err,c-err} define value and session
errors.
$\red^\ast$ denotes the reflexive–transitive closure of $\red$.

\myparagraph{Properties.}
We recall properties of a session $\M$ from \cite[Def.~5.1]{Scalas2019}.

\begin{definition}[Properties]\label{def:properties}
A session $\M$ is:
\begin{enumerate}
  \item\label{mp:safe}
  \emph{communication safe} iff $\M \red^\ast \Mi$ implies $\Mi \not\prestruct \error$.
  \item\label{mp:live}
  \emph{live} iff $\M \red^\ast \Mi = \Pi_{i\in I} \proctag{\pp[i]}{\PP[i]}$ implies:
  \begin{enumerate}
    \item\label{mp:live:in}
    If $\PP[i]=\procin{\pp[j]}{x}{\PPi[i]}$ then $\exists \M[2],\val.\ \Mi \red^\ast \M[2] \pc \proctag{\pp[i]}{\PPi[i]\sub{\val}{\pv{x}}}$.
    \item\label{mp:live:out}
    If $\PP[i]=\procout{\pp[j]}{\e}{\PPi[i]}$ then $\exists \M[2].\ \Mi \red^\ast \M[2] \pc \proctag{\pp[i]}{\PPi[i]}$.
    \item\label{mp:live:sel}
    If $\PP[i]=\procsel{\pp[j]}{\lab}\,\PPi[i]$ then $\exists \M[2].\ \Mi \red^\ast \M[2] \pc \proctag{\pp[i]}{\PPi[i]}$.
    \item\label{mp:live:bra}
    If $\PP[i]=\procbrasub{\pp[j]}{\lab[k]:\PPi[k]}{k\in K}$ then
    $\exists
    \M[2],h\in K.\ \Mi \red^\ast \M[2] \pc \proctag{\pp[i]}{\PPi[h]}$.
\end{enumerate}
\end{enumerate}
\end{definition}

\noindent
Thus $\M$ is \emph{communication safe} if no errors arise,
and \emph{live} if every pending action eventually communicates.
Since these properties are undecidable on processes,
the bottom-up system checks the corresponding properties on \emph{typing
contexts} (\SEC{sec:semantics}), which soundly imply the process-level guarantees.
\begin{remark}[Process liveness]\label{rem:liveness}
Liveness for branching (\Cref{def:properties} (\ref{mp:live:bra}))
requires that \emph{some} branch $h\in K$ is eventually taken. Another branch
may never fire. Consider:\\[1mm]
\centerline{\small
$\begin{array}{rl}
\M[1] \ &=\ \proctag{\pp}{\pfix{X} \procbra{\pq}{\lab[1]: \poc{X},\ \lab[2]:\inact}}
          \pc \proctag{\pq}{\pfix{X} \procsel{\pp}{\lab[1]}\,\poc{X}}
\\
\M[2] \ &=\ \proctag{\pp}{\pfix{X} \procbra{\pq}{\lab[1]: \poc{X},\ \lab[2]:\procout{\pr}{\e}{\inact}}}
          \pc \proctag{\pq}{\pfix{X} \procsel{\pp}{\lab[1]}\,\poc{X}}
          \pc \proctag{\pr}{\procin{\pp}{x}{\inact}}
  \end{array}
$}\\[1mm]  
$\M[1]$ is live (always takes $\lab[1]$), but $\M[2]$ is not live: the $\lab[2]$ branch
requires $\pp$ and $\pr$ to interact, which never happens while $\pp$
and $\pq$ loop on $\lab[1]$.
\end{remark}

\section{Global and Local Types}
\label{sec:types}
This section introduces global and local types and coinductive full-merging projection.
\subsection{Syntax of Global and Local Types and Subtyping}
\label{sec:global-local-types}
\textit{Global types} describe interactions between all participants.

\begin{definition}[Global types]
\label{def:global}
\emph{Base types} ($\B, \B',
\dots$) and
\emph{global types} ($\G, \G', \dots$)
are defined by:\\[1mm]
\centerline{$
\begin{array}{c}
\B ::= \tbool \SEP \tnat \SEP \tint \SEP \tfloat\qquad
\G  ::= \tend \SEP \Gvt{\pp}{\pq}{\B}\Gi \SEP
\GvtPair{\pp}{\pq}{\lab[i]:\G[i]}{i\in I} \SEP
\toc{\sfmu}\ty. \G \SEP \ty
\end{array}
$}
\end{definition}
\noindent
The constructor $\Gvt{\pp}{\pq}{\B}\G$ denotes a payload communication
from $\pp$ to $\pq$, while
$\GvtPair{\pp}{\pq}{\lab[i]: \G[i]}{i \in I}$ denotes a labelled choice
where $\pp$ selects at $\pq$.
Recursion is guarded, and index sets $I$ are finite and nonempty.
The set of \emph{participants} of a global type $\G$ is defined as:
$\pt{\ty} = \pt{\tend} = \emptyset$; 
$\pt{\toc{\sfmu} \ty. \G} = \pt{\G}$; 
$\pt{\Gvt{\pp}{\pq}{\B}\G} = \pt{\G}\cup\{\pp, \pq\}$; 
$\pt{\GvtPair{\pp}{\pq}{\lab[i]: \G[i]}{i \in I}} = \left(\bigcup_{i \in
  I}{\pt{\G[i]}}\right) \cup \{\pp, \pq\}$.
We denote $\ftv{\G}$ as the set of \emph{free type variables} in $\G$ 
and if $\ftv{\G}=\emptyset$, we call $\G$ \emph{closed}. 
We define $\unfold{\G} \;=\;
\unfold{\Gi \sub{\toc{\sfmu} \ty.\Gi}{\ty}}$ if 
$\G = \toc{\sfmu} \ty.\Gi$; otherwise 
$\unfold{\G} \;=\; \G$. 

\emph{Local types} specify individual participant behaviour;
\emph{typing contexts} abstract process behaviour.

\begin{definition}[Local types and typing contexts] 
\label{def:local} \emph{Local types} ($\T, \T', \dots$) and typing
contexts ($\ctx,\ctxi,\dots$) are defined by:\\[1mm] 
\centerline{$
\begin{array}{c}
\T ::=  \tend \SEP \tout \pp{\B} \T \SEP \tin \pp{\B} \T 
         \SEP \tselsub{\pp}{\lab[i]: \T[i]}{i\in I}
         \SEP \tbrasub{\pp}{\lab[i]: \T[i]}{i\in I}
         \SEP \toc{\sfmu} \ty. \T 
         \SEP \ty \qquad 
\ctx ::= \ctx,\ctxi \SEP \ptag{\pp}{\T}         
\end{array}
$}
\end{definition}
Local types are endpoint views of protocols:
$\tout \pp{\B} \T$ and $\tin \pp{\B} \T$ are output and input,
while $\tselsub{\pp}{\lab[i]: \T[i]}{i\in I}$ and
$\tbrasub{\pp}{\lab[i]: \T[i]}{i\in I}$ are the dual
selection/branching behaviours.
Recursion is guarded, index sets are finite and nonempty, and we
often omit trailing $\tend$.
We use $\ftv{\T}$ for the set of \emph{free type variables} in $\T$,
and  $\unfold{\T}$ for unfolding. 
Typing contexts ($\ctx$, $\ctxi$, \dots) are mappings from participants
$\pp[0],\ldots,\pp[n-1]$ to local types, and are nonempty.  

\begin{definition}[Subtyping] The \emph{synchronous subtyping relation} is
 coinductively defined as:
\label{def:coind-subtyping}

\centerline{\small$
\begin{array}{c}
    \cinferruleleft[\rulename{S1}]{ \T \subt \Ti }{ \tout{\pp}{\B}{\T} \subt \tout{\pp}{\B}{\Ti}} \quad
    \cinferruleleft[\rulename{S2}]{ \T \subt \Ti }{ \tin{\pp}{\B}{\T} \subt \tin{\pp}{\B}{\Ti}} \quad
    \cinferruleleft[\rulename{S5}]{ \T \sub{\toc{\sfmu} \ty . \T}{\ty} \subt \Ti }{ \toc{\sfmu} \ty . \T \subt \Ti} \quad
   \cinferruleleft[\rulename{S6}]{ \T \subt \Ti \sub{\toc{\sfmu} \ty
        . \Ti}{\ty} }{ \T \subt \toc{\sfmu} \ty . \Ti}\\[2mm]
\cinferruleleft[\rulename{S3}]{ I \subseteq J \\ \forall i \in I : \T[i] \subt \Ti[i]}{ \tselsub{\pp}{\lab[i] : \T[i]}{i \in I} \subt \tselsub{\pp}{\lab[i] : \Ti[i]}{i \in J}} \quad
    \cinferruleleft[\rulename{S4}]{ J \subseteq I \\ \forall i \in J : \T[i] \subt \Ti[i]}{ \tbrasub{\pp}{\lab[i] : \T[i]}{i \in I} \subt \tbrasub{\pp}{\lab[i] : \Ti[i]}{i \in J}} \quad
    \cinferruleleft[\rulename{S7}]{ }{ \tend \subt \tend } 
\end{array}
$}
\end{definition}
Rule~\rulename{S3} is covariant: subtypes offer \emph{fewer} selections ($I \subseteq J$).
Rule~\rulename{S4} is contravariant: subtypes offer \emph{more} branches ($J \subseteq I$).
We extend $\subt$ to typing contexts as:
$\ctx \subt \ctxi$ if $\dom{\ctx}=\dom{\ctxi}$ and $\forall \pp\in
\dom{\ctx}, \ctx(\pp)\subt \ctxi(\pp)$.  

\subsection{Merging and Projections}
\label{sec:merge-and-projection}
A participant uninvolved in a choice must anticipate different continuations.
The \emph{endpoint projection} (EPP) merges these into a single compatible local type.
We use \emph{coinductive full merge}, which creates a local type that is a
\emph{subtype} of each continuation.

The coinductive full merge ($\comergesym$)
\emph{merges} local types by combining continuations in
the branching and recursively merging continuations on shared labels.

\begin{definition}[Coinductive merge relation]
\label{def:coind-merge}
The coinductive merge relation between a set of local types
  and a local type is defined by the following rules. We use the
  coinductive full merge relation ($\comergesym$), where the index sets
  $I_j$ in rule \rulename{M5} may differ.\\[1mm]
\centerline{
  \small
$
\begin{array}{c}
   \cinferruleleft[\rulename{M1}]{\quad}{\comergefull{\set{\tend}}{\tend}} \quad
    \cinferruleleft[\rulename{M2}]{ \comergefull{\set{\T[j]}_{j \in J}}{\T}}{\comergefull{\set{\tout{\pp}{\B}{\T[j]}}_{j \in J}}{\tout{\pp}{\B}{\T}}}\quad
    \cinferruleleft[\rulename{M3}]{ \comergefull{\set{\T[j]}_{j \in J}}{\T}}{\comergefull{\set{\tin{\pp}{\B}{\T[j]}}_{j \in J}}{\tin{\pp}{\B}{\T}}}\\[3mm]
    \cinferruleleft[\rulename{M4}]{ \forall i \in I :
    \comergefull{\set{\T[i,j]}_{j \in J}}{\T[i]}}{\comergefull{\set{\tselsub{\pp}{\lab[i] : \T[i, j]}{i \in I}}_{j \in J}}{\tselsub{\pp}{\lab[i] : \T[i]}{i \in I}}} \quad
    \cinferruleleft[\rulename{M5}]{ \forall k \in J,\; \forall i \in I_{k} : \comergefull{\set{\T[i,j]}_{j \in J \mid i \in I_{j}}}{\T[i]}}{\comergefull{\set{\tbrasub{\pp}{\lab[i] : \T[i, j]}{i \in I_{j}}}_{j \in J}}{\tbrasub{\pp}{\lab[i] : \T[i]}{i \in {\bigcup}_{j \in J} I_{j}}}}\\[3mm]
   \cinferruleleft[\rulename{M6}]{ \comergefull{S \cup \set{\T \sub{\toc{\sfmu} \ty . \T}{\ty}}}{\Ti}}{ \comergefull{S \cup \set{\toc{\sfmu} \ty . \T}}{\Ti}} \quad
   \cinferruleleft[\rulename{M7}]{ \comergefull{S}{\T \sub{\toc{\sfmu} \ty . \T}{\ty}}}{ \comergefull{S}{\toc{\sfmu} \ty . \T}}
\end{array}$
}
\end{definition}
Rule~\rulename{M4} requires identical labels across selections;
\rulename{M5} allows their union.
Full merge lets uninvolved participants offer additional choices
while preserving observable behaviour.

Using the coinductive merge, we define a \emph{coinductive projection}
from a global type to a local type, parameterised by the projected participant~$\pp$.
To see why both \emph{full} and \emph{coinductive} merging are needed,
consider projecting $\G[\text{map}]$ onto~$\role{w_1}$.
At the choice $\role{r}\gar\role{m}\;\goc{\{}\labname{cont}\gc\ldots,\;\labname{stop}\gc\ldots\goc{\}}$,
$\role{w_1}$ is uninvolved, so rule~\rulename{PR6} projects each branch independently and merges.
The $\labname{cont}$ branch recurses, giving label set $\set{\labname{map},\labname{stop}}$;
the $\labname{stop}$ branch gives only $\set{\labname{stop}}$.
A standard merge requires identical label sets and would reject this.
Full merge (\rulename{M5}) takes the union, and coinductive unfolding (\rulename{M6})
handles the $\mu$-type, yielding $\T[\role{w_1}]$ itself.

\begin{definition}[Coinductive projection]
\label{def:coind-proj}
The coinductive projection from a global type to a local type is defined by the following rules. 
When the merge used in rule \rulename{PR6} is the full merge
($\comergesym$), we obtain the projection relation
\emph{coinductive full-merging projection} ($\cofullproj{p}$), which is the one
used in this paper.
\centerline{
\small
$
\begin{array}{c}
    \cinferruleleft[\rulename{PR1}]{ \G \cofullproj{p} \T}{ \Gvt{\pp}{\pq}{\B}{\G}  \cofullproj{p} \tout{\pq}{\B}{\T}} \quad 
    \cinferruleleft[\rulename{PR2}]{ \G \cofullproj{p} \T}{ \Gvt{\pq}{\pp}{\B}{\G}  \cofullproj{p} \tin{\pq}{\B}{\T}} \quad 
    \cinferruleleft[\rulename{PR3}]{ \G \cofullproj{p} \T \\ \pp \notin \set{\pq, \pr} }{ \Gvt{\pq}{\pr}{\B}{\G}  \cofullproj{p} \T} \\[3mm]  
   \cinferruleleft[\rulename{PR4}]{ \forall i \in I : \G[i] \cofullproj{p} \T[i]}{ \GvtPair{\pp}{\pq}{\lab[i] : \G[i]}{i \in I}  \cofullproj{p} \tselsub{\pq}{\lab[i] : \T[i]}{i \in I}} \quad 
    \cinferruleleft[\rulename{PR5}]{ \forall i \in I : \G[i] \cofullproj{p} \T[i]}{ \GvtPair{\pq}{\pp}{\lab[i] : \G[i]}{i \in I}  \cofullproj{p} \tbrasub{\pq}{\lab[i] : \T[i]}{i \in I}} \\[3mm]  
    \cinferruleleft[\rulename{PR6}]{ \forall i \in I : \G[i] \cofullproj{p} \T[i] \\ \comergefull{\set{\T[i]}_{i \in I}}{\T} \\ \pp \notin \set{\pq, \pr}}{ \GvtPair{\pq}{\pr}{\lab[i] : \G[i]}{i \in I}  \cofullproj{p} \T} \quad
    \cinferruleleft[\rulename{PR9}]{ \G \sub{\toc{\sfmu} \ty . \G}{\ty} \cofullproj{p} \T \\ \pp \in \pt{\G}}{ \toc{\sfmu} \ty . \G \cofullproj{p} \T} \\[3mm] 
    \cinferruleleft[\rulename{PR10}]{ \G \cofullproj{p} \T \sub{\toc{\sfmu} \ty . \T}{\ty}}{ \G \cofullproj{p} \toc{\sfmu} \ty . \T} \quad
    \cinferruleleft[\rulename{PR11}]{ }{ \gend \cofullproj{p} \tend} \quad
   \cinferruleleft[\rulename{PR12}]{ \pp \not \in \pt{\G} }{ \toc{\sfmu} \ty . \G \cofullproj{p} \tend}
\end{array}
$
}
\end{definition}
\noindent
These rules project communications to the sender or receiver, skip
uninvolved participants, merge uninvolved branches, and handle
recursion coinductively; rule~\rulename{PR9} requires $\pp$ to occur in
the recursive body, while rule~\rulename{PR12} projects to $\tend$
when $\pp$ is absent from it.
To ensure projections are live, we need to know which participants are
\emph{unavoidable}, i.e., must eventually be involved
after some bounded number of steps in any fair path.

\begin{definition}[Subterms and unavoidable set]
  \label{def:subterms-unavoidable}
  The \emph{subterms} of a global type $\G$ are the set of (unfolded) global types that appear in $\G$. 
  The \emph{unavoidable set} of $\G$ is the set of participants that must eventually be involved in any (fair) path from $\G$.
  Let $\Gi = \unfold{\G}$. Then $\subterms{\G}$ and $\unavoidable(\G)$ are defined as the least sets satisfying the following:
  \[
  \begin{array}{l|l|l}
    \Gi & \subterms{\G} & \unavoidable(\G) \\
    \hline
    \gend & \set{\gend} & \emptyset \\
    \Gvt{\pp}{\pq}{\B}{\Gii} & \set{\Gi} \cup \subterms{\Gii} & \set{\pp, \pq} \cup \unavoidable(\Gii) \\
    \GvtPair{\pp}{\pq}{\lab[i] : \G[i]}{i \in I} & \set{\Gi} \cup \bigcup_{i \in I} \subterms{\G[i]} & \set{\pp, \pq} \cup \bigcap_{i \in I} \unavoidable(\G[i]) \\
  \end{array}
  \]
\end{definition}
We can now define a global type as \emph{balanced} if no reachable
participant can be avoided.

\begin{definition}[Balanced, Definition~3.3 in \cite{synchronous-subtyping}]
\label{def:balanced-global-types}
A global type $\G$ is \emph{balanced} if, for all $\Gi \in \subterms{\G}$, we have that $\unavoidable(\Gi) = \pt{\Gi}$.
\end{definition}

Without balancedness, the top-down system does not guarantee liveness.
To illustrate, consider a variant of our Monte Carlo example 
in \SEC{sec:overview} where workers send results
back to the manager instead of the reducer, and only on termination
does the manager forward the aggregated result to the reducer
(Figure~\ref{fig:unbalanced-example}). 

\begin{figure}[tb]
\begin{minipage}[c]{0.38\linewidth}
\centering
\begin{tikzpicture}[
  >=Stealth,
  role/.style={draw, rounded corners=3pt, fill=gray!10, minimum width=1cm, minimum height=0.45cm, font=\scriptsize, align=center},
  msg/.style={->, semithick},
  stopmsg/.style={->, semithick, dashed, red!70!black}
]
  \node[role] (m) at (0,2.2) {$\role{m}$};
  \node[role] (w1) at (-2.0,0.35) {$\role{w_1}$};
  \node[role] (w2) at (2.0,0.35) {$\role{w_2}$};
  \node[role] (r) at (0,0.35) {$\role{r}$};

  \draw[msg] (m.210) to[bend right=12] node[left, font=\tiny, pos=0.25] {$\pv{x}$} (w1.70);
  \draw[msg] (m.330) to[bend left=12] node[right, font=\tiny, pos=0.25] {$\pv{x}$} (w2.110);
  \draw[msg] (w1.50) to[bend right=12] node[right, font=\tiny, pos=0.25] {$\pv{4\sqrt{1\!-\!x^2}}$} (m.230);
  \draw[msg] (w2.130) to[bend left=12] node[left, font=\tiny, pos=0.25] {$\pv{4\sqrt{1\!-\!x^2}}$} (m.310);
  \draw[stopmsg] (m) -- node[right, font=\tiny, pos=0.35] {$\pv{\bar{z}}$} (r);
  \draw[stopmsg] (m.west) to[bend right=40] node[above, font=\tiny] {\labname{stop}} (w1.north west);
  \draw[stopmsg] (m.east) to[bend left=40] node[above, font=\tiny] {\labname{stop}} (w2.north east);
\end{tikzpicture}
\end{minipage}%
\hfill
\begin{minipage}[c]{0.59\linewidth}
\footnotesize
\[
\begin{aligned}
\G[\text{unbal}] = \toc{\sfmu} \ty .\;&
\role{m}\gar\role{w_1}\;\goc{\Big\{}\labname{map} \gc
\Gvt{\role{m}}{\role{w_1}}{\tfloat} {} \\
&
\role{m}\gar\role{w_2}\;\goc{\Big\{}\labname{map} \gc
\Gvt{\role{m}}{\role{w_2}}{\tfloat} {} \\
&\quad
\Gvt{\role{w_1}}{\role{m}}{\tfloat} \;
\Gvt{\role{w_2}}{\role{m}}{\tfloat} \\
&\quad
\role{m}\gar\role{w_1}\;\goc{\big\{}
\labname{cont} \gc \GvtPairs{\role{m}}{\role{w_2}}{\labname{cont} \gc \ty},\; {} \\
&\qquad
\labname{stop} \gc
\GvtPairs{\role{m}}{\role{w_2}}{\labname{stop} \gc
\Gvt{\role{m}}{\role{r}}{\tfloat} \gend}\goc{\big\}}\goc{\Big\}}\goc{\Big\}}
\end{aligned}
\]
\end{minipage}
\caption{Unbalanced variant: workers send results to the manager, not the reducer. Dashed arrows indicate messages that only occur on termination.}
\label{fig:unbalanced-example}
\end{figure}
Let $\Gi$ be the subterm at the choice
$\role{m} \gar \role{w_1}\; \goc{\{}\labname{cont}\gc\ldots,\, \labname{stop}\gc\ldots\goc{\}}$.
The reducer $\role{r}$ is reachable from $\Gi$ (via the $\labname{stop}$ branch),
so $\role{r} \in \pt{\Gi}$.
However, the $\labname{cont}$ branch loops back without ever involving $\role{r}$,
so $\role{r} \notin \unavoidable(\Gi)$ (the unavoidable set uses
$\bigcap$ over branches, and $\role{r}$ is absent from the $\labname{cont}$ branch).
Since $\pt{\Gi} \neq \unavoidable(\Gi)$, $\G[\text{unbal}]$ is \emph{not balanced}.
Concretely, the projections of $\G[\text{unbal}]$ can type a session
in which the manager and workers loop indefinitely via the
$\labname{cont}$ branch, never reaching the $\labname{stop}$ branch
that communicates with $\role{r}$, so the reducer starves and liveness fails.
In balanced $\G[\text{map}]$ in \SEC{sec:overview},
workers send to the reducer inside the loop, ensuring every participant
is unavoidable at every subterm.
See also \cite[Example~3.3]{Franco2024}
and \cite[Example 3.9]{HYC2016}
for unbalanced types. 

\section{Labelled Transition Semantics}
\label{sec:semantics}
As a \textit{behavioural} type system, the framework equips global and local types with
operational semantics, allowing them to evolve during the execution of the
session.

\begin{definition}[Local type and typing context semantics]
\label{def:typecontexts}
\label{def:semantics}
We define a transition relation $\transa{\czeta}$
for local types, $\transt{\calpha}$ for one-sided transitions of contexts,
and $\transi{\inter}$ for two-sided synchronous transitions in contexts,
where an observable $\obs$ is either a payload from $\B$ or a label $\lab$,
an action $\czeta$ is a send or receive of an observable,
a tagged action $\calpha = \pp:\czeta$ is accompanied by the performing participant,
and an interaction $\inter = \interaction{\pp}{\pq}{\obs}$ is an action between two participants.
\[
\obs \defeq \B \bnfor \lab
\qquad
\czeta \defeq \qsend{\pp}{\obs} \bnfor \qrecv{\pp}{\obs}
\qquad
\calpha \defeq \pp:\czeta
\qquad
\inter \defeq \interaction{\pp}{\pq}{\obs}
\]
\small
\[
\begin{array}{c}
\inferrule*[left=\rulename{LR1}]{ }{ \tout{\pp}{\B}{\T}  \transa{\qsend{\pp}{\B}} \T} \quad
    \inferrule*[left=\rulename{LR2}]{ }{ \tin{\pp}{\B}{\T}  \transa{\qrecv{\pp}{\B}} \T} \quad
    \inferrule*[left=\rulename{LR3}]{j \in I}{ \tselsub{\pp}{\lab[i] : \T[i]}{i \in I}  \transa{\qsend{\pp}{\lab[j]}} \T[j]} \quad
    \inferrule*[left=\rulename{LR4}]{j \in I}{ \tbrasub{\pp}{\lab[i] : \T[i]}{i \in I}  \transa{\qrecv{\pp}{\lab[j]}} \T[j]} \\[3mm]
    \inferrule*[left=\rulename{LR5}]{\T \sub{\toc{\sfmu} \ty . \T}{\ty} \transa{\czeta} \Ti}{\toc{\sfmu} \ty . \T \transa{\czeta} \Ti} \quad
    \inferrule*[left=\rulename{LTag}]{\T \transa{\czeta} \Ti}{\ctx, \ptag{\pp}{\T} \transt{\pp : \czeta} \ctx, \ptag{\pp}{\Ti}} \quad
    \inferrule*[left=\rulename{LEnv}]{\ctx[1] \transt{\pp : \qsend{\pq}{\obs}} \ctxi[1] \\ \ctx[2] \transt{\pq : \qrecv{\pp}{\obs}} \ctxi[2]}{\ctx[1], \ctx[2] \transi{\interaction{\pp}{\pq}{\obs}} \ctxi[1], \ctxi[2]}
\end{array}
\]
\end{definition}
Rule \rulename{LR1} permits the sending of a payload of sort $\B$ to a
participant $\pp$, and rule \rulename{LR2} is 
its dual for receiving. 
Rule \rulename{LR3} sends label $\lab[j]$ to $\pp$,
continuing according to the continuation $\T[j]$. 
Rule \rulename{LR4} is the dual of \rulename{LR3}.
Rule \rulename{LR5} unfolds a recursive type.
The one-sided transition (rule \rulename{LTag})
tags the action with the participant
that performs it. 
The relation $\transi{\interaction{\pp}{\pq}{\obs}}$ captures
two-sided synchronous reductions in contexts by combining
the one-sided transitions (\rulename{LEnv}).%

For example, if
$\unfold{\T[\role{m}]} = \tout{\role{w_1}}{\tfloat} \, \Ti$, then
the transition is: $\T[\role{m}] \transa{\qsend{\role{w_1}}{\tfloat}} \Ti
$. 

\begin{definition}[Global type semantics]
\label{def:semantics-global}
The labelled transition relation of global types, written
$\transg{\inter}$, is defined by the following rules.  
\small  
\[
\begin{array}{c}
    \inferrule*[left=\rulename{GR1}]{ }{ \Gvt{\pp}{\pq}{\B} \G \;
      \transg{\interaction{\pp}{\pq}{\B}} \; \G} \quad 
    \inferrule*[left=\rulename{GR2}]{ \G \;
      \transg{\interaction{\pr}{\role{s}}{\obs}} \; \Gi \\ \set{\pr,
        \role{s}} \cap \set{\pp, \pq} = \emptyset}{ \Gvt{\pp}{\pq}{\B}
      \G \; \transg{\interaction{\pr}{\role{s}}{\obs}} \; \Gvt{\pp}{\pq}{\B} \Gi } \quad
    \inferrule*[left=\rulename{GR3}]{ \G \sub{\toc{\sfmu} \ty . \G}{\ty} \transg{\inter} \Gi}{\toc{\sfmu} \ty . \G \transg{\inter} \Gi} \\[3mm]
    \inferrule*[left=\rulename{GR4}]{ j \in I}{
      \GvtPair{\pp}{\pq}{\lab[i] : \G[i]}{i \in I} \;
      \transg{\interaction{\pp}{\pq}{\lab[j]}} \; \G[j]} \quad
    \inferrule*[left=\rulename{GR5}]{ \forall i \in I : \G[i]
      \transg{\interaction{\pr}{\role{s}}{\obs}} \Gi[i] \ \ \set{\pr, \role{s}} \cap \set{\pp, \pq} = \emptyset}{ \GvtPair{\pp}{\pq}{\lab[i] : \G[i]}{i \in I} \transg{\interaction{\pr}{\role{s}}{\obs}} \GvtPair{\pp}{\pq}{\lab[i] : \Gi[i]}{i \in I}} \quad
\end{array}
\]
\end{definition}
Rule
\rulename{GR1} exposes the head communication of a message whose
payload is of sort $\B$, while \rulename{GR4} handles selections by
emitting the chosen label.
Recursive protocols evolve by one unfolding step under
\rulename{GR3}, mirroring \rulename{LR5} on the local side.

Rules \rulename{GR2} and \rulename{GR5} are key: they allow
interactions between \emph{disjoint} participant pairs to proceed
inside the continuation of another interaction.
Concretely, if $\set{\pp,\pq} \cap \set{\pr,\role{s}} = \emptyset$,
then $\pp\gar\pq\;\goc{\langle}\B\goc{\rangle};\;
\pr\gar\role{s}\;\goc{\langle}\Bi\goc{\rangle};\;\G$ and
$\pr\gar\role{s}\;\goc{\langle}\Bi\goc{\rangle};\;
\pp\gar\pq\;\goc{\langle}\B\goc{\rangle};\;\G$ have
exactly the same set of traces.
Although global types are written as a sequential chain of
interactions, they therefore naturally represent concurrent
behaviour: any interleaving of independent interactions is
semantically equivalent.
This observation is central to our synthesis algorithm
(\SEC{sec:proof-main-theorem}), which must choose a particular
interleaving when constructing a global type from a set of local
types.

We write $\ctx \transi{}$ to mean $\exists \ctxi, \inter.\ \ctx \transi{\inter} \ctxi$,
and $\ctx \transi{} \ctxi$ to mean $\exists \inter.\ \ctx \transi{\inter} \ctxi$.
Similarly, $\ctx \transi{\pp\pq}$ means $\exists \obs, \ctxi.\ \ctx \transi{\interaction{\pp}{\pq}{\obs}} \ctxi$,
and $\ctx \transi{\pp\pq} \ctxi$ means $\exists \obs.\ \ctx \transi{\interaction{\pp}{\pq}{\obs}} \ctxi$.
We write $\ctx \ntransi{}$ to mean $\neg(\ctx \transi{})$.
The same conventions apply to local types, one-sided, and global type reductions ($\transa{}$, $\transt{}$, and $\transg{}$).

\subsection{Fair and Live Reduction Sequences}
\label{subsec:typingcontext}
Many interesting properties are defined in terms of possible reduction
sequences of contexts and global types.

\begin{definition}[Reduction sequences]
  \label{def:stuck}
  \label{def:labels}
A \emph{reduction sequence} is a finite or infinite sequence of contexts related by synchronous transitions:
$\ctx[0] \transi{\inter[0]} \ctx[1] \transi{\inter[1]} \ctx[2]
  \transi{\inter[2]} \ldots$;
it may be \emph{empty}, consisting of the single context $\ctx[0]$ alone.
We write $(\ctx[i] \transi{\inter[i]} \ctx[i+1])_{i \in I}$ for such a
sequence, where $I$ indexes its \emph{transitions}, and write
$\overline{I}$ for its set of \emph{context indices}, so that
$\overline{I} = I \cup \set{\vert I \vert}$ if the sequence is finite
and $\overline{I} = I$ otherwise.
We write $\mathsf{stuck}(\ctx)$ if
$\ctx \ntransi{}$ and $\mathsf{unstuck}(\ctx)$ if $\ctx \transi{}$. 
\end{definition}
The context is \emph{stuck} if 
it cannot take any step, and \emph{unstuck}
if it can take at least one step. 
We are usually interested in \emph{maximal} reduction sequences
meaning that either the sequence is infinite or the final 
context is stuck.

\begin{definition}[Traces]
\label{def:traces}
The \emph{traces} of a global type $\G$ (resp.\ a context $\ctx$) are the (finite or infinite) sequences of interactions labelling maximal reduction sequences using $\transg{}$ (resp.\ $\transi{}$):
\[
\trace(\G) = \left\{\, \inter[0]\,\inter[1]\,\inter[2]\,\cdots
  \;\middle|\;
  \G = \G[0] \transg{\inter[0]} \G[1] \transg{\inter[1]} \G[2]
  \transg{\inter[2]} \cdots
  \text{ is maximal}
\,\right\}.
\]
\[
\trace(\ctx) = \left\{\, \inter[0]\,\inter[1]\,\inter[2]\,\cdots
  \;\middle|\;
  \ctx = \ctx[0] \transi{\inter[0]} \ctx[1] \transi{\inter[1]} \ctx[2]
  \transi{\inter[2]} \cdots
  \text{ is maximal}
\,\right\}.
\]
We write $\Gi \subseteq \G$ when $\trace(\Gi) \subseteq \trace(\G)$.
\end{definition}

While the term \emph{safety} has many uses in the literature, we adopt a specific definition of a \emph{safe} context. 
A context is \emph{safe} if it cannot reach a state where there is a mismatch between the expectations of the send and receive endpoints.

\begin{definition}[Safe contexts]
    \label{def:safe-contexts}
    A context $\ctx$ is \emph{safe} if, for all contexts $\ctx \transi{}^{*} \ctxi$ reachable from $\ctx$, we have that: 
    if $\ctxi \transt{\pp : \qsend{\pq}{\obs}}$ and $\ctxi \transt{\pq : \qrecv{\pp}{\obsi}}$, then $\ctxi \transi{\interaction{\pp}{\pq}{\obs}}$.
\end{definition}

For example, in a safe context, if $\pp$ is trying to send a $\tbool$ to $\pq$
then it should not be the case that $\pq$ is trying to receive an $\tint$.

We say a maximal reduction sequence is \emph{fair} in the sense that
enabled interactions are always eventually taken.

\begin{definition}[Fair and live reduction sequences]
A maximal reduction sequence $(\ctx[i] \transi{\inter[i]} \ctx[i+1])_{i \in I}$ is  
\emph{fair} when, for all $j \in \overline{I}$,
$\ctx[j] \transi{\pp \pq}$ implies that
$\inter[n] = \interaction{\pp}{\pq}{\obs}$
for some $n \in I$ with $n \ge j$ and some $\obs$.
\end{definition}

For example, we may have a context $\ctx$ which can self-loop in two different ways $\ctx \transi{\interaction{\pp}{\pq}{\B}} \ctx$ and $\ctx \transi{\interaction{\pr}{\role{s}}{\Bi}} \ctx$.
The reduction sequence $\ctx \transi{\interaction{\pp}{\pq}{\B}} \ctx \transi{\interaction{\pp}{\pq}{\B}} \ctx \transi{\interaction{\pp}{\pq}{\B}} \ldots$ is unfair because the interaction between $\pr$ and $\role{s}$ 
is enabled but never taken.

\begin{definition}[Live paths]
\label{def:live-paths}
A maximal reduction sequence $(\ctx[i] \transi{\inter[i]} \ctx[i+1])_{i \in I}$ is
\emph{live} if, for all $i \in \overline{I}$:
\begin{enumerate}
    \item if $\ctx[i] \transt{\pp : \qsend{\pq}{\B}}$, then there exists a $j \in I, j \ge i$ such that $\ctx[j] \transt{\pq : \qrecv{\pp}{\B}}$
    \item if $\ctx[i] \transt{\pq : \qrecv{\pp}{\B}}$, then there exists a $j \in I, j \ge i$ such that $\ctx[j] \transt{\pp : \qsend{\pq}{\B}}$
    \item if $\ctx[i] \transt{\pp : \qsend{\pq}{\lab}}$, then there exist a $j \in I, j \ge i$ and a label $\labi$ such that $\ctx[j] \transt{\pq : \qrecv{\pp}{\labi}}$
    \item if $\ctx[i] \transt{\pq : \qrecv{\pp}{\lab}}$, then there exist a $j \in I, j \ge i$ and a label $\labi$ such that $\ctx[j] \transt{\pp : \qsend{\pq}{\labi}}$
\end{enumerate}
\end{definition}

We disregard starvation caused solely by an unfair scheduler, and
therefore require liveness along every \emph{fair} maximal reduction
sequence.

\begin{definition}[Live contexts]
\label{def:livecontexts}
\label{def:safelivecontexts}
A context $\ctx$ is \emph{live} (written $\live(\ctx)$)
if all fair reduction sequences 
starting from that context are live.
We write $\slive(\ctx)$ if
$\safe(\ctx)$ and
$\live(\ctx)$. 
\end{definition}

\noindent
Note that context liveness is a type-level analogue of the process
liveness of Definition~\ref{def:properties}: rather than reasoning
about all reachable process states, we check a semantic property of
the typing context.
The two notions are linked by soundness
(Theorems~\ref{thm:topdown} and~\ref{thm:bottomup}): if
$\slive(\ctx)$ and $\M$ is typed by~$\ctx$, then $\M$ is live in the
sense of Definition~\ref{def:properties}.

\subsection{Fair Scheduling}
\label{sec:roundrobin}
As established above, interactions between disjoint participant pairs
commute in the global type semantics, so a context with several
simultaneously enabled pairs admits many equivalent interleavings.
To synthesise a global type that faithfully captures \emph{all}
behaviours of the local context, we must choose an interleaving that
is \emph{fair}, so every enabled pair is eventually scheduled.
Combined with the assumption that the context is live, a fair schedule
guarantees that the synthesised global type preserves the semantics of
each participant.
We achieve fairness via a simple deterministic strategy, cycling
through enabled senders in a fixed order, captured by a relation
$\rrtrans{}$.

We define this strategy using two functions:
a set of \emph{enabled senders}, which contains all participants who are
the sender in an enabled communication pair at state $\ctx$;
the \emph{next-active} function chooses the
next participant who is the sender in an enabled communication pair
from $\ctx$.

\begin{definition}[Enabled and next-active functions] \leavevmode
\label{def:enabled}
\label{def:next-active}
Let $n$ be the total number of participants
and assume $\mathsf{unstuck}(\ctx)$.
When $\ctx$ can transition, we define:\\[1mm]
\centerline{
  $\mathsf{enabled}(\ctx) = \{\, \pp \mid \exists \pq.\ \ctx \transi{\pp\pq}
    \,\}$ \qquad $\textsf{next-active}(\ctx, \pp[i]) =
\operatorname*{argmin}_{\pp[j] \in \mathsf{enabled}(\ctx)}
\big((j-i) \bmod n\big)$}\\[1mm]
That is, $\textsf{next-active}(\ctx, \pp[i])$ selects the enabled sender
$\pp[j]$ whose index $j$ minimises $(j-i) \bmod n$, i.e.\ the first
enabled participant at or after $\pp[i]$ in the round-robin order
(including $\pp[i]$ itself, when it is enabled).
We omit the first argument and write $\textsf{next-active}(\pp[i])$
when the context is unambiguous.
\end{definition}

For example, if we are in a state with $\ctx \transi{\pp[0] \pp[1]}$
and $\ctx \transi{\pp[2] \pp[3]}$ then we have $\mathsf{enabled}(\ctx)
= \set{\pp[0],\pp[2]}$ and 
we have $\textsf{next-active}(\pp[0]) = \textsf{next-active}(\pp[3]) = \pp[0]$ and
$\textsf{next-active}(\pp[1]) = \textsf{next-active}(\pp[2]) = \pp[2]$. 
The strategy cycles through participants, prioritising each one in turn.
Hence $\rrtrans{}$ is a relation between pairs of contexts and a prioritised participant.

\begin{definition}[Round-robin]
  We write $(\ctx, \pp[i]) \rrtrans{\inter} (\ctxi, \pp[k])$ to show that the
  state $\ctx$ with participant $\pp[i]$ prioritised transitions to the state
  $\ctxi$ with participant $\pp[k]$ prioritised.
  \[
\pp[j] = \textsf{next-active}(\ctx, \pp[i]),\quad
\inter = \interaction{\pp[j]}{\pq}{\obs},\quad
\ctx \transi{\inter} \ctxi,\quad
k = (j+1)\bmod n.
\]
\end{definition}

\begin{figure}[t]
\centering
\newcommand{\rrhl}[1]{\begingroup\setlength{\fboxsep}{0pt}\colorbox{yellow!45}{\strut$#1$}\endgroup}
\newcommand{\rrhlgray}[1]{\begingroup\setlength{\fboxsep}{0pt}\colorbox{yellow!45}{\strut{\color{gray!70}$#1$}}\endgroup}
\newcommand{\rrpgray}[1]{\begingroup\renewcommand{\rolecolor}{gray!70}\pp[#1]\endgroup}
\newcommand{\rrna}[1]{\begingroup\setlength{\fboxsep}{0.4pt}\fcolorbox{orange!85!black}{white}{\strut$#1$}\endgroup}
\newcommand{\rrnahl}[1]{\begingroup\setlength{\fboxsep}{0.4pt}\fcolorbox{orange!85!black}{yellow!45}{\strut$#1$}\endgroup}
\begin{subfigure}[b]{0.48\textwidth}
\centering
\begin{tikzpicture}[scale=0.78, transform shape,
  rrbox/.style={
    draw,
    rounded corners=3pt,
    inner xsep=1pt,
    inner ysep=1pt,
    align=left,
    font=\tiny
  },
  rrarr/.style={-{Stealth[length=1.9mm,width=1.4mm]}, semithick, violet!75!black},
  rrlabel/.style={font=\tiny, text=violet!75!black}
]
  \node[rrbox] (a) at (0,1.5) {
    \begin{tabular}{@{}c@{}}
      $\mathbf{\circled{1}\;(\ctx[a],\, \pp[0])}$ \\
      \begin{tabular}{@{}l|l@{}}
        \rrnahl{\pp[0]} & \rrhl{\T[0]} \\
        $\rrpgray{1}$ & {\color{gray!70}$\T[1]$} \\
        $\pp[2]$ & {$\T[2]$} \\
        $\rrpgray{3}$ & {\color{gray!70}$\T[3]$}
      \end{tabular}
    \end{tabular}
  };

  \node[rrbox] (b) at (2.5,2.65) {
    \begin{tabular}{@{}c@{}}
      $\mathbf{\circled{2}\;(\ctx[a],\, \pp[1])}$ \\
      \begin{tabular}{@{}l|l@{}}
        $\pp[0]$ & $\T[0]$ \\
        \rrhlgray{\rrpgray{1}} & \rrhlgray{\T[1]} \\
        \rrna{\pp[2]} & {$\T[2]$} \\
        $\rrpgray{3}$ & {\color{gray!70}$\T[3]$}
      \end{tabular}
    \end{tabular}
  };

  \node[rrbox] (c) at (2.5,0.35) {
    \begin{tabular}{@{}c@{}}
      $\mathbf{\circled{3}\;(\ctx[b],\, \pp[1])}$ \\
      \begin{tabular}{@{}l|l@{}}
        $\rrpgray{0}$ & {\color{gray!70}$\tend$} \\
        \rrhlgray{\rrpgray{1}} & \rrhlgray{\tend} \\
        \rrna{\pp[2]} & {$\T[2]$} \\
        $\rrpgray{3}$ & {\color{gray!70}$\T[3]$}
      \end{tabular}
    \end{tabular}
  };

  \node[rrbox] (d) at (5.0,2.65) {
    \begin{tabular}{@{}c@{}}
      $\mathbf{\circled{4}\;(\ctx[c],\, \pp[3])}$ \\
      \begin{tabular}{@{}l|l@{}}
        \rrna{\pp[0]} & {$\T[0]$} \\
        $\rrpgray{1}$ & {\color{gray!70}$\T[1]$} \\
        $\rrpgray{2}$ & {\color{gray!70}$\tend$} \\
        \rrhlgray{\rrpgray{3}} & \rrhlgray{\tend}
      \end{tabular}
    \end{tabular}
  };

  \node[rrbox] (e) at (5.0,0.35) {
    \begin{tabular}{@{}c@{}}
      $\mathbf{\circled{5}\;(\ctx[d],\, \pp[3])}$ \\
      \begin{tabular}{@{}l|l@{}}
        $\rrpgray{0}$ & {\color{gray!70}$\tend$} \\
        $\rrpgray{1}$ & {\color{gray!70}$\tend$} \\
        $\rrpgray{2}$ & {\color{gray!70}$\tend$} \\
        \rrhlgray{\rrpgray{3}} & \rrhlgray{\tend}
      \end{tabular}
    \end{tabular}
  };

  \draw[rrarr] (a.east) -- (b.west) node[midway, above=1pt, rrlabel] {$\lab[1]$};
  \draw[rrarr] (a.east) -- (c.west) node[midway, below=1pt, rrlabel] {$\lab[2]$};
  \draw[rrarr, bend right=25] (b.west) to node[midway, above=1pt, rrlabel] {$\lab[3]$} (a.north east);
  \draw[rrarr] (b.east) -- (d.west) node[midway, above=1pt, rrlabel] {$\lab[4]$};
  \draw[rrarr] (c) edge[loop above, looseness=5, min distance=2mm] node[rrlabel] {$\lab[3]$} (c);
  \draw[rrarr] (d) edge[loop above, looseness=5, min distance=2mm] node[rrlabel] {$\lab[1]$} (d);
  \draw[rrarr] (d.south) -- (e.north) node[midway, right=1pt, rrlabel] {$\lab[2]$};
  \draw[rrarr] (c.east) -- (e.west) node[midway, above=1pt, rrlabel] {$\lab[4]$};

\end{tikzpicture}
\caption{Reduction graph on states $(\ctx,\pp)$,
drawing states with identical transitions once
({\setlength{\fboxsep}{1pt}\colorbox{yellow!45}{\strut prioritised}},
\mbox{\color{gray!70}not enabled},
{\setlength{\fboxsep}{1pt}\fcolorbox{orange!85!black}{white}{\strut next-active sender}}).}
\label{fig:rr-graph}
\end{subfigure}%
\hfill
\begin{subfigure}[b]{0.48\textwidth}
\centering
\begin{minipage}[c]{\linewidth}
\centering
\footnotesize
$\begin{array}{@{}r@{~=~}l@{\quad}r@{~=~}l@{}}
\T[0] & \toc{\sfmu} \ty . \tsel{\pp[1]}{\lab[1]:\ty,\ \lab[2]:\tend}
& \T[2] & \toc{\sfmu} \ty . \tsel{\pp[3]}{\lab[3]:\ty,\ \lab[4]:\tend} \\[2pt]
\T[1] & \toc{\sfmu} \ty . \tbra{\pp[0]}{\lab[1]:\ty,\ \lab[2]:\tend}
& \T[3] & \toc{\sfmu} \ty . \tbra{\pp[2]}{\lab[3]:\ty,\ \lab[4]:\tend}
\end{array}$

\vspace{4pt}
\hrule
\vspace{4pt}

{\scriptsize
$\begin{aligned}
\G ={} & \toc{\sfmu}\ty[1].\; \pp[0]\gar\pp[1]\;
\color{gtypecolor}\left\{\color{black}
\begin{array}{@{}l@{}}
\lab[1] : \toc{\sfmu}\ty[2].\; \pp[2]\gar\pp[3]\;
\goc{\{}\lab[3] : \ty[1],\; \\
\quad \lab[4] : \toc{\sfmu}\ty[3].\; \pp[0]\gar\pp[1]\;
\goc{\{}\lab[1] : \ty[3],\; \lab[2] : \gend\goc{\}}\goc{\big\}} \\
\lab[2] : \toc{\sfmu}\ty[4].\; \pp[2]\gar\pp[3]\;
\goc{\{}\lab[3] : \ty[4],\; \lab[4] : \gend\goc{\}}
\end{array}
\color{gtypecolor}\right\}
\end{aligned}$
}
\end{minipage}
\caption{Local types $\ctx$ (top) and synthesised global type $\G$ (bottom).
The two disjoint pairs ($\pp[0]$/$\pp[1]$ and $\pp[2]$/$\pp[3]$) are interleaved by the round-robin schedule.}
\label{fig:rr-types}
\end{subfigure}
\caption{Round-robin reduction graph with inferred global type.
Nodes \circled{1} and \circled{2} share the same context but differ in priority;
this distinction forces both pairs to be scheduled before the state repeats,
so the synthesised global type captures all interleavings.
The binder $\toc{\sfmu}\ty[2]$ in $\G$ is redundant and can be removed in a post-pass.}
\label{fig:round-robin-reduction-diagram}
\end{figure}

The two disjoint pairs in Figure~\ref{fig:round-robin-reduction-diagram} are truly concurrent, so the round-robin schedule must interleave them to explore all reachable states.
\begin{example}[Round-robin execution of Figure~\ref{fig:round-robin-reduction-diagram}]
\label{ex:round-robin-execution}
The context has two disjoint pairs ($\pp[0]$/$\pp[1]$ and $\pp[2]$/$\pp[3]$), and the execution proceeds as follows.
\begin{enumerate}[leftmargin=*, itemsep=1pt, topsep=2pt, label=\circled{\arabic*}]
\item Both pairs are enabled. With priority $\pp[0]$, $\textsf{next-active}$ chooses $\pp[0]$, which sends $\lab[1]$ or $\lab[2]$ to $\pp[1]$.
\item After $\lab[1]$, the context is still $\ctx[a]$ but the priority has advanced to $\pp[1]$. The scheduler skips to $\pp[2]$, which sends $\lab[3]$ back to \circled{1} or $\lab[4]$ to \circled{4}.
\item After $\lab[2]$ from \circled{1}, the pair $\pp[0]$/$\pp[1]$ has terminated. Only $\pp[2]$ remains enabled, sending $\lab[3]$ on the self-loop or $\lab[4]$ to \circled{5}.
\item After $\lab[4]$ from \circled{2}, the pair $\pp[2]$/$\pp[3]$ has terminated. Only $\pp[0]$ remains enabled, sending $\lab[1]$ on the self-loop or $\lab[2]$ to \circled{5}.
\item Terminal state: all local types are $\tend$.
\end{enumerate}
\end{example}

The round-robin schedule guarantees fairness by construction; the synthesis algorithm in the next section relies on this to ensure every reachable interaction is captured.

\section{Building Global Types}
\label{sec:building}
This section shows the main result of this paper: we infer a global
type from any safe-and-live typing context, and show that it is a principal
global type of the context.  

\subsection{Proof Outline}
\label{sec:proofoutline}
The proof proceeds in several steps. 
\begin{description}
\item[Synthesis definition and well-definedness.]
We define a synthesis function
$\textsf{synth}(\Sigma, \pp[i], \ctx)$ that constructs a global type
from a safe-and-live context by following its operational semantics
using fair scheduling. The environment~$\Sigma$ binds revisited
states to recursion variables
(Definition~\ref{def:synth-env}).
We show termination of this procedure
(Lemma~\ref{lemma:synth-terminates-for-slive-contexts})
via a finiteness argument on the reachable state space
(Lemmas~\ref{lemma:finite-reachable-local-states}--\ref{lemma:finite-branching-context-transitions}).

\item[Structural properties of the synthesised type.]
We prove that the synthesised type faithfully reflects the
participants in the context:
every active participant appears in the result
(Lemma~\ref{lemma:nonterminated-participant-appears}),
but not if its local type is $\tend$
(Lemma~\ref{lemma:terminated-participant-not-appears}),
and the synthesised type is balanced
(Lemma~\ref{lemma:synthesised-balanced}).

\item[Backward closure and association]
We define a composite relation $\cofullprojsub{p}$
(Definition~\ref{def:composition-projection-supertype}) that
packages coinductive projection and subtyping into a single
coinductive relation.
We show that the candidate global type we synthesise is related
by this composite relation to the local context we started with
(Lemma~\ref{lemma:backward-closure-synth-candidate}),
using a substitution lemma
(Lemma~\ref{lemma:substitution-for-synth}), which resolves the
environment bindings up to unfolding equivalence, and
preservation of safety and liveness under reduction
(Lemma~\ref{lemma:preservation-reduction}).

\item[Splitting the composite relation]
We show the composite relation can be decomposed into
separate coinductive projection and subtyping derivations
(Lemma~\ref{lemma:composite-implies-association}),
using an auxiliary result that a common widening bound yields a merge
(Lemma~\ref{lemma:widening-bound-yields-merge}), and using
inversion lemmas for subtyping
(Lemmas~\ref{lemma:unfold-compatible-with-subtyping}--\ref{lemma:inversion-subtyping}).
These results establish that 
the synthesised global type is associated with the context
(Theorem~\ref{theorem:existence-associated-global-type-for-safe-live-contexts}).

\item[Principal global types]
We show that association implies trace inclusion
(Lemma~\ref{lemma:association-implies-traces}) and that the
synthesised type has the same traces as the context
(Lemma~\ref{lemma:trace-equivalence-synth}).
Combining these with the existence theorem gives the
principality result: the synthesised type is the smallest
global type associated with the context
(Theorem~\ref{theorem:synthesised-type-is-principal}).
\end{description}

\subsection{Inferring a Global Type from a Local Context}
\label{sec:proof-main-theorem}
We define a synthesis function that infers global types from safe-and-live contexts.
It linearises parallel behaviour via fair scheduling,
tracking visited states to synthesise recursive types with $\mu$-binders.

\begin{definition}[Synthesis environment and function]
  \label{def:synth-env}
  \label{def:synth-rules}
A \emph{synthesis environment} $\Sigma$ is a finite partial map from
context-participant pairs to type variables: 
$\Sigma \;::=\; \varnothing \;\SEP\; \Sigma,\, (\ctx, \pp[j]) \mapsto
\ty$.  
We write $\Sigma(\ctx, \pp[j]) = \ty$ when $(\ctx, \pp[j]) \in
\dom{\Sigma}$. Then we define
the \emph{synthesis function}, 
$\textsf{synth}(\Sigma, \pp[i], \ctx)$, by the table below,
where $\pp[j] = \textsf{next-active}(\ctx, \pp[i])$: 
\begin{enumerate}
\item[\textbf{(1)}] if $\textsf{stuck}(\ctx)$, then
$\textsf{synth}(\Sigma, \pp[i], \ctx) = \gend$;
\item[\textbf{(2)}] if $\Sigma(\ctx, \pp[j]) = \ty$, then
$\textsf{synth}(\Sigma, \pp[i], \ctx)= \ty$;
\item[\textbf{(3)}] otherwise, let $\pq$ be the unique participant such that
$\ctx \transi{\pp[j]\pq}$, and let $\ty$ be a fresh type variable.
We case-split on the form of the enabled action:
\begin{enumerate}
\item[\textbf{(a)}] if $\unfold{\ctx(\pp[j])} = \tout{\pq}{\B}{\T'}$
  (value send), then
  $\textsf{synth}(\Sigma, \pp[i], \ctx) =
    \toc{\sfmu} \ty .\; \Gvt{\pp[j]}{\pq}{\B}{\G'}$
  where $\ctx \transi{\interaction{\pp[j]}{\pq}{\B}} \ctxi$ and
  $\G' = \textsf{synth}\!\left(\Sigma \cup \set{(\ctx, \pp[j]) \mapsto
    \ty},\; \pp[(j{+}1)\bmod n],\; \ctxi\right)$;
\item[\textbf{(b)}] if $\unfold{\ctx(\pp[j])} = \tselsub{\pq}{\lab[i] : \T[i]}{i \in I}$
  (label select), then
  $\textsf{synth}(\Sigma, \pp[i], \ctx) =
    \toc{\sfmu} \ty .\; \pp[j] \gar \pq\; \goc{\{}\, \lab : \G[\lab]
    \,\goc{\}}_{\lab \in L}$
  where
  $L = \{\, \lab \mid \ctx \transi{\interaction{\pp[j]}{\pq}{\lab}} \ctx[\lab]
  \,\}$ and
  $\G[\lab] = \textsf{synth}\!\left(\Sigma \cup \set{(\ctx, \pp[j]) \mapsto
    \ty},\; \pp[(j{+}1)\bmod n],\; \ctx[\lab]\right)$ for each $\lab \in L$.
\end{enumerate} 
\end{enumerate}  
\end{definition}

\begin{utheorem}[Complexity of Synthesis]
For a context $\ctx = \set{\ptag{\pp[i]}{\T[\pp[i]]}}_{i \in I}$ with
$n = \vert I \vert$, computing the graph representation of
$\textsf{synth}(\varnothing, \pp[i], \ctx)$, with one vertex per
reachable context-participant pair, takes time
$O(n^2 b \vert \ctx \vert^n)$, where $b$ bounds the number of labels
in any selection type of $\ctx$.
\end{utheorem}
\begin{proof}
The graph is computed by a standard reachability traversal that
maintains a global visited map over context-participant pairs, unlike
the path-local environment $\Sigma$ of
Definition~\ref{def:synth-rules} and of Algorithm~\ref{alg:synth} in
\SEC{sec:implementation}, so each reachable pair is expanded at most
once and later arrivals add edges to the existing vertex.
The number of reachable contexts is bounded by
$\prod_{i \in I} \vert \T[\pp[i]] \vert \le \vert \ctx \vert^n$, so
there are at most $n \vert \ctx \vert^n$ reachable
context-participant pairs, each with at most $b$ outgoing edges.
Processing one such pair requires only $O(n)$ overhead to select the
next active participant, determine its partner, and collect the
enabled observables, plus $O(n)$ per edge to build the successor pair
and query the visited map.
Hence the total running time is
$O(n \vert \ctx \vert^n \cdot nb) = O(n^2 b \vert \ctx \vert^n)$.
\end{proof}

\noindent
Unravelling the graph into a syntactic global type may add an
exponential factor, since $\toc{\sfmu}$-syntax cannot share subterms
across the branches of a choice; the Binary Counter family in
\SEC{sec:implementation} exhibits this growth.

\begin{example}[Synthesising Monte Carlo Global Type]
\label{ex:synthesis}
We illustrate synthesis on the running example from \SEC{sec:overview}.
Starting from the initial context
$\ctx = \set{\ptag{\role{m}}{\T[\role{m}]},\,
\ptag{\role{w_1}}{\T[\role{w_1}]},\,
\ptag{\role{w_2}}{\T[\role{w_2}]},\,
\ptag{\role{r}}{\T[\role{r}]}}$
with participant ordering
$\pp[0]{=}\role{m}$, $\pp[1]{=}\role{w_1}$, $\pp[2]{=}\role{w_2}$, $\pp[3]{=}\role{r}$,
we compute $\textsf{synth}(\varnothing, \role{m}, \ctx)$.

At $\ctx$, $\role{m}$ is the sole enabled sender, so synthesis applies
Definition~\ref{def:synth-rules}(3b), binding $(\ctx, \role{m}) \mapsto \ty$.
It then follows six deterministic interactions:
\begin{enumerate}[label=\textbf{(\arabic*)}, leftmargin=*, itemsep=1pt, topsep=2pt]
\item $\role{m}{\gar}\role{w_1}\;\goc{\{}\labname{map}\goc{\}}$. Only $\role{m}$ is enabled; selects $\labname{map}$ for $\role{w_1}$.
\item $\role{m}{\gar}\role{w_1}\;\goc{\langle}\tfloat\goc{\rangle}$. Sends the payload value to $\role{w_1}$.
\item $\role{w_1}{\gar}\role{r}\;\goc{\langle}\tfloat\goc{\rangle}$. Priority advances; $\role{w_1}$ forwards to $\role{r}$.
\item $\role{m}{\gar}\role{w_2}\;\goc{\{}\labname{map}\goc{\}}$. Round-robin returns to $\role{m}$; selects $\labname{map}$ for $\role{w_2}$.
\item $\role{m}{\gar}\role{w_2}\;\goc{\langle}\tfloat\goc{\rangle}$. Sends the payload value to $\role{w_2}$.
\item $\role{w_2}{\gar}\role{r}\;\goc{\langle}\tfloat\goc{\rangle}$. Priority advances; $\role{w_2}$ forwards to $\role{r}$.
\end{enumerate}
Let $\ctx[6]$ be the context after these six steps:
{\footnotesize\[
\ctx[6] = \bigl\{\,
\role{m} \gc \tbra{\role{r}}{\labname{cont} \gc \T[\role{m}],\;
\labname{crash} \gc \tend,\;
\labname{stop} \gc
\tsel{\role{w_1}}{\labname{stop} \gc
\tsel{\role{w_2}}{\labname{stop} \gc \tend}}},\;
\role{r} \gc \tsel{\role{m}}{\labname{cont} \gc \T[\role{r}],\;
\labname{stop} \gc \tend},\;
\role{w_1} \gc \T[\role{w_1}],\;
\role{w_2} \gc \T[\role{w_2}]
\,\bigr\}
\]}
In $\ctx[6]$, the reducer chooses $\labname{cont}$ or $\labname{stop}$.
The $\labname{cont}$ branch restores $\ctx$ with priority $\role{m}$,
so rule~(2) reuses~$\ty$.
The $\labname{stop}$ branch terminates both workers and reaches a stuck context,
so rule~(1) returns $\gend$.
The result is $\G[\text{inf}] = \G[\text{map}]$ from \SEC{sec:overview}.
\end{example}

We first show that synthesis is a well-defined
construction on safe-and-live inputs. Termination of the 
synthesis procedure follows directly from several standard finiteness 
properties (Lemmas~\ref{lemma:finite-reachable-local-states}--\ref{lemma:finite-branching-context-transitions},
see \appref{app:building-proofs}).

\begin{restatable}[Termination of Synthesis]{lemma}{synthterminatesforslivecontexts}
  \label{lemma:synth-terminates-for-slive-contexts}
  For any context $\ctx$ and any priority $\pp[i]$,
  the computation $\textsf{synth}(\varnothing, \pp[i], \ctx)$ terminates.
\end{restatable}
\begin{proof}[Proof (full details in \appref{app:building-proofs})]
Let $S = \set{(\ctxi,\pp) \mid \ctx \transi{}^{*}\ctxi,\; \pp \in \dom{\ctx}}$,
which is finite by
Lemma~\ref{lemma:finite-reachable-context-priority-space}.
Each recursive call of $\textsf{synth}$ adds a fresh pair
$(\ctxi,\pp[j])$ to~$\Sigma$, so the measure
$m(\Sigma) = |S \setminus \dom{\Sigma}|$ strictly decreases.
Since the non-recursive cases (stuck context or revisited state)
terminate immediately and each context has only finitely many
successors
(Lemma~\ref{lemma:finite-branching-context-transitions}),
the computation terminates.
\end{proof}

Distinct runs of $\textsf{synth}$ fold the same cyclic behaviour at
different entry points, so their results agree only up to
unfolding. We make this precise with the following equivalence.

\begin{definition}[Unfolding Equivalence]
\label{def:unfolding-equivalence}
\emph{Unfolding equivalence} $\unfeq$ is the largest relation on
closed global types such that $\G \unfeq \Gi$ implies one of:
(1)~$\unfold{\G} = \unfold{\Gi} = \gend$;
(2)~$\unfold{\G} = \Gvt{\pp}{\pq}{\B}{\G[1]}$ and
$\unfold{\Gi} = \Gvt{\pp}{\pq}{\B}{\Gii[1]}$ with
$\G[1] \unfeq \Gii[1]$;
(3)~$\unfold{\G} = \GvtPair{\pp}{\pq}{\lab[i] : \G[i]}{i \in I}$ and
$\unfold{\Gi} = \GvtPair{\pp}{\pq}{\lab[i] : \Gii[i]}{i \in I}$ with
$\G[i] \unfeq \Gii[i]$ for all $i \in I$.
\end{definition}

\begin{restatable}[Substitution for Synth]{lemma}{substitutionforsynth}
\label{lemma:substitution-for-synth}
\proofapp{\appref{app:building-proofs}}
Let $\slive(\ctx)$ and $\slive(\ctxi)$, let $\pp[j]$ be an enabled
sender of $\ctxi$, and let $\ty$ be a type variable. Then for every
priority $\pp[i]$,
  \[
  \textsf{synth}\!\left(\{\,(\ctxi, \pp[j]) \mapsto \ty\,\}, \pp[i], \ctx\right)\sub{\textsf{synth}(\varnothing, \pp[j], \ctxi)}{\ty}
  \unfeq
  \textsf{synth}(\varnothing, \pp[i], \ctx).
  \]
\end{restatable}

Termination shows synthesis returns some global type.
Now we use the liveness assumption on the context to demonstrate 
any active participant is in fact eventually incorporated into the synthesised type.

\begin{restatable}[Non-Terminated Participant Appears in Synthesised Type]{lemma}{nonterminatedparticipantappearsinsynthesisedtype}
  \label{lemma:nonterminated-participant-appears}
  For any context $\ctx$ satisfying $\slive(\ctx)$, any priority $\pp[i]$, and any participant $\pp \in \dom{\ctx}$, if $\unfold{\ctx(\pp)} \neq \tend$, then
  $\pp \in \pt{\textsf{synth}(\varnothing, \pp[i], \ctx)}$.
\end{restatable}
\begin{proof}[Proof (full details in \appref{app:building-proofs})]
By contradiction.
If $\pp$ never appears in the synthesised type, then
no interaction along any path in the recursion tree of
$\textsf{synth}$ involves~$\pp$, and so
$\ctx(\pp)$ is unchanged throughout.
Every leaf of the recursion tree is either a stuck context
(case~(1)) or a revisited state (case~(2)).
In case~(1), extending the corresponding reduction to a maximal
sequence yields a fair path on which $\pp$ is never served,
contradicting \emph{liveness}.
In case~(2), the revisited state creates a cycle that can be
repeated forever, yielding an infinite fair path that again never
involves~$\pp$, contradicting liveness.
\end{proof}

Two further lemmas
(Lemmas~\ref{lemma:terminated-participant-not-appears}
and~\ref{lemma:no-transition-safe-live-only-end},
\appref{app:building-proofs})
give the converse: terminated participants do not appear in the
synthesised type, and a stuck safe-and-live context consists entirely of $\tend$ locals.
The second key use of liveness comes in the backward-closure argument
(Lemma~\ref{lemma:backward-closure-synth-candidate}), via
Lemma~\ref{lemma:preservation-reduction}, ensuring successor
contexts remain safe and live.

At this point the remaining gap is how to show that 
projecting our candidate global type recovers the original
context. Because full merge may widen offered branches, ordinary
projection is slightly too rigid for this step.
To establish association, we use a combined
relation $\G \cofullprojsub{p} \T$ packaging coinductive full-merging
projection with branch subtyping from merge
(Definition~\ref{def:composition-projection-supertype}).
The key difference from ordinary projection is at branching:
\rulename{SPR5} allows branching on a superset $I \supseteq J$ of the global labels (reflecting \rulename{M5}),
while \rulename{SPR4} requires an exact label match for selections.
{\small
\[
\begin{array}{c}
  {
            \cinferrule[\rulename{SPR4}]{ \forall i \in I : \G[i] \cofullprojsub{p} \T[i]}{ \GvtPair{\pp}{\pq}{\lab[i] : \G[i]}{i \in I}  \cofullprojsub{p} \tselsub{\pq}{\lab[i] : \T[i]}{i \in I}} \quad
            \cinferrule[\rulename{SPR5}]{ \forall j \in J : \G[j]
              \cofullprojsub{p} \T[j] \\ J \subseteq I}{
              \GvtPair{\pq}{\pp}{\lab[j] : \G[j]}{j \in J}
              \cofullprojsub{p} \tbrasub{\pq}{\lab[i] : \T[i]}{i \in
                I}} \quad
            \cinferrule[\rulename{SPR6}]{ \forall i \in I : \G[i]
              \cofullprojsub{p} \T \\ \pp \notin \set{\pq, \pr}}{
              \GvtPair{\pq}{\pr}{\lab[i] : \G[i]}{i \in I}
              \cofullprojsub{p} \T}
            }
\end{array}
\]
}
We write $\G \cofullprojsubctx \ctx$ when $\dom{\ctx} \supseteq \pt{\G}$ and
$\G \cofullprojsub{p} \ctx(\pp)$ for all $\pp \in \dom{\ctx}$;
for $\pp \notin \pt{\G}$ this forces $\unfold{\ctx(\pp)} = \tend$
(Lemma~\ref{lemma:completed-participants}).

We make use of standard inversion lemmas for subtyping
(Lemmas~\ref{lemma:unfold-compatible-with-subtyping}
and~\ref{lemma:inversion-subtyping}, see \appref{app:building-proofs})
and that safety and liveness are preserved under reduction.

\begin{restatable}[Preservation of Safety and Liveness under Reduction]{lemma}{preservationofsafetyandlivenessunderreduction}
    \label{lemma:preservation-reduction}
    If $\ctx$ is safe and live, and $\ctx \transi{\interaction{\pp}{\pq}{\obs}} \ctxi$, then $\ctxi$ is safe and live.
\end{restatable}

\begin{proof}[Proof (full details in \appref{app:building-proofs})]
Safety is inherited because every state reachable from~$\ctxi$ is
also reachable from~$\ctx$.
Liveness is preserved because any fair reduction from~$\ctxi$
can be extended by prepending the step from~$\ctx$, yielding a
fair reduction from~$\ctx$ which is live by assumption.
\end{proof}

The existence proof proceeds by coinduction. We define a candidate relation
pairing the synthesised global type with each participant's local type,
and show it is backward-closed under~$\cofullprojsub{p}$,
establishing $\G \cofullprojsub{p} \ctx(\pp)$ for every~$\pp$.
For each participant $\pp$:
\[
R_{\pp}
=
\set{
\left(\textsf{synth}(\varnothing, \pp[i], \ctx), \ctx(\pp)\right)
\mid
\slive(\ctx),\ \pp \in \dom{\ctx},\ \pp[i] \in \partset
}.
\]
We then define its closure $\widehat{R}_{\pp}$ inductively by the
following rules, where \rulename{RUnfEq} subsumes closure under left
unfolding, since $\toc{\sfmu}\ty.\G \unfeq
\G\sub{\toc{\sfmu}\ty.\G}{\ty}$ (Lemma~\ref{lemma:unfeq-properties}):
{\small
\[
\begin{aligned}
\text{\rulename{RCand}}\quad &
\frac{(\G,\T)\in R_{\pp}}{(\G,\T)\in \widehat{R}_{\pp}}
\qquad
\text{\rulename{RUnfEq}}\quad &
\frac{(\G,\T)\in \widehat{R}_{\pp} \qquad \G \unfeq \Gi}{(\Gi,\T)\in \widehat{R}_{\pp}}
\qquad
\text{\rulename{RUnfR}}\quad &
\frac{(\G,\toc{\sfmu} \ty . \T)\in \widehat{R}_{\pp}}{(\G,\T \sub{\toc{\sfmu} \ty . \T}{\ty})\in \widehat{R}_{\pp}}
\end{aligned}
\]
}

\begin{restatable}[Backward Closure of Candidate Relation]{lemma}{backwardclosureofsynthcandidate}
  \label{lemma:backward-closure-synth-candidate}
  For any participant $\pp$, the relation $\widehat{R}_{\pp}$ is closed backwards under the rules defining $\cofullprojsub{p}$.
\end{restatable}
\begin{proof}[Proof (full details in \appref{app:building-proofs})]
We show that for every $(\G,\T) \in \widehat{R}_{\pp}$
there is a rule of $\cofullprojsub{p}$ whose conclusion is
$\G \cofullprojsub{p} \T$ and whose premises all lie
in~$\widehat{R}_{\pp}$.
Each pair traces back to some
$\textsf{synth}(\varnothing, \pp[i], \ctx)$ and $\ctx(\pp)$.
We case-split on the synthesis step that produced~$\G$.
The key case is when synthesis chose an interaction
$\ctx \transi{\interaction{\pp[j]}{\pq}{\obs}} \ctx[\lab]$:
by Lemma~\ref{lemma:preservation-reduction}, each
successor~$\ctx[\lab]$ remains safe and live, so the recursive
sub-calls place the premises back in~$\widehat{R}_{\pp}$,
up to unfolding equivalence
(via the substitution lemma and rule \rulename{RUnfEq}).
We then match the role of~$\pp$ (sender, receiver, or uninvolved)
to the corresponding rule of~$\cofullprojsub{p}$.
See \appref{app:building-proofs} for the full case analysis.
\end{proof}

\subsection{Existence of Well-Formed Global Types and
  Principal Global Type}
\label{sec:principal}
The first main result of the section states that synthesis
produces a global type associated with the source context.
We prove two main theorems:
(1) any safe-and-live context yields a well-formed global type (Theorem~\ref{theorem:existence-associated-global-type-for-safe-live-contexts});
(2) this type is \emph{principal}, having the minimum trace set (Theorem~\ref{theorem:synthesised-type-is-principal}).
We use the association relation \cite{revisited,pischke2025,BHYZ2023,lmcs2025}
between contexts and global types.

\begin{definition}[Association]
\label{def:association}
We say typing context $\ctxi$ is \emph{associated} with global type
$\G$, written $\ctxi \assoc \G$, if $\ctxi$ can be split into two
disjoint (possibly empty) sub-contexts $\ctxi = \ctxact, \ctxterm$
where:
(1)~$\ctxact$ contains subtypes of projections of $\G$:
$\dom{\ctxact} = \pt{\G}$ and there exists $\ctx[0]$ such that
$\dom{\ctx[0]} = \pt{\G}$, $\G \cofullproj{p} \ctx[0](\pp)$ for all
$\pp \in \pt{\G}$, and $\ctxact \subt \ctx[0]$;
(2)~$\ctxterm$ contains only terminated endpoints:
$\unfold{\ctxterm(\pp)} = \tend$ for all $\pp \in \dom{\ctxterm}$.
\end{definition}
\noindent
The sub-context $\ctxterm$ accounts for participants that have
completed their protocol and hence no longer appear in~$\G$.

The composite relation is designed so that pointwise projection directly yields association, linking the candidate global type back to the source context.
\begin{restatable}[Pointwise Composite Relation Implies Association]{lemma}{compositeimpliesassociation}
\label{lemma:composite-implies-association}
\proofapp{\appref{app:building-proofs}}
If $\G \cofullprojsubctx \ctx$, then $\ctx \assoc \G$.
\end{restatable}

\begin{restatable}[Existence of Associated Global Type for Safe, Live Contexts]{theorem}{existenceassociatedglobaltypeforsafelivecontexts}
\label{theorem:existence-associated-global-type-for-safe-live-contexts}
For any context $\ctx$ satisfying $\slive(\ctx)$, and any participant $\pp[i]$, let
$\G = \textsf{synth}(\varnothing, \pp[i], \ctx)$. Then $\ctx \assoc \G$.
\end{restatable}
\begin{proof}
By Lemma~\ref{lemma:synth-terminates-for-slive-contexts},
$\G$ is well-defined. For each $\pp \in \dom{\ctx}$, the pair
$(\G, \ctx(\pp)) \in R_{\pp} \subseteq \widehat{R}_{\pp}$.
By Lemma~\ref{lemma:backward-closure-synth-candidate},
$\widehat{R}_{\pp}$ is backward-closed under the rules of
$\cofullprojsub{p}$, so $\G \cofullprojsub{p} \ctx(\pp)$.
Moreover $\pt{\G} \subseteq \dom{\ctx}$ (immediate from
Definition~\ref{def:synth-rules}); hence
$\G \cofullprojsubctx \ctx$, and
Lemma~\ref{lemma:composite-implies-association} yields
$\ctx \assoc \G$.
\end{proof}

\begin{restatable}[Synthesised Global Type is Balanced]{lemma}{synthesisedbalanced}
\label{lemma:synthesised-balanced}
\proofapp{\appref{app:building-proofs}}
For any context $\ctx$ satisfying $\slive(\ctx)$ and any priority $\pp[i]$, let
$\G = \textsf{synth}(\varnothing, \pp[i], \ctx)$. Then $\G$ is balanced.
\end{restatable}

In addition to existence, we demonstrate that 
synthesis returns the semantically \emph{smallest} global type
compatible with the context (it does not create any extra behaviours).
This notion is exactly what we call a \emph{principal global type}.

\begin{definition}[Principal Global Type]
\label{def:principal-type}
A global type $\G$ is a \emph{principal global type} for a context
$\ctx$ if $\ctx \assoc \G$ and
for all \emph{balanced} global types $\Gi$ such that $\ctx \assoc \Gi$, we have
$\G \subseteq \Gi$.
That is, $\G$ is the smallest balanced global type associated with $\ctx$.
\end{definition}

\begin{restatable}[Completeness of Association]{lemma}{completenessofassociation}
\label{lemma:completeness-association}
\proofapp{\appref{app:building-proofs}}
If $\ctx \assoc \G$ with $\G$ balanced and $\ctx \transi{\inter} \ctxi$,
then there exists a balanced $\Gi$ such that $\G \transg{\inter} \Gi$
and $\ctxi \assoc \Gi$.
\end{restatable}

\begin{restatable}[Soundness of Association]{lemma}{soundnessofassociation}
\label{lemma:soundness-association}
\proofapp{\appref{app:building-proofs}}
If $\ctx \assoc \G$ with $\G$ balanced and $\G \transg{\inter} \Gi$,
then there exist $\interi$, a balanced $\Gii$, and $\ctxi$ such that
$\G \transg{\interi} \Gii$,
$\ctx \transi{\interi} \ctxi$,
and $\ctxi \assoc \Gii$.
\end{restatable}

The principality argument is semantic in that we compare associated global
types by their trace sets rather than by their syntax.
\begin{restatable}[Association Implies Trace Inclusion]{lemma}{associationimpliestraces}
\label{lemma:association-implies-traces}
\proofapp{\appref{app:building-proofs}}
If $\ctx \assoc \G$ with $\G$ balanced, then $\trace(\ctx) \subseteq \trace(\G)$.
\end{restatable}

\begin{restatable}[Soundness of the Combined Relation]{lemma}{soundnessofcombinedrelation}
\label{lemma:soundness-combined-relation}
\proofapp{\appref{app:building-proofs}}
If $\G \cofullprojsubctx \ctx$ and $\G \transg{\inter} \Gi$,
then there exists $\ctxi$ such that $\ctx \transi{\inter} \ctxi$
and $\Gi \cofullprojsubctx \ctxi$.
\end{restatable}

\begin{restatable}[Soundness Implies Trace Inclusion]{lemma}{soundnessimpliestraces}
\label{lemma:soundness-implies-traces}
\proofapp{\appref{app:building-proofs}}
If $\G \cofullprojsubctx \ctx$, then $\trace(\G) \subseteq \trace(\ctx)$.
\end{restatable}

\begin{restatable}[Combined Relation Implies Trace Equality]{lemma}{combinedrelationimpliestraceequality}
\label{lemma:combined-relation-implies-trace-equality}
If $\G \cofullprojsubctx \ctx$ and $\G$ is balanced, then $\trace(\ctx) = \trace(\G)$.
\end{restatable}
\begin{proof}
By Lemma~\ref{lemma:composite-implies-association},
$\ctx \assoc \G$, so $\trace(\ctx) \subseteq \trace(\G)$
by Lemma~\ref{lemma:association-implies-traces}, as $\G$ is balanced.
The reverse inclusion $\trace(\G) \subseteq \trace(\ctx)$
follows by Lemma~\ref{lemma:soundness-implies-traces}.
\end{proof}

\begin{restatable}[Trace Equivalence of Context and Synthesised Type]{lemma}{traceequivalencesynth}
\label{lemma:trace-equivalence-synth}
For any context $\ctx$ satisfying $\slive(\ctx)$ and any priority $\pp[i]$, let
$\G = \textsf{synth}(\varnothing, \pp[i], \ctx)$. Then $\trace(\ctx) = \trace(\G)$.
\end{restatable}
\begin{proof}
By the proof of Theorem~\ref{theorem:existence-associated-global-type-for-safe-live-contexts},
we have $\G \cofullprojsubctx \ctx$, and $\G$ is balanced by
Lemma~\ref{lemma:synthesised-balanced}.
The result follows by Lemma~\ref{lemma:combined-relation-implies-trace-equality}.
\end{proof}
\begin{restatable}[Synthesised Global Type is Principal]{theorem}{synthesisedtypeisprincipal}
\label{theorem:synthesised-type-is-principal}
For any context $\ctx$ satisfying $\slive(\ctx)$, and any priority $\pp[i]$,
let $\G = \textsf{synth}(\varnothing, \pp[i], \ctx)$. Then $\G$ is a principal global type for $\ctx$.
\end{restatable}
\begin{proof}
By Definition~\ref{def:principal-type}, it suffices to show
(1)~$\ctx \assoc \G$ and
(2)~$\G \subseteq \Gi$ for every balanced $\Gi$ with $\ctx \assoc \Gi$.

For~(1), Theorem~\ref{theorem:existence-associated-global-type-for-safe-live-contexts}
gives $\ctx \assoc \G$.

For~(2), suppose $\ctx \assoc \Gi$ with $\Gi$ balanced.
By Lemma~\ref{lemma:association-implies-traces}, $\trace(\ctx) \subseteq \trace(\Gi)$.
By Lemma~\ref{lemma:trace-equivalence-synth}, $\trace(\G) = \trace(\ctx)$.
Hence $\trace(\G) \subseteq \trace(\Gi)$, i.e.\ $\G \subseteq \Gi$.
\end{proof}

\noindent
With existence and principality established, we now show that
the top-down approach accepts all bottom-up typable sessions.

\section{Completeness of the Top-Down Typing System}
\label{sec:typing} 
This section connects the synthesis result back to session typing and
proves completeness of the top-down system with respect to liveness.
Process typing $\Gamma \,\vdash\, \PP : \T$ is standard
\cite[Figure~5]{thien-nobuko-popl-25}, including the subsumption rule
\rulename{Sub}; \appref{app:typing} gives the full rules. We
therefore recall only the session-level top-down and bottom-up
judgements.

\begin{definition}[Top-down and bottom-up typing]
\label{def:typing}
We recall the session-level typing rules for a multiparty session~$\M$.
Rule \rulename{T-Sess} types $\M$ from a projected global type, while
rule \rulename{B-Sess} types $\M$ from a safe local context.

\small
\[
\begin{array}{c}
\inferrule[\rulename{Sub}]{
\Gamma \vdash \PP : \T\quad \T\subt \Ti
}{
\Gamma \vdash \PP : \Ti
}
\quad
\inferrule[\rulename{T-Sess}]{
{\begin{array}{c}
\forall i \in I. \quad
\vdash \PP[i] : \T[i]
\quad
\G \cofullproj{\pp[i]} \T[i]
\\
\G \ \balanced
\quad
\pt{\G}\subseteq \set{\pp[i] \ | \ i \in I}
\end{array}}
}{
\provestop \prod_{i \in I} \proctag{\pp[i]}{\PP[i]} :
{\set{\ptag{\pp[i]}{\T[i]}}_{i\in I}}}
\quad
\inferrule[\rulename{B-Sess}]{ \forall i \in I \, \vdash \PP[i] : \T[i] \quad \safe(\set{\ptag{\pp[i]}{\T[i]}}_{i\in I})}{
\provesbot \prod_{i \in I} \proctag{\pp[i]}{\PP[i]} :
{\set{\ptag{\pp[i]}{\T[i]}}_{i\in I}}}
\end{array}
\]
\end{definition}

We now define the completeness with respect to a property $\varphi$.

\begin{definition}[Completeness wrt $\varphi$]
Suppose $\vdash \PP[i] : \T[i]$ with $i \in I$. Let
$\M=\prod_{i \in I} \proctag{\pp[i]}{\PP[i]}$
and $\ctx=\set{\ptag{\pp[i]}{\T[i]}}_{i\in I}$. 
Typing system $\vdash^\star$ is complete with
respect to $\varphi$ if $\varphi(\ctx)$ implies both
(1) there exists $\ctxi$ such that
$\vdash^\star \M:\ctxi$ with
$\varphi(\ctxi)$ and $\ctx\subt\ctxi$; and
(2) for all $\Mi$ such that $\M\red^\ast \Mi$,
there exists $\ctxii$ such that $\vdash^\star \Mi:\ctxii$ and
$\varphi(\ctxii)$.
\end{definition}  

We state the main theorems of the top-down and bottom-up systems.
Soundness of the top-down system and the bottom-up theorems are
proven in \cite{revisited} and \cite{Scalas2019} respectively;
subject reduction of the top-down system follows from our
association lemmas (\SEC{sec:building}).

\begin{theorem}[Top-down system]
\label{thm:topdown} Assume $\provestop\M:\ctx$.
\begin{enumerate}
\item (\text{Subject Reduction})
  \proofapp{\appref{app:typing-sr}}
  If $\M\red^\ast \Mi$, then there exist $\ctxi$ and $\ctxii$
  such that $\ctx \transi{}^{*} \ctxi$, $\ctxi \subt \ctxii$, and
  $\provestop\Mi:\ctxii$.
\item (\text{Soundness}) (\cite[Theorem~7]{revisited})
  If $\provestop\M:\ctx$, then (a) $\slive(\ctx)$; and
  (b) $\M$ is safe and live.
\end{enumerate}
\end{theorem}
\begin{theorem}[Bottom-up system \cite{Scalas2019}]
\label{thm:bottomup} Assume $\provesbot\M:\ctx$ and $\live(\ctx)$. 
\begin{enumerate} 
\item (\text{Subject Reduction with Liveness}) (\cite[Theorem~15]{Scalas2019})
  If $\M\red^\ast \Mi$, then 
  $\provesbot\Mi:\ctxi$ such that $\ctx \transi{}^{*} \ctxi$ and
  $\live(\ctxi)$. 

\item (\text{Liveness}) (\cite[Theorem~15]{Scalas2019})
If $\provesbot\M:\ctx$ and $\live(\ctx)$, 
then $\M$ is safe and live. 
\end{enumerate}
\end{theorem}

We now arrive at the main result:
 
\begin{theorem}[Top-down $=$ Bottom-up]
\label{thm:proc:completeness}
\leavevmode\\
(1)\label{item:topbot}
If $\provestop\M:\ctx$, then
$\provesbot\M:\ctx$ with $\live(\ctx)$.
\\
(2)\label{item:bottop}
If $\provesbot\M:\ctx$ with $\live(\ctx)$, then
there exists $\ctxi$ such that $\provestop\M:\ctxi$ and
$\ctx\subt \ctxi$.
\end{theorem}

\begin{proof}
(1)
Assume $\provestop\M:\ctx$.
By Theorem~\ref{thm:topdown}(2), $\slive(\ctx)$, hence
$\safe(\ctx)$ and $\live(\ctx)$.
The premises of \rulename{T-Sess} give
$\vdash \PP[i] : \T[i]$ for each~$i$, so \rulename{B-Sess} yields
$\provesbot\M:\ctx$.
Thus $\provesbot\M:\ctx$ with $\live(\ctx)$.

(2)
Assume $\provesbot\M:\ctx$ with $\live(\ctx)$.
From \rulename{B-Sess} we obtain $\vdash \PP[i] : \T[i]$ for each~$i$
and $\safe(\ctx)$, hence $\slive(\ctx)$.
Let $\G = \textsf{synth}(\varnothing, \pp[0], \ctx)$.
By Theorem~\ref{theorem:existence-associated-global-type-for-safe-live-contexts},
$\ctx \assoc \G$, and $\G$ is balanced by
Lemma~\ref{lemma:synthesised-balanced}.
By Definition~\ref{def:association}, $\ctx$ splits as
$\ctx = \ctxact, \ctxterm$, and clause~(1) provides $\ctx[0]$ with
$\dom{\ctx[0]} = \pt{\G}$,
$\G \cofullproj{p} \ctx[0](\pp)$ for all $\pp \in \pt{\G}$,
and $\ctxact \subt \ctx[0]$.
Define $\ctxi$ with $\dom{\ctxi} = \dom{\ctx}$ by
$\ctxi(\pp) = \ctx[0](\pp)$ for $\pp \in \pt{\G}$ and
$\ctxi(\pp) = \tend$ otherwise.
Then $\ctx \subt \ctxi$, since on $\pt{\G}$ this is
$\ctxact \subt \ctx[0]$, and elsewhere
$\unfold{\ctx(\pp)} = \tend$ gives $\ctx(\pp) \subt \tend$
(rules \rulename{S5} and \rulename{S7}).
So \rulename{Sub} gives $\vdash \PP[i] : \ctxi(\pp[i])$ for each~$i$.
Moreover $\G \cofullproj{\pp[i]} \ctxi(\pp[i])$ holds for each~$i$,
by the witness property for $\pp[i] \in \pt{\G}$ and by
Lemma~\ref{lemma:completed-participants}(4) for
$\pp[i] \notin \pt{\G}$.
Since $\pt{\G} \subseteq \dom{\ctx} = \set{\pp[i] \mid i \in I}$
and $\G$ is balanced,
rule \rulename{T-Sess} applies, yielding $\provestop\M:\ctxi$.
\end{proof}

By Theorems~\ref{thm:topdown}, \ref{thm:bottomup}
and \ref{thm:proc:completeness}, we have:  

\begin{corollary}
\label{cor:completeness}~
\begin{enumerate}
\item $\provestop\M:\ctx$ is complete with respect to the intersection of
 safety and liveness.  
\item 
$\provesbot\M:\ctx$ with $\live(\ctx)$ is complete with respect to 
the intersection of safety and liveness.  
\end{enumerate}
\end{corollary}

In other words, for the class of type-checked live sessions, the
protocol-driven top-down approach is just as expressive as the
model-checking bottom-up approach.

\noindent
It may be tempting to weaken the liveness condition $\live(\ctx)$ in
Theorem~\ref{thm:proc:completeness}, either to deadlock-freedom or to
the following notion of
\emph{weak liveness}, which requires progress to remain possible from
every state rather than guaranteed along every fair path. However, this is
not possible when using the subsumption rule as these properties are not downwards-closed under $\subt$.

\begin{definition}[Weakly live contexts]
\label{def:weak-live}
A context $\ctx$ is \emph{weakly live} (written $\wlive(\ctx)$) if for every
$\ctxi$ with $\ctx \transi{}^{*} \ctxi$:
\begin{enumerate}
\item[(1)] $\ctxi \transt{\pp : \qsend{\pq}{\obs}}$ implies
  $\ctxi \transi{}^{*} \ctxii \transi{\interaction{\pp}{\pq}{\obs}}$ for some
  $\ctxii$; and
\item[(2)] $\ctxi \transt{\pp : \qrecv{\pq}{\obs}}$ (for some $\obs$) implies
  $\ctxi \transi{}^{*} \ctxii \transi{\interaction{\pq}{\pp}{\obsi}}$ for some
  $\ctxii$ and $\obsi$.
\end{enumerate}
\end{definition}

\noindent
From every reachable state, weak liveness requires that if a one-sided output is enabled
for a participant then the corresponding two-sided transition, with the same payload, is eventually feasible,
and that, if a one-sided input is enabled, then at least some input is eventually feasible for that participant.
This is the same definition as \cite[Fig.~5(5)]{Scalas2019} (they call this \emph{live} and use \emph{live+} for our notion of live).

\begin{example}[Weak liveness is not enough]
\label{ex:balancedness-required}
\label{ex:weak-liveness-not-closed}
Consider the (unbalanced) global type
\[
\G \;=\; \toc{\sfmu}\ty.\,\GvtPairs{\pp}{\pq}{
    \lab[1]\gc \ty,\;\;
    \lab[2]\gc \Gvt{\pp}{\role{r}}{\B}\gend
}
\]
and its coinductive full-merging projection
$\ctx = \set{\pp:\T[\pp],\,\pq:\T[\pq],\,\role{r}:\T[\role{r}]}$ with
\[
\T[\pp] = \toc{\sfmu}\ty.\,\tsel{\pq}{\lab[1]\gc \ty,\;\;
                                       \lab[2]\gc \tout{\role{r}}{\B}\tend}
\quad
\T[\pq] = \toc{\sfmu}\ty.\,\tbra{\pp}{\lab[1]\gc \ty,\;\;
                                       \lab[2]\gc \tend}
\quad
\T[\role{r}] = \tin{\pp}{\B}\tend.
\]
Deleting $\pp$'s $\lab[2]$ branch gives
$\Ti[\pp] = \toc{\sfmu}\ty.\,\tsel{\pq}{\lab[1]\gc \ty}$ and, since
subtypes may offer fewer selections (rule \rulename{S3}), a subtype
$\ctxi = \set{\pp:\Ti[\pp],\,\pq:\T[\pq],\,\role{r}:\T[\role{r}]}
\subt \ctx$ in which every reduction performs only
$\interaction{\pp}{\pq}{\lab[1]}$, so $\role{r}$ waits on $\pp$
forever.

The context $\ctx$ is deadlock-free and weakly live
($\role{r}$'s pending input becomes feasible once $\pp$ selects
$\lab[2]$), but not live, since the reduction in which $\pp$ always selects
$\lab[1]$ never delivers a message to $\role{r}$. Its subtype $\ctxi$
is not weakly live (and similarly not deadlock-free, taking a
participant blocked forever on a receive to be deadlocked, as in
\cite{DiGiustoUrsoLozesPPDP2025}).
Participant $\role{r}$ is stuck and
its pending input $\ctxi \transt{\role{r} : \qrecv{\pp}{\B}}$ violates
clause~(2) of Definition~\ref{def:weak-live}. Hence weak liveness is
not downwards-closed under $\subt$, unlike liveness, and
\rulename{Sub} would let this defect reach typed
processes (see also Example~5.14 of~\cite{Scalas2019}).

Concretely, our process syntax has only \emph{single}
selection $\procsel{\pq}{\lab}\,\PP$, so a process commits to one branch
while its selection type advertises a whole set of labels.
Any similar setting will require a typing system featuring subsumption. Consider:
\[
\small
\PPi[\pp] = \pfix{X}\,\procsel{\pq}{\lab[1]}\,\poc{X}
\qquad
\PP[\pq] = \pfix{Y}\,\procbra{\pp}{\lab[1] \gc \poc{Y},\; \lab[2] \gc \inact}
\qquad
\PP[\role{r}] = \procin{\pp}{x}{\inact}
\]
where $\PPi[\pp]$ has minimal type $\ctxi(\pp)$, while $\PP[\pq]$ and
$\PP[\role{r}]$ realise $\T[\pq]$ and $\T[\role{r}]$. Since
$\ctxi \subt \ctx$, rule \rulename{Sub} derives
$\vdash \PPi[\pp] : \ctx(\pp)$, so the session
$\Mi = \proctag{\pp}{\PPi[\pp]} \pc \proctag{\pq}{\PP[\pq]} \pc
\proctag{\role{r}}{\PP[\role{r}]}$
is typed by the weakly live context $\ctx$. Yet
$\PPi[\pp]$ never fires its branch to $\role{r}$, so $\role{r}$ blocks
forever and $\Mi$ is not live.
\end{example}

\noindent
This is why Theorem~\ref{thm:proc:completeness} is stated for $\live$.
Were the bottom-up judgement to require only deadlock-freedom or
$\wlive(\ctx)$, the session $\Mi$ above would type-check despite not
being live.

\section{Implementation}
\label{sec:implementation}
We implement a Haskell toolchain (GHC 9.6.6) and test it on the
examples from this paper and case studies from the literature.
\subsection{Implemented Features}
Our toolchain implements both sides of the theory. On the \textbf{top-down}
strategy (Figure~\ref{fig:topdown}), it includes the four \textbf{projection algorithms} of
\cite{thien-nobuko-popl-25} (IP, IF, CP, and CF; Definitions
\ref{def:coind-merge} and \ref{def:coind-proj}), \textbf{balancedness
checking} (Definition~\ref{def:balanced-global-types}), \textbf{synchronous
subtyping} (Definition~\ref{def:coind-subtyping}), and \textbf{process type
checking} (\appref{app:typing}). On the \textbf{bottom-up} strategy (Figure~\ref{fig:bottomup}), it
implements \textbf{local type inference} from processes
\cite[Sec.~5.2]{thien-nobuko-popl-25}, \textbf{safety and liveness checking} of
typing contexts via {\sf mpstk}~\cite{Scalas2019}, and \textbf{global type
synthesis} from safe-and-live contexts as in \SEC{sec:building}.
While \cite{thien-nobuko-popl-25} proposed the four projection
algorithms, no implementation was provided; ours is the first
implementation, in particular of the coinductive full-merging
projection (CF).

The toolchain parses the paper's
surface syntax into global, local, and
context automata on which the core algorithms operate. In particular, balancedness
is checked on the graph representation by comparing reachable and
unavoidable participants at each reachable vertex. Both sets can be
computed iteratively in $O(\vert \G \vert^2)$ time.

\subsection{Global Type Inference}

The core feature of our implementation is the synthesis algorithm (Algorithm~\ref{alg:synth}), which constructs a global type 
from a local context. The algorithm follows the structure of the function $\textsf{synth}(\Sigma, \pp[i], \ctx)$ (Definition~\ref{def:synth-rules} in \SEC{sec:building}). 

\noindent\begin{tabular}{@{}p{0.42\textwidth}@{\hspace{12pt}\hspace{12pt}}p{0.42\textwidth}@{}}
\begin{minipage}[t]{\linewidth}
\begin{algorithm}[H]
\scriptsize
\caption{Global Type Synthesis}\label{alg:synth}
\begin{algorithmic}
\Require Safe and live $\ctx$, initial priority $\pp[i]$.
\Ensure $\textsf{synth}(\varnothing, \pp[i], \ctx)$.
\Function{Synth}{$\Sigma, \pp[i], \ctx$}
  \If{$\textsf{stuck}(\ctx)$} \Return $\gend$ \EndIf
  \State $\pp[j] \gets \textsf{next-active}(\ctx, \pp[i])$
  \If{$(\ctx, \pp[j]) \in \dom{\Sigma}$}
    \State \Return $\Sigma(\ctx, \pp[j])$
  \EndIf
  \State Choose fresh $\ty$
  \State $\Sigma' \gets \Sigma \cup \set{(\ctx, \pp[j]) \mapsto \ty}$
  \State Choose $\pq$ such that $\ctx \transi{\pp[j]\pq}$
  \State $L \gets \set{\obs \mid \ctx \transi{\interaction{\pp[j]}{\pq}{\obs}}}$
  \ForAll{$\obs \in L$}
    \State Choose $\ctxi$ with
    \Statex \hspace{\algorithmicindent}$\ctx \transi{\interaction{\pp[j]}{\pq}{\obs}} \ctxi$
    \State $\G[\obs] \gets \Call{Synth}{\Sigma', \pp[(j{+}1)\bmod n], \ctxi}$
  \EndFor
  \State \Return $\toc{\sfmu} \ty . \pp[j] \rightarrow \pq \{\, \obs : \G[\obs] \,\}_{\obs \in L}$
\EndFunction
\end{algorithmic}
\end{algorithm}
\end{minipage}
&
\begin{minipage}[t]{\linewidth}
\vspace{28pt}
\small
\noindent\textbf{Algorithm outline.}
\begin{enumerate}[leftmargin=*,itemsep=2pt,topsep=3pt,label=(\arabic*)]
\item Select the next active sender $\pp[j]$ by cycling through participants in round-robin order (\SEC{sec:semantics}).
\item If the pair $(\ctx, \pp[j])$ was already visited, return the stored recursion variable to close the loop.
\item Otherwise, record a fresh variable $\ty$ and enumerate every observable enabled for $\pp[j]$.
\item Recurse on each successor context; assemble all branches into a global type node.
\item When all processes are idle ($\textsf{stuck}$), return $\gend$. Memoisation converts revisited states into $\mu$-binders.
\end{enumerate}
\end{minipage}
\end{tabular}

\subsection{Evaluations}

\newcommand{\tabY}{{\color{green!50!black}$\checkmark$}}
\newcommand{\tabN}{{\color{red!70!black}$\times$}}
\newcommand{\tabslow}[1]{{\color{blue!70!black}#1}}
\newcommand{\taberr}[1]{{\color{red}$^{#1}$}}
\begin{table}[t]
\centering
\caption{Top-down and bottom-up benchmarks give median times over $10$
  runs. \#p is the number of participants; Citation gives the source
  and, where available, the entry in Table~\ref{tab:example-global-types}.
  \emph{Top-down}, from the global type $G$: its size $|G|$ and
  balancedness (Bal), the size $|\Delta_{\textsf{pro}}|$ of its
  projection, success/failure of the inductive plain (IP), inductive
  full (IF) and coinductive plain (CP) projections, and the runtime of
  coinductive full (CF) projection, which succeeds on every row.
  \emph{Bottom-up}, from a local context $\Delta$: its size
  $|\Delta|$, safety, deadlock-freedom (see \cite{Scalas2019}) and liveness
  ($\textsf{safe}$/$\textsf{df}$/$\textsf{live}$), the {\sf mpstk}
  time to check them, and the size $|G_{\textsf{inf}}|$, balancedness
  (Bal) and synthesis time (Synth.)\ of the synthesised global type.
  Rows marked \emph{no source global type} are contexts that are not
  the direct projection of any global type (each safe-and-live one is
  associated with its synthesised type); for Travel Agency, Company
  Communication, Ping Pong, Circuit Breaker, and Ring($n$), the
  bottom-up input is the projection of the global type.
  {\color{red}$^1$} {\sf mpstk} stack overflow.
  {\color{red}$^2$} Time limit exceeded.}
\label{tab:benchmarks}
\resizebox{\textwidth}{!}{%
\scriptsize
\setlength{\tabcolsep}{3.5pt}%
\begin{tabular}{l l r r c r c c c r r c c c r r c r}
\toprule
& & & \multicolumn{7}{c}{Top-down} & \multicolumn{8}{c}{Bottom-up} \\
\cmidrule(lr){4-10}\cmidrule(lr){11-18}
Example & Citation & \#p & $|G|$ & Bal & $|\Delta_{\textsf{pro}}|$ & IP & IF & CP & CF & $|\Delta|$ & $\textsf{safe}$ & $\textsf{df}$ & $\textsf{live}$ & mpstk & $|G_{\textsf{inf}}|$ & Bal & Synth. \\
\midrule
Simple Travel Agency & \cite[Fig.~1(a)]{YoshidaGheri2020}; \exampletabref{g} & 2 & 9 & \tabY & 18 & \tabY & \tabY & \tabY & 52.0\,\textmu s & 18 & \tabY & \tabY & \tabY & \tabslow{190.0\,ms} & 9 & \tabY & 297.0\,\textmu s \\
Better Travel Agency & \cite[Fig.~1(b)]{YoshidaGheri2020}; \exampletabref{h} & 2 & 12 & \tabY & 24 & \tabY & \tabY & \tabY & 59.0\,\textmu s & 24 & \tabY & \tabY & \tabY & \tabslow{200.0\,ms} & 12 & \tabY & 299.0\,\textmu s \\
OAuth & \cite[Ex.~1]{Scalas2019}; \exampletabref{b} & 3 & 13 & \tabY & 27 & \tabN & \tabY & \tabN & 74.0\,\textmu s & 27 & \tabY & \tabY & \tabY & \tabslow{200.0\,ms} & 13 & \tabY & \tabslow{1.4\,ms} \\
Two Buyer & \cite[Ex.~2]{Scalas2019}; \exampletabref{c} & 3 & 20 & \tabN & 41 & \tabN & \tabN & \tabN & 116.0\,\textmu s & 41 & \tabY & \tabY & \tabN & \tabslow{210.0\,ms} & 20 & \tabN & \tabslow{2.8\,ms} \\
Inductive Plain & \cite[Ex.~4.8]{thien-nobuko-popl-25}; \exampletabref{i} & 3 & 10 & \tabY & 20 & \tabY & \tabY & \tabY & 43.0\,\textmu s & 20 & \tabY & \tabY & \tabY & \tabslow{190.0\,ms} & 10 & \tabY & 66.0\,\textmu s \\
Inductive Full & \cite[Ex.~4.8]{thien-nobuko-popl-25}; \exampletabref{j} & 3 & 10 & \tabY & 22 & \tabN & \tabY & \tabN & 48.0\,\textmu s & 27 & \tabY & \tabY & \tabY & \tabslow{200.0\,ms} & 16 & \tabY & 281.0\,\textmu s \\
Semantic Recursion & \cite{Tirore2023}; \exampletabref{f} & 4 & 9 & \tabY & 20 & \tabN & \tabN & \tabY & 56.0\,\textmu s & 20 & \tabY & \tabY & \tabY & \tabslow{190.0\,ms} & 7 & \tabY & 63.0\,\textmu s \\
Ring & \cite{Castro2026}; \exampletabref{k} & 3 & 15 & \tabY & 31 & \tabN & \tabY & \tabN & 66.0\,\textmu s & 31 & \tabY & \tabY & \tabY & \tabslow{200.0\,ms} & 15 & \tabY & \tabslow{1.5\,ms} \\
Odd-Even & \cite[Ex.~2.1]{Li2023}; \exampletabref{a} & 3 & 33 & \tabY & 68 & \tabN & \tabN & \tabN & 150.0\,\textmu s & 68 & \tabY & \tabY & \tabY & \tabslow{200.0\,ms} & 43 & \tabY & \tabslow{10.4\,ms} \\
Monte Carlo Gmap & \S \ref{sec:overview} & 4 & 18 & \tabY & 44 & \tabN & \tabN & \tabN & 148.0\,\textmu s & 42 & \tabY & \tabY & \tabY & \tabslow{220.0\,ms} & 18 & \tabY & \tabslow{9.8\,ms} \\
Monte Carlo Gmin & \S \ref{sec:overview} & 4 & 15 & \tabY & 32 & \tabY & \tabY & \tabY & 93.0\,\textmu s & 39 & \tabY & \tabY & \tabY & \tabslow{220.0\,ms} & 15 & \tabY & \tabslow{9.3\,ms} \\
Independent Pairs & Fig \ref{fig:round-robin-reduction-diagram} & 4 & 13 & \tabY & 20 & \tabY & \tabY & \tabY & 46.0\,\textmu s & 24 & \tabY & \tabY & \tabY & \tabslow{200.0\,ms} & 13 & \tabY & \tabslow{1.4\,ms} \\
Company Communication & \cite[Fig.~9]{DBLP:conf/ecoop/GheriLSTY22}; \exampletabref{m} & 6 & 30 & \tabY & 68 & \tabN & \tabY & \tabN & 178.0\,\textmu s & 68 & \tabY & \tabY & \tabY & \tabslow{270.0\,ms} & 44 & \tabY & \tabslow{202.0\,ms} \\
Online Wallet & \cite[Fig.~1]{DBLP:conf/rv/NeykovaYH13}; \exampletabref{n} & 3 & 26 & \tabY & 53 & \tabN & \tabY & \tabN & 155.0\,\textmu s & 53 & \tabY & \tabY & \tabY & \tabslow{220.0\,ms} & 26 & \tabY & \tabslow{18.3\,ms} \\
Adder & \cite[Fig.~1(a)]{FASE16EndpointAPI}; \exampletabref{l} & 2 & 13 & \tabY & 26 & \tabY & \tabY & \tabY & 62.0\,\textmu s & 26 & \tabY & \tabY & \tabY & \tabslow{200.0\,ms} & 13 & \tabY & 349.0\,\textmu s \\
Distributed Logging & \cite[Fig.~16]{DBLP:conf/ecoop/LagaillardieNY22}; \exampletabref{o} & 2 & 17 & \tabY & 34 & \tabY & \tabY & \tabY & 86.0\,\textmu s & 34 & \tabY & \tabY & \tabY & \tabslow{200.0\,ms} & 17 & \tabY & 541.0\,\textmu s \\
E-Voting & \cite{DBLP:conf/ecoop/LagaillardieNY22}; \exampletabref{p} & 2 & 22 & \tabY & 44 & \tabY & \tabY & \tabY & 105.0\,\textmu s & 44 & \tabY & \tabY & \tabY & \tabslow{200.0\,ms} & 22 & \tabY & \tabslow{1.3\,ms} \\
Ping Pong & \cite[Fig.~14]{BHYZ2023} & 2 & 6 & \tabY & 12 & \tabY & \tabY & \tabY & 27.0\,\textmu s & 12 & \tabY & \tabY & \tabY & \tabslow{190.0\,ms} & 6 & \tabY & 48.0\,\textmu s \\
Circuit Breaker & \cite[Fig.~14]{BHYZ2023} & 3 & 57 & \tabY & 123 & \tabN & \tabN & \tabN & 323.0\,\textmu s & 123 & \tabY & \tabY & \tabY & \tabslow{340.0\,ms} & 72 & \tabY & \tabslow{202.8\,ms} \\
MapReduce(5) & \cite[Ex.~3]{Scalas2019}; \exampletabref{d} & 5 & 25 & \tabY & 62 & \tabN & \tabN & \tabN & 179.0\,\textmu s & 62 & \tabY & \tabY & \tabY & \tabslow{270.0\,ms} & 35 & \tabY & \tabslow{324.9\,ms} \\
MapReduce(6) & \cite[Ex.~3]{Scalas2019}; \exampletabref{d} & 6 & 31 & \tabY & 78 & \tabN & \tabN & \tabN & 234.0\,\textmu s & 78 & \tabY & \tabY & \tabY & \tabslow{370.0\,ms} & 45 & \tabY & \tabslow{3.96\,s} \\
MapReduce(7) & \cite[Ex.~3]{Scalas2019}; \exampletabref{d} & 7 & 37 & \tabY & 94 & \tabN & \tabN & \tabN & 297.0\,\textmu s & 94 & \tabY & \tabY & \tabY & \tabslow{590.0\,ms} & 55 & \tabY & \tabslow{32.48\,s} \\
Independent Workers(7) & \cite[Ex.~4]{Scalas2019}; \exampletabref{e} & 7 & 59 & \tabY & 91 & \tabN & \tabN & \tabN & 410.0\,\textmu s & 61 & \tabY & \tabY & \tabY & \tabslow{320.0\,ms} & 73 & \tabY & \tabslow{2.77\,s} \\
Independent Workers(10) & \cite[Ex.~4]{Scalas2019}; \exampletabref{e} & 10 & 215 & \tabY & 169 & \tabN & \tabN & \tabN & \tabslow{1.8\,ms} & 91 & \tabY & \tabY & \tabY & \tabslow{1.83\,s} & \taberr{2} & \taberr{2} & \taberr{2} \\
Ring($5$)  & \exampletabref{u} & 5 & 7 & \tabY & 20 & \tabY & \tabY & \tabY & 51.0\,\textmu s & 20 & \tabY & \tabY & \tabY & \tabslow{210.0\,ms} & 7 & \tabY & 510.0\,\textmu s \\
Ring($8$)  & \exampletabref{u} & 8 & 10 & \tabY & 32 & \tabY & \tabY & \tabY & 87.0\,\textmu s & 32 & \tabY & \tabY & \tabY & \tabslow{220.0\,ms} & 10 & \tabY & \tabslow{10.7\,ms} \\
Ring($12$) & \exampletabref{u} & 12 & 14 & \tabY & 48 & \tabY & \tabY & \tabY & 161.0\,\textmu s & 48 & \tabY & \tabY & \tabY & \tabslow{250.0\,ms} & 14 & \tabY & \tabslow{489.8\,ms} \\
Ring($15$) & \exampletabref{u} & 15 & 17 & \tabY & 60 & \tabY & \tabY & \tabY & 229.0\,\textmu s & 60 & \tabY & \tabY & \tabY & \tabslow{300.0\,ms} & 17 & \tabY & \tabslow{6.92\,s} \\
$\ctx[5]$ & \cite[Fig.~4]{thien-nobuko-popl-25} & 3 & \multicolumn{7}{c}{\color{black!60}\emph{no source global type}} & 21 & \tabY & \tabY & \tabY & \tabslow{210.0\,ms} & 9 & \tabY & 308.0\,\textmu s \\
$\ctx[6]$ & \cite[Fig.~4]{thien-nobuko-popl-25} & 2 & \multicolumn{7}{c}{\color{black!60}\emph{no source global type}} & 10 & \tabN & \tabY & \tabN & \tabslow{210.0\,ms} & 3 & \tabY & 53.0\,\textmu s \\
$\ctx[7]$ & \cite[Fig.~4]{thien-nobuko-popl-25} & 5 & \multicolumn{7}{c}{\color{black!60}\emph{no source global type}} & 17 & \tabN & \tabY & \tabN & \tabslow{210.0\,ms} & 4 & \tabY & 131.0\,\textmu s \\
$\ctx[8]$ & \cite[Fig.~4]{thien-nobuko-popl-25} & 1 & \multicolumn{7}{c}{\color{black!60}\emph{no source global type}} & 3 & \tabY & \tabN & \tabN & \tabslow{200.0\,ms} & 1 & \tabY & 47.0\,\textmu s \\
$\ctx[9]$ & \cite[Fig.~4]{thien-nobuko-popl-25} & 3 & \multicolumn{7}{c}{\color{black!60}\emph{no source global type}} & 11 & \tabY & \tabY & \tabN & \tabslow{200.0\,ms} & 4 & \tabY & 40.0\,\textmu s \\
Coinductive Full(2) & \cite[Thm.~4.24]{thien-nobuko-popl-25}; \exampletabref{q} & 3 & 25 & \tabY & 62 & \tabN & \tabN & \tabN & 133.0\,\textmu s & 62 & \tabY & \tabY & \tabY & \tabslow{220.0\,ms} & 51 & \tabY & \tabslow{4.9\,ms} \\
Coinductive Full(3) & \cite[Thm.~4.24]{thien-nobuko-popl-25}; \exampletabref{r} & 3 & 42 & \tabY & 217 & \tabN & \tabN & \tabN & 393.0\,\textmu s & 217 & \tabY & \tabY & \tabY & \tabslow{540.0\,ms} & 314 & \tabY & \tabslow{105.4\,ms} \\
Coinductive Full(4) & \cite[Thm.~4.24]{thien-nobuko-popl-25}; \exampletabref{s} & 3 & 63 & \tabY & 1312 & \tabN & \tabN & \tabN & \tabslow{3.0\,ms} & 1312 & \taberr{1} & \taberr{1} & \taberr{1} & \taberr{1} & 2681 & \tabY & \tabslow{2.08\,s} \\
Coinductive Full(5) & \cite[Thm.~4.24]{thien-nobuko-popl-25}; \exampletabref{t} & 3 & 92 & \tabY & 14239 & \tabN & \tabN & \tabN & \tabslow{141.4\,ms} & 14239 & \taberr{1} & \taberr{1} & \taberr{1} & \taberr{1} & \taberr{2} & \taberr{2} & \taberr{2} \\
Coinductive Full Optimised(2) & \cite[Thm.~4.24]{thien-nobuko-popl-25} & 3 & \multicolumn{7}{c}{\color{black!60}\emph{no source global type}} & 40 & \tabY & \tabY & \tabY & \tabslow{210.0\,ms} & 25 & \tabY & \tabslow{1.1\,ms} \\
Coinductive Full Optimised(3) & \cite[Thm.~4.24]{thien-nobuko-popl-25} & 3 & \multicolumn{7}{c}{\color{black!60}\emph{no source global type}} & 61 & \tabY & \tabY & \tabY & \tabslow{220.0\,ms} & 42 & \tabY & \tabslow{3.4\,ms} \\
Coinductive Full Optimised(4) & \cite[Thm.~4.24]{thien-nobuko-popl-25} & 3 & \multicolumn{7}{c}{\color{black!60}\emph{no source global type}} & 86 & \tabY & \tabY & \tabY & \tabslow{240.0\,ms} & 63 & \tabY & \tabslow{9.1\,ms} \\
Coinductive Full Optimised(5) & \cite[Thm.~4.24]{thien-nobuko-popl-25} & 3 & \multicolumn{7}{c}{\color{black!60}\emph{no source global type}} & 119 & \tabY & \tabY & \tabY & \tabslow{270.0\,ms} & 92 & \tabY & \tabslow{22.7\,ms} \\
Binary Counter(2) & Example~\ref{ex:binary-counter} & 4 & \multicolumn{7}{c}{\color{black!60}\emph{no source global type}} & 50 & \tabY & \tabY & \tabY & \tabslow{220.0\,ms} & 49 & \tabY & 806.0\,\textmu s \\
Binary Counter(3) & Example~\ref{ex:binary-counter} & 5 & \multicolumn{7}{c}{\color{black!60}\emph{no source global type}} & 69 & \tabY & \tabY & \tabY & \tabslow{340.0\,ms} & 117 & \tabY & \tabslow{6.9\,ms} \\
Binary Counter(4) & Example~\ref{ex:binary-counter} & 6 & \multicolumn{7}{c}{\color{black!60}\emph{no source global type}} & 88 & \tabY & \tabY & \tabY & \tabslow{960.0\,ms} & 269 & \tabY & \tabslow{72.6\,ms} \\
Binary Counter(5) & Example~\ref{ex:binary-counter} & 7 & \multicolumn{7}{c}{\color{black!60}\emph{no source global type}} & 107 & \tabY & \tabY & \tabY & \tabslow{4.37\,s} & 605 & \tabY & \tabslow{646.9\,ms} \\
\bottomrule
\end{tabular}
}
\end{table}

We evaluate our implementation on examples from the literature and this
paper, organised around the two directions studied in the paper.
Table~\ref{tab:benchmarks} collects both directions.
Its top-down columns start from global protocols and ask
how much the choice of projection algorithm affects the reach of the
top-down approach. Its bottom-up columns start from local
contexts and ask the converse question: once safety and liveness have
been established bottom-up, can we reconstruct an associated global
type, and what does that reconstruction cost?
All benchmarks were run natively on an Apple~M1 processor (8~cores)
with 16\,GB of RAM, running macOS~14.3 and GHC~9.6.6, with {\sf mpstk}
on Java~8 and the mCRL2~202507 toolset.
Selected global types used in the evaluation are collected in
Table~\ref{tab:example-global-types} (\appref{app:example-global-types}).
The benchmark suite includes both hand-written examples from the
literature and several parameterised families, including MapReduce,
Independent Workers, Coinductive Full, and Binary Counter, so the
evaluation covers a broad range of protocols.
Rows in Table~\ref{tab:benchmarks} that fail safety or liveness
are included as diagnostics about implementation behaviour outside the
hypotheses of Sections~\ref{sec:building} and~\ref{sec:typing}. The
equivalence result concerns the safe-and-live rows.

\paragraph{\textbf{Top-down projection.}}
The top-down columns of Table~\ref{tab:benchmarks} show that the
practical reach of the top-down approach depends on the projection
algorithm rather than the methodology.
Coinductive full-merging projection~(CF)~\cite{thien-nobuko-popl-25}
succeeds on every benchmark, always within microseconds to
milliseconds. OAuth
(\exampletabref{b}) and Company Communication (\exampletabref{m})
are projectable only by the full-merging algorithms (IF and CF),
while Two Buyer (\exampletabref{c}), Odd-Even (\exampletabref{a}),
MapReduce (\exampletabref{d}),
Independent Workers (\exampletabref{e}), Circuit
Breaker~\cite{BHYZ2023}, and the Coinductive Full
family require coinduction as well and are projectable only by CF.
In particular, CF projects the entire ``Less Is More'' suite
of~\cite{Scalas2019}, designed to expose projection limitations, and
Ring (\exampletabref{k}), the main example of~\cite{Castro2026}. For
these benchmarks, the obstruction lies in weaker projection
algorithms, not in the absence of a suitable global specification.

\paragraph{\textbf{Bottom-up synthesis.}}
In the reverse direction, the tool checks each local context with
{\sf mpstk} and then runs synthesis. Every synthesised global type
was validated by reprojecting it and checking pointwise that the
input context is a subtype of the reprojection, confirming 
the association result of \SEC{sec:building} at the implementation level. 
On small and medium instances verification
dominates reconstruction. All syntheses from
safe-and-live contexts return balanced global types; the remaining
negative case is Two Buyer, which lies outside the safe-and-live
fragment. Among the largest instances, {\sf mpstk} exhausts its
stack on the two largest Coinductive Full contexts, while
Independent Workers(10) and Coinductive Full(5) exceed the
synthesis time limit.

The size columns also show that principality is semantic rather than
purely syntactic. Coinductive Full(4) has $|G|{=}63$ but
$|\Delta|{=}1312$ and $|G_{\mathsf{inf}}|{=}2681$, so a compact
global type can project to a large local context and synthesise
back to an even larger principal type.

Conversely, the Binary Counter protocol (Example~\ref{ex:binary-counter})
shows small local contexts ($|\Delta|{=}50$-$107$) which
induce global types that grow exponentially
($|G_{\mathsf{inf}}|{=}49$-$605$ for $2$-$5$ bits).

\paragraph{\textbf{Comparison of top-down and bottom-up.}}
Taken together, the benchmarks demonstrate a practical
advantage for the top-down approach with coinductive full-merging
projection. CF removes the projectability gap exhibited by weaker
EPP algorithms and is substantially faster than {\sf mpstk}
safety/liveness checking alone.

Conversely, associated global types can still be reconstructed from
every verified local context, confirming the theoretical equivalence.

\section{Related Work and Conclusion}
\label{sec:related}
Different theories and frameworks of MPST have been integrated
across programming languages and tools \cite{Y2024,BETTYTOOLBOOK}.
We discuss related work that is most relevant to our main
results. 

\myparagraph{Synthesis and composition of global protocols}
In earlier work in the context of 
communicating finite
state machines (CFSMs) 
\cite{basu_bultan_2011,basu_bultan_2016,finkel_lozes_2017},
synthesis of global protocols
refers to the automatic construction of a global
behaviour from local specifications represented by CFSMs.
Basu and Bultan 
\cite{basu_bultan_2011,basu_bultan_2016} 
introduced a ``send-synchronisability'' problem 
which asks whether the send
projections of its executions remain identical under both
asynchronous and rendezvous communication semantics.
They claimed that
send-synchronisability is decidable under mailbox and peer-to-peer
communications.
Finkel and Lozes \cite{finkel_lozes_2017} 
have shown that their proofs
are flawed by establishing
its undecidability in peer-to-peer systems.
Recently, Delpy et al. 
\cite{delpy_muscholl_sutre_concur2025}
have proven undecidability of mailbox communication, closing the open problem. 

In multiparty session types, \emph{building a global
  type or graph} from a set of participants performs an analogous synthesis: it
organises a temporal and spatial relation among participants,
ensuring a well-formed global specification is built.
The early works in \cite{Lange2012,Denielou2013} 
build global types from CFSMs, while  
\cite{Lange2015} builds more expressive \emph{global graphs}
with interleaved parallel and mixed choice constructs 
from CFSMs. Recent works focus on
composing two or more global types or choreographic automata to    
generate a well-formed global type. 
Barbanera et al. \cite{Barbanera2019JLAMP,Barbanera2021,Barbanera2022}
use \emph{gateways} 
to connect two protocols via forwarding participants at once, 
while Stolze et al.~\cite{Stolze2021} provide an alternative
but restricted approach which has no recursions.
Gheri and Yoshida \cite{GY2023} have proposed
\emph{hybrid multiparty session types} which are
global types enriched with local type actions,  
enabling a composition of more than two subprotocols.
Subsequent work
\cite{Barbanera2023a}
has introduced \emph{multi-compatibility}, allowing multiple
compositions of global types using the gateway approach. 

Our aim in constructing the inference algorithm is
to prove that there always exists a global type
which can fully capture 
liveness properties of \emph{communicating processes}
(Theorem~\ref{thm:proc:completeness} and Corollary \ref{cor:completeness}).
We aim to demonstrate that the typability of the top-down system
is the same as the bottom-up system. 
The preceding works study the synthesis or composition
of \emph{protocols} (types or CFSMs), but do not
directly relate them to
\emph{properties of typed processes}. Our inference algorithm
can build the \emph{principal global type} (minimal with respect to
trace sets) without altering the semantics of local protocols.

\myparagraph{Top-down and bottom-up approaches}
Honda et al. \cite{HYC2016,Honda2008} 
propose the first 
endpoint projection (EPP) algorithm with
\emph{inductive plain merging} 
from a global type with a parallel composition to local types.
Later, to widen the set of well-formed global types,
\emph{inductive full merge} is proposed and implemented in many
tools and languages \cite{Y2024}.   
Scalas and Yoshida~\cite{Scalas2019} have discovered 
that a proof of type safety with 
inductive full merge in the literature was flawed.
Recent work \cite{revisited} has corrected their proofs, 
and proved that the EPP with inductive full
merging is in fact type sound. The EPP with \emph{coinductive plain merging}
was proposed in \cite{DBLP:journals/jar/TiroreBC25} and
proven sound using Rocq. 
The sound and complete EPP algorithm
with \emph{coinductive full merge} was proposed in
\cite{thien-nobuko-popl-25}, 
which we applied to design our synthesis algorithm,
implemented, and proved
complete with respect to liveness of communicating processes. 

A bottom-up approach which directly verifies a typing context
to prove safety of processes was first introduced in
\cite{Scalas2019}.
We have implemented the first complete bottom-up tool which 
infers a principal global type from a set of MPST processes,
implementing the global type inference in this paper and   
the local type inference in \cite{thien-nobuko-popl-25}, 
combined with the model-checking tool \texttt{mpstk} \cite{Scalas2019}. 

Dagnino et al. \cite{Dagnino2023} propose
a typing system of an asynchronous network (which consists of a
process similar to a local type and a queue) typed by a global type, and
develop a non-deterministic type inference system. 
Their well-formedness checking of global types is undecidable, which
is resolved by bounding the paths of global type trees.   

In the context of CFSMs, the top-down works \cite{Li2023,DBLP:conf/concur/MajumdarMSZ21,LiSWZ25OOPSLA}
have studied an \emph{implementability
(projectability) problem} for more 
expressive sender-driven global types in CFSMs from
various angles, see 
\cite{Li2025DecisionProblems,Stutz2023PowerOfChoice,Li2026} 
for a comprehensive summary of their works.
Recent work \cite{DiGiustoUrsoLozesPPDP2025} studies
realisability and complementability 
(a global type $\G$ admits a complementary one $\overline{\G}$
that gives all those behaviours that are not described by $\G$)
of $\G$ in peer-to-peer and synchronous CFSMs.

The completeness established by
\cite{Li2023,DiGiustoUrsoLozesPPDP2025,Stutz2023PowerOfChoice} is
with respect to \emph{realisability}, in the sense that no
realisable global type is rejected. This should not be confused with this paper's notion of completeness of the top-down system with respect to bottom-up \emph{typability} (\SEC{introduction}). The coinductive
full-merging EPP of \cite{thien-nobuko-popl-25} provides one direction
of our correspondence, and the construction of a global type from an
arbitrary safe-and-live local context (\SEC{sec:building}) provides
the other. This line of work on realisability works with communicating
finite-state machines and targets deadlock-freedom, which is not
downwards-closed under subtyping and so is not suitable for our
process syntax, which requires the type system to have a subsumption
rule.
The type system of Stutz and D'Osualdo \cite{StutzDOsualdo2025}
uses CFSMs as types directly and additionally handles session
interleaving and delegation, features absent from our setting.
Barbanera et al. \cite{Franco2024} propose
a \emph{corecursive} projection on \emph{coinductive} global and local
types, and prove
that corecursively projectable global types are automatically
balanced.

Castro-Perez et al. \cite{Castro2026} propose a synthetic approach
that defines a typing system based on an operational correspondence
between a process and a global labelled transition system, extending
their earlier works \cite{DBLP:conf/issta/FerreiraJ23,JY2020}. Their use of
the term ``synthetic'' differs from the synthesis discussed before.
They state that their framework is the first top-down system capable
of projecting and typing the examples in \cite[Fig.~4]{Scalas2019}.
In contrast, related projections appeared prior to \cite{Castro2026}:
the coinductive full merge in
\cite[Definition~4.23]{thien-nobuko-popl-25} and
the system in \cite{Franco2024} 
can project all examples
in \cite[Fig.~12]{Castro2026} (\cite[Fig.~4]{Scalas2019}), and the
corecursive projection in \cite{Franco2024} can project
\cite[Fig.~4(1,3,4)]{Scalas2019}. See ``CF'' in
Table~\ref{tab:example-global-types}
(\appref{app:example-global-types}) and
Table~\ref{tab:benchmarks}. Further, their recent article
\cite{Barbanera2026} has shown that their typing system
in \cite{Franco2024} makes it possible to
type all of the examples in \cite[Fig.~4(1,2,3,4)]{Scalas2019}.
Notice also that the term ``liveness'' is used in \cite{Castro2026} to
denote ``progress'' (deadlock-freedom in \cite[Def.~5.1]{Scalas2019}),
which is weaker than the liveness used in our paper and in other works~\cite{thien-nobuko-popl-25,lmcs2025,DBLP:conf/ecoop/LagaillardieNY22,BSYZ2022}.
For example, $\M[2]$ is \emph{not} live in our paper but live in 
\cite{Castro2026}. 

Our focus is not on global-type construction but on the \emph{typability problem of
endpoint processes}, comparing a specification-guided (top-down) approach against a model-checking (bottom-up) one.

\myparagraph{Conclusion}
This paper has proven that
the \emph{correctness-by-protocol-construction} approach (the top-down
system) offers exactly the same typability as the model-checking (bottom-up)
approach, contrary to the prevailing understanding in the session
type literature.  We have proven this via a novel inference
algorithm that produces a principal global type from any
live typing context.

This equivalence is at the level of expressiveness. In practice
the two methodologies remain complementary. The top-down system is
the natural choice when a global protocol is being designed or
already exists, since verification is substantially cheaper in
practice (\SEC{sec:implementation}) and $\G$ doubles as
documentation. The bottom-up system is the natural choice when
only implementations are available, although verification incurs
\PSPACE{} cost synchronously and is undecidable
asynchronously.

We have implemented four kinds of projection
algorithms, type-checking systems, and
global and local type inference algorithms, and tested
them on case studies from the literature.
The results in \cite{thien-nobuko-popl-25} have shown that
for most realistic protocols, the complexity of
the top-down approach remains semi-quadratic 
(inductive full-merging EPPs) or polynomial (coinductive EPPs),
but the safety and liveness
checks of a typing context $\ctx$
for the bottom-up approach are \PSPACE-complete 
in the size of typing contexts, hence typically exponential in the number of states.

Our benchmarks (Table~\ref{tab:benchmarks})
validate these results in practice.
We have tested the top-down approach with coinductive full-merging
projection across a broad range of examples from the literature,
including all benchmarks of~\cite{Scalas2019} and~\cite{Castro2026},
and verified our inference algorithm by reprojecting every synthesised
global type and checking subtyping against the original context.
The results confirm that the top-down approach offers the same
expressiveness as the bottom-up approach (CF projection succeeds on
every benchmark) while substantially outperforming the {\sf mpstk}
model checker on these examples.

In the asynchronous case, the top-down system is still practically
implementable when combined with synchronous subtyping from
\cite{Scalas2019,BHYZ2023,lmcs2025,DBLP:conf/ecoop/HouLY24,BSYZ2022}
or sound asynchronous subtyping
algorithms \cite{BravettiCLYZ19L,DBLP:journals/pacmpl/Castro-PerezY20,CYV2022,BocchiKMY24,BocchiK0025},
although the bottom-up system remains undecidable.
A fully sound and complete \emph{decidable} top-down
characterisation of asynchronous live bottom-up typability is
precluded by the undecidability of precise asynchronous subtyping
\cite{BravettiCZ17,LangeY17}, leaving sound-but-incomplete
subtyping algorithms as the practical route.
Future work includes decidability and typability in both systems
for richer MPSTs, including
sender-driven choice \cite{Li2025DecisionProblems,Stutz2023PowerOfChoice,DiGiustoUrsoLozesPPDP2025}
and mixed choice \cite{PetersY24}.

Our setting also assumes a \emph{single global session}, and so we do not
treat session delegation, dynamic session initialisation, or
interleaved sessions. With these features, the bottom-up MPST
system is known not to guarantee deadlock-freedom in general, and
hence not liveness. Restricted to a single session per typing
context, as in
\cite{Scalas2019,thien-nobuko-popl-25,DiGiustoUrsoLozesPPDP2025,Stutz2023PowerOfChoice},
the liveness characterisation established here should extend. A
systematic study of how these features interact with the precision
of top-down/bottom-up correspondences is a natural direction for
future work.

The principal global type $G_{\textsf{inf}}$ produced by our inference
algorithm is also of independent interest. Combined with existing
techniques for composition, decomposition, and hierarchical
specification of multiparty protocols
\cite{Barbanera2019JLAMP,Barbanera2022,Barbanera2023a,Barbanera2021,Barbanera2024,GY2023,Stolze2021},
one could obtain $G_{\textsf{inf}}$ from disjoint subsets of endpoints
and then compose the resulting global types $G^i_{\textsf{inf}}$ into a
single global type for the entire system. Such a distributed,
hierarchical workflow offers modularity and compositionality
not available to the bottom-up approach.

Our completeness result identifies three sufficient conditions for the
top-down system to match the bottom-up one with respect to liveness
(precise synchronous subtyping, balancedness, and coinductive
full-merging projection; \SEC{introduction}). Whether other
combinations can yield completeness for liveness or deadlock-freedom
remains open.

\smallskip

\myparagraph{Acknowledgements. }
We thank the reviewers for their helpful comments and suggestions.
The work is partially supported by 
NII MOU Grants,
EPSRC EP/T006544/2, EP/T014709/2, EP/Y005244/1,
EP/V000462/1, EP/X015955/1, EP/Z0005801/1, 
 Horizon EU TaRDIS 101093006 (UKRI No.~10066667), 
Advanced Research and Invention Agency~(ARIA), 
and a grant from the Simons Foundation. 

\section*{Data-Availability Statement}
\label{sec:dataavailability}
An artifact accompanying this paper is archived on Zenodo at
\url{https://doi.org/10.5281/zenodo.21237706}.
It is a Haskell (GHC 9.6.6) implementation of the toolchain
described in \SEC{sec:implementation}, packaged as a Docker image
with the benchmark harness that reproduces
Table~\ref{tab:benchmarks}.

\bibliography{misc/references} 

\ifnotsplit{\appendix
\renewcommand{\proofapp}[1]{}
\section{Auxiliary Material}

\subsection{Big-Step Evaluation of Expressions}
\label{app:expr-eval}

To make the premises $\eval{\e}{\val}$ in
Figure~\ref{fig:reduction_sessions} explicit, we record a standard
big-step semantics for closed expressions. Variables are eliminated by
communication substitution before evaluation, so the relation is only
used on closed expressions. We write $\sem{\diamond}$ for the usual
partial meta-level operation of a unary operator
$\diamond \in \{\pv{\neg}, \pv{\sqrt{~}}\}$, and
$\sem{\odot}$ for binary operators
$\odot \in \{\pv{\vee}, \pv{+}, \pv{-}, \pv{*}, \pv{/}, \pv{\le}\}$.

\[\begin{array}{c}
\inferrule*[left=\rulename{EVal}]
  {\phantom{\eval{\e}{\val}}}
  { \eval{\val}{\val} }
\qquad
\inferrule*[left=\rulename{ERand}]
  { \pv{f} \in [0,1] }
  { \eval{\pv{\mathtt{rand}()}}{\pv{f}} }
\\[4mm]
\inferrule*[left=\rulename{EUn}]
  { \eval{\e}{\val} \\ \vali = \sem{\diamond}(\val) }
  { \eval{\diamond\,\e}{\vali} }
\qquad
\inferrule*[left=\rulename{EBin}]
  { \eval{\e[1]}{\val[1]} \\ \eval{\e[2]}{\val[2]} \\
    \val = \sem{\odot}(\val[1], \val[2]) }
  { \eval{\e[1] \mathbin{\odot} \e[2]}{\val} }
\\[4mm]
\inferrule*[left=\rulename{EChoiceL}]
  { \eval{\e[1]}{\val} }
  { \eval{\e[1] \mathbin{\pv{\oplus}} \e[2]}{\val} }
\qquad
\inferrule*[left=\rulename{EChoiceR}]
  { \eval{\e[2]}{\val} }
  { \eval{\e[1] \mathbin{\pv{\oplus}} \e[2]}{\val} }
\end{array}\]

If one of the meta-level operations is undefined because the operand
sorts do not match, then no evaluation rule applies. In the typed
fragment this cannot occur.

\subsection{Multiparty Session Typing System}
\label{app:typing}
We give the full typing rules for the synchronous MPST calculus. 
The judgement $\Gamma \,\vdash\, \PP : \T$ states that, under term-variable assumptions in $\Gamma$, the process $\PP$ follows the protocol described by the local type $\T$. Rule \rulename{T0} assigns the inactive process $\inact$ the terminated type $\tend$, expressing that no further communication is possible. 
Rule \rulename{TOut} types an output: if the expression $\e$ has base sort $\B$ and the continuation $\PP$ follows $\T$, then sending $\e$ to participant $\pp$ behaves as $\tout{\pp}{\B}{\T}$, i.e., first outputs a $\B$-value to $\pp$ and then continues as $\T$.
Dually, \rulename{TIn} types an input: assuming a fresh variable $x{:}\B$ for the received value, the continuation $\PP$ follows $\T$, hence receiving from $\pp$ follows $\tin{\pp}{\B}{\T}$.
Internal choice and external choice mirror the selection/branching constructs. If the continuation after choosing $\lab$ follows $\T[j]$, then selecting $\lab$ towards $\pp$ follows $\tsel{\pp}{\lab: \T}$, which offers exactly one label. Note that internal choices offering multiple labels only arise from subtyping. 
Conversely, \rulename{TBra} types an external choice: if for every $i\in I$ the branch process $\PP[i]$ follows $\T[i]$, then offering to accept one of the labels $\lab[i]$ from $\pp$ follows $\tbrasub{\pp}{\lab[i]: \T[i]}{i\in I}$ and continues accordingly once a label is selected by the peer.
Rule \rulename{TIf} ensures control-flow does not alter the protocol: when the expression $\e$ has type $\tbool$ and both branches follow the same session type $\T$, the conditional as a whole follows $\T$. Combined with the subsumption rule \rulename{TSub}, this allows compatible but different types in each branch. 
Rule \rulename{TVar} retrieves the session type assigned to a process variable $X$ from the environment, while \rulename{TRec} ties the knot: if the body $\PP$ follows $\T$ under the assumption that $X$ also follows $\T$, then the recursive process $\pfix{X}\PP$ follows $\T$.
Finally, \rulename{TSub} provides subsumption: whenever $\PP$ follows $\T$ and $\T \subt \Ti$, the same process also follows the supertype $\Ti$, allowing safe widening of behaviours such as enlarging offered branches or narrowing accepted branches. 

\begin{definition}[Typing rules for processes]
\label{def:typing-processes}
The typing rules for processes are as follows.

\[
\begin{array}{c}
    \inferrule[\rulename{T0}]{}{\Gamma \,\vdash\, \inact \,:\, \tend} \quad
    \inferrule[\rulename{TOut}]{\Gamma \vdash \e : \B \\ \Gamma \,\vdash\, \PP \,:\, \T}{\Gamma \,\vdash\, \procout{\pp}{\e}{\PP} \,:\, \tout{\pp}{\B}{\T}} \quad
    \inferrule[\rulename{TIn}]{\Gamma, x{:}\B \,\vdash\, \PP \,:\, \T}{\Gamma \,\vdash\, \procin{\pp}{x}{\PP} \,:\, \tin{\pp}{\B}{\T}} \\[3mm]
    \inferrule[\rulename{TSel}]{\Gamma \,\vdash\, \PP \,:\, \T }{\Gamma \,\vdash\, \procsel{\pp}{\lab} \,\PP \,:\, \tsel{\pp}{\lab : \T}} \quad
    \inferrule[\rulename{TBra}]{\forall i \in I : \Gamma \,\vdash\, \PP[i] \,:\, \T[i] }{\Gamma \,\vdash\, \procbrasub{\pp}{\lab[i] : \PP[i]}{i \in I} \,:\, \tbrasub{\pp}{\lab[i] : \T[i]}{i \in I}} \\[3mm]
    \inferrule[\rulename{TIf}]{\Gamma \vdash \e : \tbool \\ \Gamma \,\vdash\, \PP[1] \,:\, \T \\ \Gamma \,\vdash\, \PP[2] \,:\, \T }{\Gamma \,\vdash\, \cond{\e}{\PP[1]}{\PP[2]} \,:\, \T} \\[3mm]
    \inferrule[\rulename{TVar}]{\Gamma(X) = \T }{\Gamma \,\vdash\, X \,:\, \T} \quad
    \inferrule[\rulename{TRec}]{\Gamma, X{:}\T \,\vdash\, \PP \,:\, \T }{\Gamma \,\vdash\, \pfix{X}\PP \,:\, \T} \quad
    \inferrule[\rulename{TSub}]{\Gamma \,\vdash\, \PP \,:\, \T \\ \T \subt \Ti }{\Gamma \,\vdash\, \PP \,:\, \Ti}
\end{array}
\]
\end{definition}

\section{Proofs}

\subsection{Inversion Lemmas}
\label{app:semantics-proofs}

\begin{restatable}[Inversion of One-Sided Context Reduction]{lemma}{inversionsafeonesidedreduction}
  \label{lemma:inversion-safe-one-sided-reduction}
Assume $\ctx \transt{\pp : \czeta} \ctxi$. Then:
\begin{enumerate}[leftmargin=*, label=\textbullet]
  \item If $\czeta = \qsend{\pq}{\B}$, then
    $\unfold{\ctx(\pp)} = \tout{\pq}{\B}{\ctxi(\pp)}$.
  \item If $\czeta = \qrecv{\pq}{\B}$, then
    $\unfold{\ctx(\pp)} = \tin{\pq}{\B}{\ctxi(\pp)}$.
  \item If $\czeta = \qsend{\pq}{\lab[j]}$, then there exist an index set $I$ and local types $\set{\T[i]}_{i \in I}$ such that
    $j \in I$,
    $\unfold{\ctx(\pp)} = \tselsub{\pq}{\lab[i] : \T[i]}{i \in I}$,
    and $\ctxi(\pp) = \T[j]$.
  \item If $\czeta = \qrecv{\pq}{\lab[j]}$, then there exist an index set $I$ and local types $\set{\T[i]}_{i \in I}$ such that
    $j \in I$,
    $\unfold{\ctx(\pp)} = \tbrasub{\pq}{\lab[i] : \T[i]}{i \in I}$,
    and $\ctxi(\pp) = \T[j]$.
\end{enumerate}
\end{restatable}

\begin{proof}[Proof of Lemma~\ref{lemma:inversion-safe-one-sided-reduction}]
Assume $\ctx \transt{\pp : \czeta} \ctxi$.
By inversion on rule \rulename{LTag}, there exist a context $\ctx_0$ and local
types $\T,\Ti$ such that
\[
\ctx = \ctx_0, \ptag{\pp}{\T},
\qquad
\ctxi = \ctx_0, \ptag{\pp}{\Ti},
\qquad
\T \transa{\czeta} \Ti.
\]
So it is enough to analyse the derivation of $\T \transa{\czeta} \Ti$.
We proceed by cases on $\czeta$, and in each case by induction on that derivation.

\begin{enumerate}
  \item $\czeta = \qsend{\pq}{\B}$.
  The last rule is either \rulename{LR1} or \rulename{LR5}.
  For \rulename{LR1}, $\T=\tout{\pq}{\B}{\Ti}$, so
  $\unfold{\ctx(\pp)}=\unfold{\T}=\tout{\pq}{\B}{\Ti}
  =\tout{\pq}{\B}{\ctxi(\pp)}$.
  For \rulename{LR5}, write $\T=\toc{\sfmu}\ty.\T[0]$ and use the induction hypothesis on
  the premise $\T[0]\sub{\toc{\sfmu}\ty.\T[0]}{\ty}\transa{\qsend{\pq}{\B}}\Ti$; then
  $\unfold{\T}=\unfold{\T[0]\sub{\toc{\sfmu}\ty.\T[0]}{\ty}}=
  \tout{\pq}{\B}{\Ti}$.

  \item $\czeta = \qrecv{\pq}{\B}$.
  Exactly symmetric to item 1, using \rulename{LR2} and \rulename{LR5}, yielding
  $\unfold{\ctx(\pp)}=\tin{\pq}{\B}{\ctxi(\pp)}$.

  \item $\czeta = \qsend{\pq}{\lab[j]}$.
  The last rule is either \rulename{LR3} or \rulename{LR5}.
  For \rulename{LR3}, there is an index set $I$ with $j\in I$ such that
  $\T=\tselsub{\pq}{\lab[i]:\T[i]}{i\in I}$ and $\Ti=\T[j]$.
  Hence
  $\unfold{\ctx(\pp)}=\tselsub{\pq}{\lab[i]:\T[i]}{i\in I}$ and
  $\ctxi(\pp)=\T[j]$.
  For \rulename{LR5}, write $\T=\toc{\sfmu}\ty.\T[0]$ and apply the induction hypothesis
  to the premise
  $\T[0]\sub{\toc{\sfmu}\ty.\T[0]}{\ty}\transa{\qsend{\pq}{\lab[j]}}\Ti$.

  \item $\czeta = \qrecv{\pq}{\lab[j]}$.
  Symmetric to item 3, using \rulename{LR4} and \rulename{LR5}. We obtain an
  index set $I$ with $j\in I$ such that
  $\unfold{\ctx(\pp)}=\tbrasub{\pq}{\lab[i]:\T[i]}{i\in I}$ and
  $\ctxi(\pp)=\T[j]$.
\end{enumerate}
\end{proof}

\begin{restatable}[Inversion of Two-Sided Context Reduction]{lemma}{inversionsafetwosidedreduction}
\label{lemma:semantics}
Assume
$\ctx \transi{\interaction{\pp}{\pq}{\obs}} \ctxi$. Then:
\begin{enumerate}[leftmargin=*, label=\textbullet]
  \item If $\obs = \B$, then
    $\unfold{\ctx(\pp)} = \tout{\pq}{\B}{\ctxi(\pp)}$ and
    $\unfold{\ctx(\pq)} = \tin{\pp}{\B}{\ctxi(\pq)}$.
  \item If $\obs = \lab[j]$, then there exist index sets $I, J$ and local types
    $\set{\T[i]}_{i \in I}, \set{\Ti[i]}_{i \in J}$ such that
    $j \in I \cap J$,
    $\unfold{\ctx(\pp)} = \tselsub{\pq}{\lab[i]:\T[i]}{i \in I}$,
    $\unfold{\ctx(\pq)} = \tbrasub{\pp}{\lab[i]:\Ti[i]}{i \in J}$,
    $\ctxi(\pp)=\T[j]$, and $\ctxi(\pq)=\Ti[j]$.
  \item If $\pr \notin \set{\pp,\pq}$, then
    $\ctx(\pr)=\ctxi(\pr)$.
\end{enumerate}
\end{restatable}

\begin{proof}[Proof of Lemma~\ref{lemma:semantics}]
Assume
$\ctx \transi{\interaction{\pp}{\pq}{\obs}} \ctxi$.
By inversion on \rulename{LEnv}, there exist contexts
$\ctx[1],\ctx[2],\ctxi[1],\ctxi[2]$ such that
\[
\ctx=\ctx[1],\ctx[2],\qquad
\ctxi=\ctxi[1],\ctxi[2],
\]
and
\[
\ctx[1] \transt{\pp : \qsend{\pq}{\obs}} \ctxi[1],
\qquad
\ctx[2] \transt{\pq : \qrecv{\pp}{\obs}} \ctxi[2].
\]

From each one-sided premise, invert \rulename{LTag} and re-apply \rulename{LTag}
with the unchanged participants of $\ctx$. This yields auxiliary one-sided
reductions from the whole context:
\[
\ctx \transt{\pp : \qsend{\pq}{\obs}} \ctx[p],
\qquad
\ctx \transt{\pq : \qrecv{\pp}{\obs}} \ctx[q],
\]
where $\ctx[p](\pp)=\ctxi(\pp)$, $\ctx[q](\pq)=\ctxi(\pq)$, and all other
participants are unchanged.

Apply Lemma~\ref{lemma:inversion-safe-one-sided-reduction} to both auxiliary
reductions.

\begin{itemize}
  \item If $\obs=\B$, we get
  $\unfold{\ctx(\pp)}=\tout{\pq}{\B}{\ctx[p](\pp)}$ and
  $\unfold{\ctx(\pq)}=\tin{\pp}{\B}{\ctx[q](\pq)}$, i.e.
  \[
  \unfold{\ctx(\pp)}=\tout{\pq}{\B}{\ctxi(\pp)},
  \qquad
  \unfold{\ctx(\pq)}=\tin{\pp}{\B}{\ctxi(\pq)}.
  \]
  \item If $\obs=\lab[j]$, the send-side inversion gives an index set $I$ and
  local types $\set{\T[i]}_{i\in I}$ with $j\in I$ such that
  \[
  \unfold{\ctx(\pp)}=\tselsub{\pq}{\lab[i]:\T[i]}{i\in I},
  \qquad
  \ctx[p](\pp)=\T[j],
  \]
  hence $\ctxi(\pp)=\T[j]$.
  The receive-side inversion gives an index set $J$ and local types
  $\set{\Ti[i]}_{i\in J}$ with $j\in J$ such that
  \[
  \unfold{\ctx(\pq)}=\tbrasub{\pp}{\lab[i]:\Ti[i]}{i\in J},
  \qquad
  \ctx[q](\pq)=\Ti[j],
  \]
  hence $\ctxi(\pq)=\Ti[j]$.
\end{itemize}

Finally, by the shape of \rulename{LEnv} and \rulename{LTag}, only $\pp$ and
$\pq$ are updated, so $\ctx(\pr)=\ctxi(\pr)$ for every
$\pr\notin\set{\pp,\pq}$.
\end{proof}

\begin{definition}[Projection of round-robin sequences]
\label{def:projection-round-robin}
Let $\sigma = ((\ctx[j], \ppi[j]) \rrtrans{\inter[j]} (\ctx[j{+}1], \ppi[j{+}1]))_{j \in I}$
be a (finite or infinite) round-robin reduction sequence.
The \emph{projection} of $\sigma$ is the context reduction sequence
$\pi(\sigma) = (\ctx[j] \transi{\inter[j]} \ctx[j{+}1])_{j \in I}$
obtained by forgetting the prioritised participant at each step.
This is well-defined because $(\ctx, \pp[i]) \rrtrans{\inter} (\ctxi, \pp[k])$
implies $\ctx \transi{\inter} \ctxi$ by the definition of round-robin.
\end{definition}

\begin{restatable}[Projection of Round-Robin Sequences is Fair]{lemma}{projectionofroundrobinsequencesisfair}
  \label{lemma:projection-of-round-robin-sequences-is-fair}
  Any maximal reduction sequence of the round-robin relation $((\ctx[j], \ppi[j]) \rrtrans{\inter[j]} (\ctx[j+1], \ppi[j+1]))_{j \in I}$,
  where $\ppi[j]$ denotes the prioritised participant at step~$j$,
  projects (Definition~\ref{def:projection-round-robin}) onto a fair reduction sequence.
\end{restatable}

\begin{proof}[Proof of Lemma~\ref{lemma:projection-of-round-robin-sequences-is-fair}]
Consider a maximal round-robin sequence
$((\ctx[j], \ppi[j]) \rrtrans{\inter[j]} (\ctx[j+1], \ppi[j+1]))_{j \in I}$,
and its projection (Definition~\ref{def:projection-round-robin})
$\pi(\sigma) = (\ctx[j] \transi{\inter[j]} \ctx[j+1])_{j \in I}$.
We show this projected sequence is fair.

Fix any index $k \in \overline{I}$ and assume an enabled interaction
$\ctx[k] \transi{\pp\pq}$, with sender $\pp$ and receiver $\pq$.
If $k$ is the final index of a finite maximal
sequence, then $\ctx[k]$ is stuck and no interaction is enabled, so
there is nothing to show; hence assume $k \in I$.
\emph{Persistence.}
Since the interaction is enabled,
$\unfold{\ctx[k](\pp)}$ is an output or selection towards $\pq$, and
$\unfold{\ctx[k](\pq)}$ is the dual input or branching from~$\pp$.
These heads pin the partners, so $\pp$ can only interact as sender
with $\pq$, and $\pq$ can only interact as receiver from~$\pp$;
any interaction involving either is thus an interaction between
$\pp$ and~$\pq$.
Hence, as long as no $\inter[j]$ with $j \ge k$ is an interaction
between $\pp$ and $\pq$, Lemma~\ref{lemma:semantics}(3) gives
$\ctx[j+1](\pp) = \ctx[j](\pp)$ and
$\ctx[j+1](\pq) = \ctx[j](\pq)$, so the same interaction remains
enabled, $\pp \in \mathsf{enabled}(\ctx[j+1])$, and the maximal
sequence does not stop at~$j+1$.

\emph{Scheduling.}
Write $\pp = \pp[m]$, let $\ppi[j] = \pp[i_j]$, and let
$d_j = (m - i_j) \bmod n$ be the cyclic distance from the priority
to~$\pp$.
Each round-robin step takes the sender
$\pp[j'] = \textsf{next-active}(\ctx[j], \ppi[j])$ and sets
$\ppi[j+1] = \pp[(j'+1)\bmod n]$.
While $\pp$ is enabled but not selected, minimality of
$\textsf{next-active}$ gives $(j' - i_j) \bmod n \le d_j$, strictly
since $\pp[j'] \neq \pp[m]$, and hence
$d_{j+1} = d_j - \big((j' - i_j) \bmod n\big) - 1 < d_j$.
Since $d_k \le n-1$ and $d_j$ strictly decreases, the round-robin
selects $\pp$ as sender within at most $n-1$ steps from~$k$.
By persistence $\pp$ is still enabled with its head towards $\pq$,
so the transition taken is
$\inter[j] = \interaction{\pp}{\pq}{\obs}$ for some $j \ge k$ and
some~$\obs$, as fairness requires.
Therefore the projection is fair.
\end{proof}

\begin{lemma}[Communication on Fair Live Paths]
\label{lemma:fair-live-communication}
Let $(\ctx[i] \transi{\inter[i]} \ctx[i+1])_{i \in I}$ be a fair and
live maximal reduction sequence from a context $\ctx[0]$ with
$\safe(\ctx[0])$, and suppose $\ctx[i] \transt{\pp : \czeta}$ for
some $i \in \overline{I}$ and action $\czeta$. Then there exists
$n \in I$ with $n \ge i$ such that $\inter[n]$ involves $\pp$.
\end{lemma}
\begin{proof}
Suppose no $\inter[n]$ with $n \ge i$ involves $\pp$.
Then $\ctx[k](\pp) = \ctx[i](\pp)$ for all positions $k \ge i$
(Lemma~\ref{lemma:semantics}(3)), so
$\ctx[k] \transt{\pp : \czeta}$ at every position $k \ge i$.
The clause of Definition~\ref{def:live-paths} matching $\czeta$
yields $j \in I$ with $j \ge i$ at which a complementary one-sided
action between $\pp$ and the partner named by $\czeta$ is enabled.
Since $\ctx[j]$ is reachable from $\ctx[0]$, safety
(Definition~\ref{def:safe-contexts}) turns the two matching
one-sided actions into a two-sided transition
$\ctx[j] \transi{\pr \role{s}}$ with $\pp \in \set{\pr, \role{s}}$.
By fairness, $\inter[n] = \interaction{\pr}{\role{s}}{\obs}$ for
some $n \in I$ with $n \ge j$ and some $\obs$, so $\inter[n]$
involves $\pp$, a contradiction.
\end{proof}

\subsection{Proofs from Section~\ref{sec:building}}
\label{app:building-proofs}

\subsubsection{Combined Relation Definition}

\begin{definition}[Composition of Projection and Subtyping]
\label{def:composition-projection-supertype}
We coinductively define a combined relation incorporating both projection and subtyping by the following rules.

\centerline{
\small
$
\begin{array}{c}
    \cinferruleleft[\rulename{SPR1}]{ \G \cofullprojsub{p} \T}{ \Gvt{\pp}{\pq}{\B}{\G}  \cofullprojsub{p} \tout{\pq}{\B}{\T}} \quad
    \cinferruleleft[\rulename{SPR2}]{ \G \cofullprojsub{p} \T}{ \Gvt{\pq}{\pp}{\B}{\G}  \cofullprojsub{p} \tin{\pq}{\B}{\T}} \quad
    \cinferruleleft[\rulename{SPR3}]{ \G \cofullprojsub{p} \T \\ \pp \notin \set{\pq, \pr} }{ \Gvt{\pq}{\pr}{\B}{\G}  \cofullprojsub{p} \T} \\[3mm]
    \cinferruleleft[\rulename{SPR4}]{ \forall i \in I : \G[i] \cofullprojsub{p} \T[i]}{ \GvtPair{\pp}{\pq}{\lab[i] : \G[i]}{i \in I}  \cofullprojsub{p} \tselsub{\pq}{\lab[i] : \T[i]}{i \in I}} \quad
    \cinferruleleft[\rulename{SPR5}]{ \forall j \in J : \G[j] \cofullprojsub{p} \T[j] \\ J \subseteq I}{ \GvtPair{\pq}{\pp}{\lab[j] : \G[j]}{j \in J}  \cofullprojsub{p} \tbrasub{\pq}{\lab[i] : \T[i]}{i \in I}} \\[3mm]
    \cinferruleleft[\rulename{SPR6}]{ \forall i \in I : \G[i] \cofullprojsub{p} \T \\ \pp \notin \set{\pq, \pr}}{ \GvtPair{\pq}{\pr}{\lab[i] : \G[i]}{i \in I}  \cofullprojsub{p} \T} \quad
    \cinferruleleft[\rulename{SPR9}]{ \G \sub{\toc{\sfmu} \ty . \G}{\ty} \cofullprojsub{p} \T \\ \pp \in \pt{\G}}{ \toc{\sfmu} \ty . \G \cofullprojsub{p} \T} \\[3mm]
    \cinferruleleft[\rulename{SPR10}]{ \G \cofullprojsub{p} \T \sub{\toc{\sfmu} \ty . \T}{\ty}}{ \G \cofullprojsub{p} \toc{\sfmu} \ty . \T} \quad
    \cinferruleleft[\rulename{SPR11}]{ }{ \gend \cofullprojsub{p} \tend} \quad
    \cinferruleleft[\rulename{SPR12}]{ \pp \not \in \pt{\G} }{ \toc{\sfmu} \ty . \G \cofullprojsub{p} \tend}
\end{array}
$
}
\end{definition}

\subsubsection{Inversion Lemmas}

\begin{lemma}[Subtyping is invariant under unfolding]
    \label{lemma:unfold-compatible-with-subtyping}
    We have that $\T \subt \unfold{\T}$ and $\unfold{\T} \subt \T$.
\end{lemma}
\begin{proof}
By coinduction.
Define the candidate relation
$R = \set{(\T, \Ti) \mid \unfold{\T} = \unfold{\Ti}}$.
Since $\unfold{}$ is idempotent, $(\T, \unfold{\T}) \in R$ and
$(\unfold{\T}, \T) \in R$ for every~$\T$.
It suffices to show $R$ is consistent with the subtyping rules,
i.e.\ for every $(\T, \Ti) \in R$ there is a subtyping rule whose
conclusion matches $\T \subt \Ti$ and whose premises lie in~$R$.

Take $(\T, \Ti) \in R$ with $\unfold{\T} = \unfold{\Ti}$.

\emph{If $\T$ is a $\toc{\sfmu}$-type}, say $\T = \toc{\sfmu} \ty . \T[0]$.
Rule \rulename{S5} applies: the premise requires
$(\T[0]\sub{\toc{\sfmu} \ty . \T[0]}{\ty}, \Ti) \in R$.
Since $\unfold{\T[0]\sub{\toc{\sfmu} \ty . \T[0]}{\ty}} = \unfold{\T} = \unfold{\Ti}$,
the premise is in~$R$.

\emph{If $\Ti$ is a $\toc{\sfmu}$-type}, say $\Ti = \toc{\sfmu} \ty . \Ti[0]$.
Rule \rulename{S6} applies: the premise requires
$(\T, \Ti[0]\sub{\toc{\sfmu} \ty . \Ti[0]}{\ty}) \in R$.
Since $\unfold{\Ti[0]\sub{\toc{\sfmu} \ty . \Ti[0]}{\ty}} = \unfold{\Ti} = \unfold{\T}$,
the premise is in~$R$.

\emph{If neither $\T$ nor $\Ti$ is a $\toc{\sfmu}$-type},
then $\T = \unfold{\T} = \unfold{\Ti} = \Ti$.
The appropriate structural rule applies:
\rulename{S1} for $\tout{}{}{}$,
\rulename{S2} for $\tin{}{}{}$,
\rulename{S3} for $\tsel{}{}$,
\rulename{S4} for $\tbra{}{}$, or
\rulename{S7} for $\tend$.
In each case, the premises are pairs $(\T[i], \T[i])$,
which lie in $R$ since $\unfold{\T[i]} = \unfold{\T[i]}$.
\end{proof}

\begin{lemma}[Inversion of Subtyping]
    \label{lemma:inversion-subtyping}
    If $\T \subt \Ti$ then:
    \begin{enumerate}[leftmargin=*, label=\textbullet]
      \item If $\Ti = \tout{\pq}{\B}{\Tii}$, then
        $\unfold{\T} = \tout{\pq}{\B}{\T[0]}$ and $\T[0] \subt \Tii$.
      \item If $\Ti = \tin{\pq}{\B}{\Tii}$, then
        $\unfold{\T} = \tin{\pq}{\B}{\T[0]}$ and $\T[0] \subt \Tii$.
      \item If $\Ti = \tselsub{\pq}{\lab[i] : \Tii[i]}{i \in I}$, then
        there exists $J \subseteq I$ such that
        $\unfold{\T} = \tselsub{\pq}{\lab[j] : \T[j]}{j \in J}$, and
        $\T[j] \subt \Tii[j]$ for all $j \in J$.
      \item If $\Ti = \tbrasub{\pq}{\lab[i] : \Tii[i]}{i \in I}$, then
        there exists $J \supseteq I$ such that
        $\unfold{\T} = \tbrasub{\pq}{\lab[j] : \T[j]}{j \in J}$, and
        $\T[i] \subt \Tii[i]$ for all $i \in I$.
      \item If $\Ti = \tend$, then
        $\unfold{\T} = \tend$.
    \end{enumerate}
\end{lemma}
\begin{proof}
Assume $\T \subt \Ti$.
By Lemma~\ref{lemma:unfold-compatible-with-subtyping},
$\unfold{\T} \subt \T \subt \Ti$, so $\unfold{\T} \subt \Ti$.
Since $\unfold{\T}$ is not a $\toc{\sfmu}$-type, rule \rulename{S5}
does not apply, and the derivation of
$\unfold{\T} \subt \Ti$ must proceed by one of the
non-unfolding rules (possibly applying \rulename{S6} to~$\Ti$ first,
but $\Ti$ is given in constructor form, so \rulename{S6} does not
apply either).
In each case the matching rule forces the head constructor of
$\unfold{\T}$:
\begin{enumerate}[leftmargin=*, label=\textbullet]
  \item $\Ti = \tout{\pq}{\B}{\Tii}$: only \rulename{S1} applies,
    so $\unfold{\T} = \tout{\pq}{\B}{\T[0]}$ with $\T[0] \subt \Tii$.
  \item $\Ti = \tin{\pq}{\B}{\Tii}$: only \rulename{S2} applies,
    so $\unfold{\T} = \tin{\pq}{\B}{\T[0]}$ with $\T[0] \subt \Tii$.
  \item $\Ti = \tselsub{\pq}{\lab[i] : \Tii[i]}{i \in I}$:
    only \rulename{S3} applies,
    so $\unfold{\T} = \tselsub{\pq}{\lab[j] : \T[j]}{j \in J}$
    with $J \subseteq I$ and $\T[j] \subt \Tii[j]$ for all $j \in J$.
  \item $\Ti = \tbrasub{\pq}{\lab[i] : \Tii[i]}{i \in I}$:
    only \rulename{S4} applies,
    so $\unfold{\T} = \tbrasub{\pq}{\lab[j] : \T[j]}{j \in J}$
    with $J \supseteq I$ and $\T[i] \subt \Tii[i]$ for all $i \in I$.
  \item $\Ti = \tend$: only \rulename{S7} applies, so $\unfold{\T} = \tend$.
\end{enumerate}
\end{proof}

\begin{restatable}[Inversion of Global Type Reduction]{lemma}{inversionglobalreduction}
\label{lemma:inversion-global-reduction}
Assume $\G \transg{\interaction{\pp}{\pq}{\obs}} \Gi$. Then one of the following holds:
\begin{enumerate}[leftmargin=*, label=\textbullet]
  \item $\obs = \B$,
    $\unfold{\G} = \Gvt{\pp}{\pq}{\B}{\G[0]}$, and $\Gi = \G[0]$.
  \item $\obs = \lab[j]$, there exists an index set $I$ with $j \in I$ such that
    $\unfold{\G} = \GvtPair{\pp}{\pq}{\lab[i] : \G[i]}{i \in I}$
    and $\Gi = \G[j]$.
  \item $\unfold{\G} = \Gvt{\pr}{\role{s}}{\B}{\G[0]}$ with
    $\set{\pp,\pq} \cap \set{\pr,\role{s}} = \emptyset$,
    $\G[0] \transg{\interaction{\pp}{\pq}{\obs}} \Gi[0]$,
    and $\Gi = \Gvt{\pr}{\role{s}}{\B}{\Gi[0]}$.
  \item $\unfold{\G} = \GvtPair{\pr}{\role{s}}{\lab[i] : \G[i]}{i \in I}$ with
    $\set{\pp,\pq} \cap \set{\pr,\role{s}} = \emptyset$,
    $\G[i] \transg{\interaction{\pp}{\pq}{\obs}} \Gi[i]$ for all $i \in I$,
    and $\Gi = \GvtPair{\pr}{\role{s}}{\lab[i] : \Gi[i]}{i \in I}$.
\end{enumerate}
\end{restatable}

\begin{proof}[Proof of Lemma~\ref{lemma:inversion-global-reduction}]
By induction on the derivation of $\G \transg{\interaction{\pp}{\pq}{\obs}} \Gi$.
\begin{itemize}
  \item \rulename{GR1}: $\G = \Gvt{\pp}{\pq}{\B}{\G[0]}$, $\obs = \B$, $\Gi = \G[0]$.
    Since $\G$ is not a $\toc{\sfmu}$-type, $\unfold{\G} = \G$. Case~1.
  \item \rulename{GR4}: $\G = \GvtPair{\pp}{\pq}{\lab[i] : \G[i]}{i \in I}$,
    $\obs = \lab[j]$, $j \in I$, $\Gi = \G[j]$.
    Since $\G$ is not a $\toc{\sfmu}$-type, $\unfold{\G} = \G$. Case~2.
  \item \rulename{GR3}: $\G = \toc{\sfmu} \ty . \G[1]$ and
    $\G[1]\sub{\toc{\sfmu} \ty . \G[1]}{\ty}
    \transg{\interaction{\pp}{\pq}{\obs}} \Gi$.
    By the induction hypothesis, one of cases~1--4 holds for
    $\G[1]\sub{\toc{\sfmu} \ty . \G[1]}{\ty}$.
    Since $\unfold{\G} = \unfold{\G[1]\sub{\toc{\sfmu} \ty . \G[1]}{\ty}}$,
    the same case holds for~$\G$.
  \item \rulename{GR2}:
    $\G = \Gvt{\pr}{\role{s}}{\B}{\G[0]}$,
    $\set{\pp,\pq} \cap \set{\pr,\role{s}} = \emptyset$,
    $\G[0] \transg{\interaction{\pp}{\pq}{\obs}} \Gi[0]$,
    $\Gi = \Gvt{\pr}{\role{s}}{\B}{\Gi[0]}$.
    Since $\G$ is not a $\toc{\sfmu}$-type, $\unfold{\G} = \G$. Case~3.
  \item \rulename{GR5}:
    $\G = \GvtPair{\pr}{\role{s}}{\lab[i] : \G[i]}{i \in I}$,
    $\set{\pp,\pq} \cap \set{\pr,\role{s}} = \emptyset$,
    $\G[i] \transg{\interaction{\pp}{\pq}{\obs}} \Gi[i]$ for all $i \in I$,
    $\Gi = \GvtPair{\pr}{\role{s}}{\lab[i] : \Gi[i]}{i \in I}$.
    Since $\G$ is not a $\toc{\sfmu}$-type, $\unfold{\G} = \G$. Case~4. \qedhere
\end{itemize}
\end{proof}

\begin{restatable}[Inversion of the Combined Relation]{lemma}{inversioncombinedrelation}
\label{lemma:inversion-combined-relation}
Assume $\G \cofullprojsub{p} \T$. Then:
\begin{enumerate}[leftmargin=*, label=\textbullet]
  \item If $\unfold{\G} = \Gvt{\pp}{\pq}{\B}{\G[0]}$, then
    $\unfold{\T} = \tout{\pq}{\B}{\T[0]}$ and $\G[0] \cofullprojsub{p} \T[0]$.
  \item If $\unfold{\G} = \Gvt{\pq}{\pp}{\B}{\G[0]}$, then
    $\unfold{\T} = \tin{\pq}{\B}{\T[0]}$ and $\G[0] \cofullprojsub{p} \T[0]$.
  \item If $\unfold{\G} = \GvtPair{\pp}{\pq}{\lab[i] : \G[i]}{i \in I}$, then
    $\unfold{\T} = \tselsub{\pq}{\lab[i] : \T[i]}{i \in I}$ and
    $\G[i] \cofullprojsub{p} \T[i]$ for all $i \in I$.
  \item If $\unfold{\G} = \GvtPair{\pq}{\pp}{\lab[j] : \G[j]}{j \in J}$, then there exists $I \supseteq J$ such that
    $\unfold{\T} = \tbrasub{\pq}{\lab[i] : \T[i]}{i \in I}$ and
    $\G[j] \cofullprojsub{p} \T[j]$ for all $j \in J$.
  \item If $\unfold{\G} = \Gvt{\pq}{\pr}{\B}{\G[0]}$ with $\pp \notin \set{\pq, \pr}$, then
    $\G[0] \cofullprojsub{p} \T$.
  \item If $\unfold{\G} = \GvtPair{\pq}{\pr}{\lab[i] : \G[i]}{i \in I}$ with $\pp \notin \set{\pq, \pr}$, then
    $\G[i] \cofullprojsub{p} \T$ for all $i \in I$.
  \item If $\unfold{\G} = \gend$, then
    $\unfold{\T} = \tend$.
\end{enumerate}
\end{restatable}

\begin{proof}[Proof of Lemma~\ref{lemma:inversion-combined-relation}]
Since $\cofullprojsub{p}$ is defined coinductively (as a greatest fixpoint),
every pair $(\G, \T) \in {}\cofullprojsub{p}$ is justified by a rule
application whose premises also lie in $\cofullprojsub{p}$.
We proceed by induction on the total number of $\toc{\sfmu}$-unfolding steps
from $\G$ to $\unfold{\G}$ plus from $\T$ to $\unfold{\T}$
(well-founded by guardedness of types).

\emph{If $\G$ is a $\toc{\sfmu}$-type}, say $\G = \toc{\sfmu} \ty . \G[1]$.
The rule justifying $(\G, \T) \in {}\cofullprojsub{p}$ is
\rulename{SPR9} (if $\pp \in \pt{\G[1]}$),
\rulename{SPR12} (if $\pp \notin \pt{\G[1]}$), or, when $\T$ is also
a $\toc{\sfmu}$-type, possibly \rulename{SPR10}; in the last case the
premise unfolds $\T$ by one step and the induction hypothesis applies
exactly as in the next paragraph.
For \rulename{SPR12}: $\T = \tend$ and $\pp \notin \pt{\G}$.
Since $\pp$ does not appear as a participant in $\unfold{\G}$,
cases~1--4 cannot arise; for cases~5--6 the conclusion
$\G[0] \cofullprojsub{p} \T$ (resp.\ $\G[i] \cofullprojsub{p} \T$ for all~$i$)
holds because $\pp \notin \pt{\G}$ implies $\pp \notin \pt{\G[0]}$
(resp.\ $\pp \notin \pt{\G[i]}$), and a straightforward coinductive argument
shows $\Gi \cofullprojsub{p} \tend$ whenever $\pp \notin \pt{\Gi}$;
case~7 gives $\unfold{\T} = \tend$ immediately.
For \rulename{SPR9}: the premise
$\G[1]\sub{\toc{\sfmu} \ty . \G[1]}{\ty} \cofullprojsub{p} \T$
holds, and $\unfold{\G} = \unfold{\G[1]\sub{\toc{\sfmu} \ty . \G[1]}{\ty}}$
with strictly fewer unfolding steps on the $\G$-side.
The result follows by the induction hypothesis.

\emph{If $\G$ is not a $\toc{\sfmu}$-type but $\T$ is},
say $\T = \toc{\sfmu} \ty . \T[1]$.
The applicable rules are \rulename{SPR10}, and \rulename{SPR3} or
\rulename{SPR6} when the head of $\G$ does not involve $\pp$; for the
latter two the premises are exactly the sub-relations required by
cases~5--6, and cases~1--4 and~7 do not arise.
For \rulename{SPR10}, the premise is
$\G \cofullprojsub{p} \T[1]\sub{\toc{\sfmu} \ty . \T[1]}{\ty}$.
Since $\unfold{\T} = \unfold{\T[1]\sub{\toc{\sfmu} \ty . \T[1]}{\ty}}$,
the induction hypothesis (strictly fewer unfolding steps on the $\T$-side)
gives the desired conclusions about $\unfold{\T}$ and the sub-relations
for cases~1--4 and~7.
For cases~5--6, the hypothesis gives
$\G[0] \cofullprojsub{p} \T[1]\sub{\toc{\sfmu} \ty . \T[1]}{\ty}$
(resp.\ $\G[i] \cofullprojsub{p} \T[1]\sub{\toc{\sfmu} \ty . \T[1]}{\ty}$ for all~$i$).
Applying \rulename{SPR10} yields
$\G[0] \cofullprojsub{p} \T$
(resp.\ $\G[i] \cofullprojsub{p} \T$).

\emph{If neither $\G$ nor $\T$ is a $\toc{\sfmu}$-type},
then $\G = \unfold{\G}$ and $\T = \unfold{\T}$.
Rules \rulename{SPR9}, \rulename{SPR10}, and \rulename{SPR12} do not apply.
The remaining rules \rulename{SPR1}--\rulename{SPR6} and \rulename{SPR11}
each match exactly one structural form of $\G$ for a given~$\pp$:
\rulename{SPR1} matches case~1,
\rulename{SPR2} matches case~2,
\rulename{SPR4} matches case~3,
\rulename{SPR5} matches case~4,
\rulename{SPR3} matches case~5,
\rulename{SPR6} matches case~6,
and \rulename{SPR11} matches case~7.
In each case the premises of the rule directly yield the stated conclusions.
\end{proof}

\subsubsection{Widening Subtyping, Merging, and the Combined Relation}
\label{app:widening-combined}

The results of this subsection connect the coinductive projection
$\cofullproj{p}$ of Definition~\ref{def:coind-proj} to the combined
relation $\cofullprojsub{p}$: full merge widens branchings, and the
combined relation absorbs exactly this widening
(rules \rulename{SPR5} and \rulename{SPR6}).
We first isolate the fragment of subtyping generated by merging.

\begin{definition}[Widening Subtyping]
\label{def:widening-subtyping}
The \emph{widening subtyping} relation $\subtw$ is defined
coinductively by the rules of Definition~\ref{def:coind-subtyping},
with rule \rulename{S3} restricted to equal index sets ($I = J$).
Thus $\T \subtw \Ti$ allows $\T$ to offer additional \emph{branches}
(rule \rulename{S4}) but not to omit \emph{selections}.
Every rule of $\subtw$ is a rule of $\subt$, so
${\subtw} \subseteq {\subt}$.
Moreover, the candidate relation used in the proof of
Lemma~\ref{lemma:unfold-compatible-with-subtyping} only invokes
structural rules with equal index sets, so that proof also shows
$\T \subtw \Ti$ whenever $\unfold{\T} = \unfold{\Ti}$;
in particular $\subtw$ is reflexive.
\end{definition}

\begin{lemma}[Unfolding Replacement]
\label{lemma:unfold-replacement}
Let $\unfold{\T} = \unfold{\Ti}$. Then:
(1) $\T \subt \Tii$ iff $\Ti \subt \Tii$, and $\Tii \subt \T$ iff
$\Tii \subt \Ti$; similarly for $\subtw$;
(2) $\G \cofullproj{p} \T$ iff $\G \cofullproj{p} \Ti$, and
$\G \cofullprojsub{p} \T$ iff $\G \cofullprojsub{p} \Ti$;
(3) $\comergefull{S}{\T}$ iff $\comergefull{S}{\Ti}$, and
$\comergefull{S \cup \set{\T}}{\Tii}$ iff
$\comergefull{S \cup \set{\Ti}}{\Tii}$.
In particular, each relation is invariant under replacing the
indicated argument by its unfolding.
\end{lemma}
\begin{proof}
We prove the first claim of~(1); the statement is symmetric in $\T$
and $\Ti$, and the remaining claims are proved identically, with
rules \rulename{S6}, \rulename{PR10},
\rulename{SPR10}, \rulename{M7}, and \rulename{M6} supplying the
unfolding steps on the corresponding side (rules that do not inspect
the replaced argument recurse directly on their premises).
Let $R = \set{(\T[0], \Tii) \mid \exists \Ti[0].\
\Ti[0] \subt \Tii \text{ and }
\unfold{\T[0]} = \unfold{\Ti[0]}}$; we show $R$ is consistent with the
subtyping rules, so that $R \subseteq {\subt}$; since
$(\T, \Tii) \in R$ via $\Ti$, this suffices.
Take $(\T[0], \Tii) \in R$ via $\Ti[0]$.
If $\T[0] = \toc{\sfmu}\ty.\T[1]$, rule \rulename{S5} applies with premise
$(\T[1]\sub{\toc{\sfmu}\ty.\T[1]}{\ty}, \Tii) \in R$ (same witness~$\Ti[0]$).
Otherwise $\T[0]$ is in constructor form, so
$\T[0] = \unfold{\T[0]} = \unfold{\Ti[0]}$.
By induction on the number of leading $\toc{\sfmu}$-binders of~$\Ti[0]$
plus those of~$\Tii$ (finite by guardedness): if the judgement
$\Ti[0] \subt \Tii$ is justified by \rulename{S5}, its premise peels
one binder of $\Ti[0]$; if by \rulename{S6}, its premise peels one
binder of $\Tii$, which \rulename{S6} restores after applying the
induction hypothesis. Iterating yields $\unfold{\Ti[0]} \subt \Tii$,
i.e.\ $\T[0] \subt \Tii$. This judgement is justified by some rule whose premises
lie in ${\subt} \subseteq R$.
\end{proof}

\begin{lemma}[Transitivity]
\label{lemma:subtyping-transitive}
The relations $\subt$ and $\subtw$ are transitive.
\end{lemma}
\begin{proof}
We treat $\subt$; the argument for $\subtw$ is identical.
Let $R = \set{(\T, \Tii) \mid \exists \Ti.\
\T \subt \Ti \subt \Tii}$; we show $R$ is consistent.
Take such a triple; by Lemma~\ref{lemma:unfold-replacement} we may
assume $\Ti$ is in constructor form.
If $\T = \toc{\sfmu}\ty.\T[0]$, rule \rulename{S5} applies with
premise in $R$ (by Lemma~\ref{lemma:unfold-replacement}(1),
$\T[0]\sub{\toc{\sfmu}\ty.\T[0]}{\ty} \subt \Ti$).
If $\Tii = \toc{\sfmu}\ty.\Tii[0]$, rule \rulename{S6} applies
symmetrically.
Otherwise all three types are in constructor form, and the
judgements $\T \subt \Ti$ and $\Ti \subt \Tii$ are justified by
structural rules on the same head constructor, which compose:
for \rulename{S1}, \rulename{S2}, \rulename{S7} directly;
for \rulename{S3} the index sets satisfy
$I_{\T} \subseteq I_{\Ti} \subseteq I_{\Tii}$ and the premises chain
through $\Ti$ for each $i \in I_{\T}$;
for \rulename{S4} dually, for each $i \in I_{\Tii}$.
In each case the resulting premises lie in~$R$.
\end{proof}

\begin{lemma}[Inversion of Subtyping at the Subtype]
\label{lemma:inversion-subtyping-sub}
If $\T \subt \Ti$ then:
\begin{enumerate}[leftmargin=*, label=\textbullet]
  \item If $\unfold{\T} = \tout{\pq}{\B}{\T[0]}$, then
    $\unfold{\Ti} = \tout{\pq}{\B}{\Tii[0]}$ and $\T[0] \subt \Tii[0]$.
  \item If $\unfold{\T} = \tin{\pq}{\B}{\T[0]}$, then
    $\unfold{\Ti} = \tin{\pq}{\B}{\Tii[0]}$ and $\T[0] \subt \Tii[0]$.
  \item If $\unfold{\T} = \tselsub{\pq}{\lab[i] : \T[i]}{i \in I}$, then
    there exists $I' \supseteq I$ such that
    $\unfold{\Ti} = \tselsub{\pq}{\lab[i] : \Tii[i]}{i \in I'}$, and
    $\T[i] \subt \Tii[i]$ for all $i \in I$.
  \item If $\unfold{\T} = \tbrasub{\pq}{\lab[i] : \T[i]}{i \in I}$, then
    there exists $I' \subseteq I$ such that
    $\unfold{\Ti} = \tbrasub{\pq}{\lab[i] : \Tii[i]}{i \in I'}$, and
    $\T[i] \subt \Tii[i]$ for all $i \in I'$.
  \item If $\unfold{\T} = \tend$, then $\unfold{\Ti} = \tend$.
\end{enumerate}
The same statements hold for $\subtw$ in place of $\subt$, with
$I' = I$ in clause~3.
Dually, in the style of Lemma~\ref{lemma:inversion-subtyping}, if
$\T \subtw \Ti$ and $\unfold{\Ti}$ is given, then $\unfold{\T}$ has
the same head constructor, with equal index set for selections and
index set $\supseteq$ for branchings, and $\subtw$-related
continuations on the labels of $\unfold{\Ti}$.
\end{lemma}
\begin{proof}
By Lemma~\ref{lemma:unfold-replacement},
$\unfold{\T} \subt \unfold{\Ti}$ (replace both arguments).
Both sides are in constructor form, so the judgement is justified by
a structural rule, which is determined by the head constructor of
$\unfold{\T}$: \rulename{S1} for clause~1, \rulename{S2} for
clause~2, \rulename{S3} for clause~3 (with
$I \subseteq I'$), \rulename{S4} for clause~4 (with $I' \subseteq I$),
and \rulename{S7} for clause~5. The premises give the stated
sub-relations. For $\subtw$, \rulename{S3} is restricted to $I' = I$;
the dual direction reads the same rule off the head constructor of
$\unfold{\Ti}$.
\end{proof}

\begin{lemma}[Merge is a Widening Lower Bound]
\label{lemma:merge-widening-bound}
If $\comergefull{S}{\T}$, then $\T \subtw \Tii$ for every
$\Tii \in S$.
\end{lemma}
\begin{proof}
Let $R = \set{(\T, \Tii) \mid \exists S \ni \Tii.\
\comergefull{S}{\T}}$; we show $R$ is consistent with the rules
of~$\subtw$. Take $(\T, \Tii) \in R$ via $S$, and consider the rule
justifying $\comergefull{S}{\T}$.
By induction on the total number of leading $\toc{\sfmu}$-binders
among the members of~$S$, we may assume the justifying rule is not
\rulename{M6}: if it is, and the unfolded member is $\Tii$ itself,
say $\Tii = \toc{\sfmu}\ty.\T[0]$, then rule \rulename{S6} applies with
premise $(\T, \T[0]\sub{\toc{\sfmu}\ty.\T[0]}{\ty}) \in R$ via the premise
merge; if the unfolded member is another member, the premise merge
still contains $\Tii$ and justifies $(\T, \Tii) \in R$ with fewer
binders.
For \rulename{M7}, $\T = \toc{\sfmu}\ty.\T[0]$ and rule \rulename{S5}
applies with premise in $R$ via the premise merge.
Otherwise all members are in constructor form and the rule is
structural:
\rulename{M1} gives $(\tend, \tend)$, matched by \rulename{S7};
\rulename{M2}/\rulename{M3} give equal value heads, matched by
\rulename{S1}/\rulename{S2}, with premises in~$R$;
\rulename{M4} gives selections with equal index sets, matched by the
restricted \rulename{S3}, with premises in~$R$;
\rulename{M5} gives $\T = \tbrasub{\pp}{\lab[i] : \T[i]}{i \in \bigcup_j I_j}$
and $\Tii = \tbrasub{\pp}{\lab[i] : \Tii[i]}{i \in I_m}$ for some~$m$,
matched by \rulename{S4} since $I_m \subseteq \bigcup_j I_j$, and for
each $i \in I_m$ the premise merge over the members offering~$i$
contains $\Tii[i]$, so $(\T[i], \Tii[i]) \in R$.
\end{proof}

\begin{lemma}[Projection is Contained in the Combined Relation]
\label{lemma:projection-in-combined}
If $\G \cofullproj{p} \Ti$ and $\T \subtw \Ti$, then
$\G \cofullprojsub{p} \T$.
In particular, $\G \cofullproj{p} \T$ implies
$\G \cofullprojsub{p} \T$.
\end{lemma}
\begin{proof}
The second claim follows from the first by reflexivity of $\subtw$
(Definition~\ref{def:widening-subtyping}).
Let $R = \set{(\G, \T) \mid \exists \Ti.\
\G \cofullproj{p} \Ti \text{ and } \T \subtw \Ti}$; we show $R$ is
consistent with the rules of $\cofullprojsub{p}$.
Take $(\G, \T) \in R$ via $\Ti$; by
Lemma~\ref{lemma:unfold-replacement} we may assume $\Ti$ is in
constructor form.
If $\T = \toc{\sfmu}\ty.\T[0]$, rule \rulename{SPR10} applies with
premise $(\G, \T[0]\sub{\toc{\sfmu}\ty.\T[0]}{\ty}) \in R$
(Lemma~\ref{lemma:unfold-replacement}(1) for $\subtw$).
Otherwise $\T$ is in constructor form; we case-split on the rule
justifying $\G \cofullproj{p} \Ti$ (since $\Ti$ is in constructor
form, rule \rulename{PR10} cannot apply):
\begin{enumerate}[leftmargin=*, label=\textbullet]
\item \rulename{PR9}: $\G = \toc{\sfmu}\ty.\G[1]$ with
  $\pp \in \pt{\G[1]}$ and premise
  $(\G[1]\sub{\toc{\sfmu}\ty.\G[1]}{\ty}, \Ti) \in {\cofullproj{p}}$,
  so the corresponding pair with $\T$ lies in $R$;
  rule \rulename{SPR9} applies.
\item \rulename{PR12}: $\Ti = \tend$ and $\pp \notin \pt{\G[1]}$.
  By Lemma~\ref{lemma:inversion-subtyping-sub} (dual direction),
  $\unfold{\T} = \tend$, so $\T = \tend$; rule \rulename{SPR12}
  applies.
\item \rulename{PR11}: $\G = \gend$, $\Ti = \tend$, so $\T = \tend$;
  rule \rulename{SPR11} applies.
\item \rulename{PR1}: $\G = \Gvt{\pp}{\pq}{\B}{\G[0]}$,
  $\Ti = \tout{\pq}{\B}{\Ti[0]}$, premise
  $\G[0] \cofullproj{p} \Ti[0]$.
  By Lemma~\ref{lemma:inversion-subtyping-sub} (dual direction),
  $\T = \tout{\pq}{\B}{\T[0]}$ with $\T[0] \subtw \Ti[0]$, so
  $(\G[0], \T[0]) \in R$; rule \rulename{SPR1} applies.
  \rulename{PR2} is symmetric via \rulename{SPR2}.
\item \rulename{PR4}: $\Ti = \tselsub{\pq}{\lab[i] : \Ti[i]}{i \in I}$
  with the index set $I$ of the global head. By the dual inversion,
  $\T = \tselsub{\pq}{\lab[i] : \T[i]}{i \in I}$ (equal sets) with
  $\T[i] \subtw \Ti[i]$, so each $(\G[i], \T[i]) \in R$;
  rule \rulename{SPR4} applies.
\item \rulename{PR5}: $\Ti = \tbrasub{\pq}{\lab[j] : \Ti[j]}{j \in J}$
  with $J$ the index set of the global head. By the dual inversion,
  $\T = \tbrasub{\pq}{\lab[i] : \T[i]}{i \in I}$ with $I \supseteq J$
  and $\T[j] \subtw \Ti[j]$ for $j \in J$, so each
  $(\G[j], \T[j]) \in R$; rule \rulename{SPR5} applies with
  $J \subseteq I$.
\item \rulename{PR3}: premise $(\G[0], \Ti) \in {\cofullproj{p}}$
  gives $(\G[0], \T) \in R$; rule \rulename{SPR3} applies.
\item \rulename{PR6}: premises $\G[i] \cofullproj{p} \Ti[i]$ with
  $\comergefull{\set{\Ti[i]}_{i \in I}}{\Ti}$.
  By Lemma~\ref{lemma:merge-widening-bound}, $\Ti \subtw \Ti[i]$;
  by Lemma~\ref{lemma:subtyping-transitive}, $\T \subtw \Ti[i]$,
  so $(\G[i], \T) \in R$ for every~$i$; rule \rulename{SPR6}
  applies. \qedhere
\end{enumerate}
\end{proof}

\begin{lemma}[Completed Participants]
\label{lemma:completed-participants}
\begin{enumerate}[leftmargin=*, label=(\arabic*)]
  \item If $\G \cofullprojsub{p} \T$ and $\pp \notin \pt{\G}$, then
    $\unfold{\T} = \tend$.
  \item If $\pp \notin \pt{\G}$ and $\unfold{\T} = \tend$, then
    $\G \cofullprojsub{p} \T$.
  \item If $\G \cofullprojsub{p} \T$ and $\pp \in \pt{\G}$, then
    $\unfold{\T} \neq \tend$.
  \item If $\pp \notin \pt{\G}$ and $\unfold{\T} = \tend$, then
    $\G \cofullproj{p} \T$.
\end{enumerate}
\end{lemma}
\begin{proof}
(1) By lexicographic induction on (the size of $\G$, the number of
leading $\toc{\sfmu}$-binders of~$\T$), analysing the rule justifying
$\G \cofullprojsub{p} \T$.
If the rule is \rulename{SPR10}, recurse on its premise (fewer
binders on the $\T$-side).
Otherwise: if $\G$ is a $\toc{\sfmu}$-type, rule \rulename{SPR9} is
excluded by its side condition, so the rule is \rulename{SPR12} and
$\T = \tend$.
If $\G = \gend$, the rule is \rulename{SPR11} and $\T = \tend$.
If $\G$ has a communication head, it cannot involve $\pp$
(as $\pp \notin \pt{\G}$), so the rule is \rulename{SPR3} or
\rulename{SPR6}; recurse on a premise, whose global type is a strict
syntactic subterm.

(2) The relation
$R = \set{(\G, \T) \mid \pp \notin \pt{\G},\ \unfold{\T} = \tend}$
is consistent: if $\T$ is a $\toc{\sfmu}$-type, apply
\rulename{SPR10}; if $\T = \tend$, apply \rulename{SPR12}
(for $\toc{\sfmu}$-types $\G$), \rulename{SPR11} (for $\G = \gend$),
or \rulename{SPR3}/\rulename{SPR6} (for foreign heads), with premises
in~$R$.

(3) By induction on the depth of a shallowest occurrence of $\pp$ in
the regular tree of $\unfold{\G}$ (finite since $\pp \in \pt{\G}$).
If the head of $\unfold{\G}$ involves $\pp$,
Lemma~\ref{lemma:inversion-combined-relation}(1)--(4) forces
$\unfold{\T}$ to be a value or choice type, hence not $\tend$.
Otherwise the head is foreign and a shallowest occurrence of $\pp$
lies in some branch $\G[i]$ with strictly smaller depth;
Lemma~\ref{lemma:inversion-combined-relation}(5)--(6) gives
$\G[i] \cofullprojsub{p} \T$, and the induction hypothesis applies.

(4) As in~(2), the same relation $R$ is consistent for
$\cofullproj{p}$: if $\T$ is a $\toc{\sfmu}$-type, apply
\rulename{PR10}; if $\T = \tend$, apply \rulename{PR12}
(for $\toc{\sfmu}$-types $\G$), \rulename{PR11} (for $\G = \gend$),
\rulename{PR3} (for foreign value heads), or \rulename{PR6}
(for foreign choice heads, projecting every branch to $\tend$ and
merging by rule \rulename{M1}), with premises in~$R$.
\end{proof}

\subsubsection{Finiteness Lemmas}

\begin{lemma}[Finiteness of reachable local types]
\label{lemma:finite-reachable-local-states}
For every $\T$, the set
$\set{\Ti \mid \T \transa{}^{*} \Ti}$
is finite.
\end{lemma}

\begin{proof}[Proof of Lemma~\ref{lemma:finite-reachable-local-states}]
Because local types are guarded, each use of \rulename{LR5} only unfolds a
guarded recursion and exposes one of finitely many constructors already present
in the syntax of $\T$. Up to identifying each recursive variable with its
unique binder, every reachable local state is a syntactic subterm of $\T$.
Since $\T$ has finitely many subterms, $\set{\Ti \mid \T \transa{}^{*} \Ti}$ is
finite.
\end{proof}

\begin{lemma}[Finiteness of reachable contexts]
\label{lemma:finite-reachable-context-states}
For every $\ctx$, the set
$\set{\ctxi \mid \ctx \transi{}^{*} \ctxi}$
is finite.
\end{lemma}

\begin{proof}[Proof of Lemma~\ref{lemma:finite-reachable-context-states}]
Write $\ctx = \set{\ptag{\pp}{\T[\pp]}}_{\pp \in \dom{\ctx}}$.
By Lemma~\ref{lemma:finite-reachable-local-states}, each
$\set{\Ti \mid \T[\pp] \transa{}^{*} \Ti}$ is finite.
Every synchronous context step (\rulename{LEnv}) updates only two participants
and each updated local type follows one local transition. Hence any
reachable $\ctxi$ satisfies
$\ctxi(\pp) \in \set{\Ti \mid \T[\pp] \transa{}^{*} \Ti}$ for each
$\pp \in \dom{\ctx}$. Therefore
\[
\set{\ctxi \mid \ctx \transi{}^{*} \ctxi}
\subseteq
\prod_{\pp\in\dom{\ctx}} \set{\Ti \mid \T[\pp] \transa{}^{*} \Ti},
\]
and the right-hand side is finite because $\dom{\ctx}$ is finite.
\end{proof}

\begin{lemma}[Finiteness of reachable context-participant spaces]
\label{lemma:finite-reachable-context-priority-space}
For any context $\ctx$, the set
$\set{(\ctxi,\pp) \mid \ctx \transi{}^{*}\ctxi \text{ and } \pp \in \dom{\ctx}}$
is finite.
\end{lemma}

\begin{proof}[Proof of Lemma~\ref{lemma:finite-reachable-context-priority-space}]
By Lemma~\ref{lemma:finite-reachable-context-states},
$\set{\ctxi \mid \ctx \transi{}^{*} \ctxi}$ is finite, and
$\dom{\ctx}$ is finite. Their product is finite.
\end{proof}

\begin{lemma}[Finiteness of branching]
\label{lemma:finite-branching-context-transitions}
For every $\ctx$, the set
$\set{(\inter,\ctxi) \mid \ctx \transi{\inter} \ctxi}$
is finite.
\end{lemma}

\begin{proof}[Proof of Lemma~\ref{lemma:finite-branching-context-transitions}]
The number of participants is finite. For each participant, the enabled
one-step local actions are finite by inspection of rules \rulename{LR1}--
\rulename{LR5} (payload send/receive, finite label choices, and guarded
unfolding). A context transition is formed by pairing one send and one matching
receive via \rulename{LEnv}. Therefore only finitely many synchronous
transitions are possible from $\ctx$.
\end{proof}

\synthterminatesforslivecontexts*
\begin{proof}[Proof of Lemma~\ref{lemma:synth-terminates-for-slive-contexts}]
Fix $\ctx$ and participant $\pp[i]$.
Define
\[
S
=
\set{(\ctxi,\pp) \mid \ctx \transi{}^{*}\ctxi \text{ and } \pp \in \dom{\ctx}}.
\]
By Lemma~\ref{lemma:finite-reachable-context-priority-space}, $S$ is finite.

We prove a stronger statement:
for every $\Sigma$, $\pp[i']$, and $\ctxi$ such that
$\ctx \transi{}^{*}\ctxi$ and $\dom{\Sigma}\subseteq S$,
the computation $\textsf{synth}(\Sigma,\pp[i'],\ctxi)$ terminates.

Use well-founded induction on the natural number
\[
m(\Sigma)=\left|S\setminus\dom{\Sigma}\right|.
\]
Consider one call $\textsf{synth}(\Sigma,\pp[i'],\ctxi)$.

\begin{enumerate}
  \item If $\textsf{stuck}(\ctxi)$, it returns $\gend$ immediately.

  \item Otherwise let $\pp[j]=\textsf{next-active}(\ctxi,\pp[i'])$.
  If $(\ctxi,\pp[j])\in\dom{\Sigma}$, it returns $\Sigma(\ctxi,\pp[j])$ immediately.

  \item Otherwise it enters the recursive branch.
  Let $\Sigma'=\Sigma\cup\set{(\ctxi,\pp[j])\mapsto\ty}$.
  Since $\ctx \transi{}^{*}\ctxi$ and $\pp[j]\in\dom{\ctx}$, we have $(\ctxi,\pp[j])\in S$,
  and by case assumption $(\ctxi,\pp[j])\notin\dom{\Sigma}$, therefore
  \[
  m(\Sigma')=m(\Sigma)-1.
  \]
  For each $\lab\in L$, choose $\ctx[\lab]$ with
  $\ctxi \transi{\interaction{\pp[j]}{\pq}{\lab}} \ctx[\lab]$.
  Then $\ctx \transi{}^{*}\ctx[\lab]$, and still $\dom{\Sigma'}\subseteq S$.
  Hence each recursive call
  $\textsf{synth}(\Sigma',\pp[(j{+}1)\bmod n],\ctx[\lab])$
  satisfies the induction hypothesis (strictly smaller measure), so it
  terminates.
  By Lemma~\ref{lemma:finite-branching-context-transitions}, $L$ is finite, so
  there are only finitely many such recursive calls.
\end{enumerate}

All branches terminate, so $\textsf{synth}(\Sigma,\pp[i'],\ctxi)$ terminates.
Instantiating $\Sigma=\varnothing$, $\pp[i']=\pp[i]$, and $\ctxi=\ctx$ yields the claim.
\end{proof}

\begin{remark}[Induction principle of the synthesis procedure]
\label{rem:synth-induction}\em
The termination proof above establishes a general induction principle.
Let $S = \set{(\ctxi,\pp) \mid \ctx \transi{}^{*}\ctxi,\; \pp \in \dom{\ctx}}$.
To prove a property $P(\Sigma,\pp[i],\ctxi)$ for all
$\Sigma$ with $\dom{\Sigma}\subseteq S$, all $\pp[i]$, and all
$\ctxi$ reachable from $\ctx$, it suffices to assume $P$ holds
for every recursive sub-call of $\textsf{synth}$ (which has
a strictly smaller measure $m$) and verify $P$ in each of the
three cases of Definition~\ref{def:synth-rules}.
\end{remark}

\subsubsection{Substitution and Structural Properties}

Throughout this subsection, we treat the fresh variable chosen in
Definition~\ref{def:synth-rules}(3) as determined by the call, so
that $\textsf{synth}$ is a function of its arguments. Two
observations follow by inspection of
Definition~\ref{def:synth-rules}: the free variables of
$\textsf{synth}(\Sigma, \pp[i], \ctx)$ lie in
$\mathrm{range}(\Sigma)$; and
$\textsf{synth}(\Sigma, \pp[i], \ctx) =
 \textsf{synth}(\Sigma, \pp[i]', \ctx)$
whenever
$\textsf{next-active}(\ctx, \pp[i]) = \textsf{next-active}(\ctx, \pp[i]')$,
since the definition reads the priority nowhere else.

\begin{lemma}[Properties of Unfolding Equivalence]
\label{lemma:unfeq-properties}
$\unfeq$ is an equivalence relation;
$\toc{\sfmu}\ty.\G \unfeq \G\sub{\toc{\sfmu}\ty.\G}{\ty}$;
and $\G \unfeq \Gi$ implies $\pt{\G} = \pt{\Gi}$.
\end{lemma}
\begin{proof}
For reflexivity, the identity relation satisfies the clauses of
Definition~\ref{def:unfolding-equivalence}; for symmetry and
transitivity, the converse of a relation satisfying the clauses, and
the composition of two such relations, again satisfy them, matching
head shapes through the middle type in the latter case. For the
unfolding law, $\unfold{\toc{\sfmu}\ty.\G} =
\unfold{\G\sub{\toc{\sfmu}\ty.\G}{\ty}}$ holds by definition of
$\unfold{\cdot}$, so the pair satisfies the clauses with
continuations related by reflexivity. For participants, $\pt{\G}$
equals the set of participants occurring in the communication heads
at unfold-reachable positions of $\G$: the inclusion $\supseteq$ is
immediate, and $\subseteq$ holds since every communication subterm
of a closed guarded type occurs at some unfold-reachable position,
by structural induction. A straightforward coinduction shows that
$\unfeq$-related types expose identical heads at all
unfold-reachable positions, so their participant sets coincide.
\end{proof}

Call an environment $\Sigma$ \emph{coherent} when its bindings are
injective on variables and every $(\ctxi, \pp[a]) \in \dom{\Sigma}$
has $\slive(\ctxi)$ with $\pp[a]$ an enabled sender of $\ctxi$.
The \emph{resolving substitutions} $\mathrm{Res}(\Sigma)$ for a
coherent $\Sigma$ are defined inductively: the empty substitution is
in $\mathrm{Res}(\varnothing)$; and if
$\theta \in \mathrm{Res}(\Sigma)$, the extension
$\Sigma \cup \set{(\ctxi, \pp[a]) \mapsto \ty}$ is coherent with
$(\ctxi, \pp[a]) \notin \dom{\Sigma}$ and $\ty$ fresh, then
$\theta[\ty \mapsto \textsf{synth}(\Sigma, \pp, \ctxi)\theta]
\in \mathrm{Res}(\Sigma \cup \set{(\ctxi, \pp[a]) \mapsto \ty})$,
where $\pp$ is any priority with
$\textsf{next-active}(\ctxi, \pp) = \pp[a]$; by the second
observation above, the resolvent does not depend on this choice.
By induction over this definition:
$\dom{\theta} = \mathrm{range}(\Sigma)$; every resolvent is a
closed, guarded global type; extensions never overwrite existing
bindings; and every prefix of the derivation of
$\theta \in \mathrm{Res}(\Sigma)$ yields a coherent sub-environment
together with its own resolving substitution.

\substitutionforsynth*
\begin{proof}[Proof of Lemma~\ref{lemma:substitution-for-synth}]
We exhibit a relation closed under the clauses of
Definition~\ref{def:unfolding-equivalence}; since $\unfeq$ is the
largest such relation, the candidate is contained in $\unfeq$. Let
\[
R =
\set{
\left(\textsf{synth}(\Sigma, \pp[i], \ctx)\theta,\;
      \textsf{synth}(\Sigma', \pp[i], \ctx)\theta'\right)
\mid
\slive(\ctx),\
\theta \in \mathrm{Res}(\Sigma),\
\theta' \in \mathrm{Res}(\Sigma')
}
\]
with $\Sigma, \Sigma'$ coherent. Both components are closed and
guarded. Given a pair in $R$, we case-split on each component
following the three cases of Definition~\ref{def:synth-rules}, with
$\pp[k] = \textsf{next-active}(\pp[i])$; the cases taken by the two
sides can differ only in case~(2), since $\textsf{stuck}(\ctx)$ is a
property of the shared~$\ctx$.

\emph{Case~(1)}: $\textsf{stuck}(\ctx)$.
Both components are $\gend$, and clause~(1) of
Definition~\ref{def:unfolding-equivalence} holds.

\emph{Case~(2)}: $(\ctx, \pp[k]) \in \dom{\Sigma}$, say
$\Sigma(\ctx, \pp[k]) = \ty$. Then
$\textsf{synth}(\Sigma, \pp[i], \ctx)\theta = \theta(\ty)$, which is
literally equal to
$\textsf{synth}(\Sigma^0, \pp, \ctx)\theta^0$, the resolvent
recorded when the binding was created, since bindings are never
overwritten; here $\Sigma^0$ and
$\theta^0 \in \mathrm{Res}(\Sigma^0)$ arise from a prefix of the
derivation of $\theta \in \mathrm{Res}(\Sigma)$, with
$(\ctx, \pp[k]) \notin \dom{\Sigma^0}$ and
$\textsf{next-active}(\ctx, \pp) = \pp[k]$, so
$\theta(\ty) = \textsf{synth}(\Sigma^0, \pp[i], \ctx)\theta^0$.
The pair therefore equals another pair of $R$ over the same $\ctx$
and $\pp[i]$, now with $(\ctx, \pp[k])$ unbound on that side; at
most one such replacement is needed per side. We may therefore
assume both sides are in case~(3).

\emph{Case~(3)}: $\neg\textsf{stuck}(\ctx)$ and
$(\ctx, \pp[k]) \notin \dom{\Sigma}, \dom{\Sigma'}$.
Both sides enter the recursive branch with the same partner~$\pq$,
the same transition set~$L$, and fresh variables $\ty[s]$ and
$\ty[s]'$. Thus
$\textsf{synth}(\Sigma, \pp[i], \ctx) = \toc{\sfmu}\ty[s].C$, where
$C$ carries the head of the enabled action of $(\ctx, \pp[k])$ with
continuations
$\textsf{synth}(\Sigma_1, \pp[(k{+}1)\bmod n], \ctx[\lab])$ and
$\Sigma_1 = \Sigma \cup \set{(\ctx, \pp[k]) \mapsto \ty[s]}$, which
is coherent. Since resolvents are closed and $\ty[s]$ is fresh,
substitution commutes with the binder, and the unfolding takes
exactly one step because $C$ begins with a communication
constructor:
\[
\unfold{\toc{\sfmu}\ty[s].(C\theta)}
= (C\theta)\sub{\toc{\sfmu}\ty[s].(C\theta)}{\ty[s]}
= C\theta_1,
\qquad
\theta_1 = \theta[\ty[s] \mapsto
\textsf{synth}(\Sigma, \pp[i], \ctx)\theta],
\]
where the second equality holds because resolvents are closed and
$\ty[s] \notin \dom{\theta}$, and
$\theta_1 \in \mathrm{Res}(\Sigma_1)$ is exactly the defining clause
of $\mathrm{Res}$. Hence the unfolding exposes the head interaction
of $(\ctx, \pp[k])$ with continuations
$\textsf{synth}(\Sigma_1, \pp[(k{+}1)\bmod n], \ctx[\lab])\theta_1$,
and symmetrically for the other side with
$\theta'_1 \in \mathrm{Res}(\Sigma'_1)$.
The two heads agree, being determined by the shared $\ctx$ and
$\pp[k]$, and by Lemma~\ref{lemma:preservation-reduction} each
successor satisfies $\slive(\ctx[\lab])$, so the paired
continuations lie in~$R$ and clause~(2) or~(3) of
Definition~\ref{def:unfolding-equivalence} holds.

Hence $R \subseteq \unfeq$. The lemma is the instance
$\Sigma = \set{(\ctxi, \pp[j]) \mapsto \ty}$ with
$\theta = [\ty \mapsto \textsf{synth}(\varnothing, \pp[j], \ctxi)]
\in \mathrm{Res}(\Sigma)$,
taking $\pp = \pp[j]$, an enabled sender of $\ctxi$, so that
$\textsf{next-active}(\pp[j]) = \pp[j]$,
against $\Sigma' = \varnothing$ with the empty substitution.
\end{proof}

\nonterminatedparticipantappearsinsynthesisedtype*
\begin{proof}[Proof of Lemma~\ref{lemma:nonterminated-participant-appears}]
Fix $\ctx$, $\pp[i]$, and $\pp$ such that $\slive(\ctx)$ and $\unfold{\ctx(\pp)} \neq \tend$.
Write
\[
\G = \textsf{synth}(\varnothing, \pp[i], \ctx),
\]
and assume, for contradiction, that $\pp \notin \pt{\G}$.

By Lemma~\ref{lemma:synth-terminates-for-slive-contexts}, the computation of
$\textsf{synth}(\varnothing, \pp[i], \ctx)$ terminates, so its call-unfolding gives a finite
recursion tree.  Consider an arbitrary root-to-leaf path in this tree:
\[
(\ctx[0],\ppi[0]) \rrtrans{\inter[0]} (\ctx[1],\ppi[1]) \rrtrans{\inter[1]} \cdots
\rrtrans{\inter[m-1]} (\ctx[m],\ppi[m]),
\]
where $\ppi[j]$ denotes the prioritised participant at step~$j$,
with $(\ctx[0],\ppi[0])=(\ctx,\pp[i])$.

Since $\pp \notin \pt{\G}$, no constructor generated along this path involves $\pp$.
Hence each $\inter[k]$ is between participants different from $\pp$.
By Lemma~\ref{lemma:semantics}, $\ctx[k+1](\pp)=\ctx[k](\pp)$ for all $k<m$.
Therefore, for all $k \le m$:
\[
\ctx[k](\pp)=\ctx(\pp), \qquad \unfold{\ctx[k](\pp)}=\unfold{\ctx(\pp)}\neq\tend.
\]
So each $\ctx[k]$ has a one-sided local action at $\pp$.

If the leaf is a case~(1) leaf, the projected context path is finite
and maximal.
By Lemma~\ref{lemma:projection-of-round-robin-sequences-is-fair} it
is fair, hence live, since $\live(\ctx)$.
As shown above, $\ctx[0] \transt{\pp : \czeta}$ for some action
$\czeta$, so Lemma~\ref{lemma:fair-live-communication} yields a step
$\inter[k]$ involving $\pp$, contradicting the fact that no
interaction along the path involves $\pp$.

So the leaf must be a case~(2) leaf. Then there is an ancestor on the same path with the
same context-participant pair:
\[
(\ctx[r],\ppi[r])=(\ctx[m],\ppi[m]), \quad r<m.
\]
Let $\sigma$ be the finite segment from index $r$ to $m$. Repeating
$\sigma$ forever yields an
infinite reduction sequence $\rho$ from $\ctx$
in which no interaction involves $\pp$, and
$\ctx[k](\pp)=\ctx(\pp)$ at every position.
Moreover $\rho$ is the projection of an infinite maximal round-robin
sequence obtained by looping the same cycle in the recursion tree,
so it is fair by
Lemma~\ref{lemma:projection-of-round-robin-sequences-is-fair}, and
hence live, since $\live(\ctx)$.
As $\ctx[0] \transt{\pp : \czeta}$ for some action $\czeta$,
Lemma~\ref{lemma:fair-live-communication} yields a step of $\rho$
involving $\pp$, contradicting the construction of~$\rho$.

Hence $\pp \in \pt{\textsf{synth}(\varnothing, \pp[i], \ctx)}$.
\end{proof}

\begin{restatable}[Terminated Participant Does Not Appear in Synthesised Type]{lemma}{terminatedparticipantnotappearsinsynthesisedtype}
  \label{lemma:terminated-participant-not-appears}
  For any environment $\Sigma$, context $\ctx$ satisfying $\slive(\ctx)$, priority $\pp[i]$, and participant $\pp \in \dom{\ctx}$, if $\unfold{\ctx(\pp)} = \tend$, then
  $\pp \notin \pt{\textsf{synth}(\Sigma, \pp[i], \ctx)}$.
\end{restatable}
\begin{proof}[Proof of Lemma~\ref{lemma:terminated-participant-not-appears}]
We prove the stronger statement: for all $\Sigma$, $\pp[i]$, and
$\ctxi$ reachable from $\ctx$, if $\slive(\ctxi)$ and
$\unfold{\ctxi(\pp)} = \tend$, then
$\pp \notin \pt{\textsf{synth}(\Sigma, \pp[i], \ctxi)}$.
By the induction principle of Remark~\ref{rem:synth-induction},
we case-split on $\textsf{synth}(\Sigma, \pp[i], \ctxi)$.
Let $\pp[k] = \textsf{next-active}(\ctxi, \pp[i])$.

\emph{Case~(1)}: the result is $\gend$, and
$\pt{\gend} = \emptyset$, so $\pp \notin \pt{\gend}$.

\emph{Case~(2)}: $\Sigma(\ctxi, \pp[k]) = \ty$, so the
result is $\ty$ and $\pt{\ty} = \emptyset$, hence
$\pp \notin \pt{\ty}$.

\emph{Case~(3)}: the result is
$\toc{\sfmu}\ty.\G[0]$ where $\G[0]$ is either
$\Gvt{\pp[k]}{\pq}{\B}{\G'}$ (sub-case~(a)) or
$\pp[k] \gar \pq\;\goc{\{} \lab : \G[\lab] \goc{\}}_{\lab \in L}$ (sub-case~(b)).
Since $\unfold{\ctxi(\pp)} = \tend$,
Lemma~\ref{lemma:end-empty} gives $\ctxi(\pp) \ntransi{}$,
so $\pp$ cannot be a sender or receiver of any context transition.
Hence $\pp \neq \pp[k]$ and $\pp \neq \pq$.
For each successor context~$\ctx[\lab]$
reached via $\ctxi \transi{\interaction{\pp[k]}{\pq}{\obs}} \ctx[\lab]$,
by Lemma~\ref{lemma:semantics}(3),
$\ctx[\lab](\pp) = \ctxi(\pp)$,
so $\unfold{\ctx[\lab](\pp)} = \tend$.
By Lemma~\ref{lemma:preservation-reduction}, $\slive(\ctx[\lab])$.
By the induction hypothesis,
$\pp \notin \pt{\G[\lab]}$ for each successor.
Since $\pp \neq \pp[k]$ and $\pp \neq \pq$,
$\pp \notin \pt{\toc{\sfmu}\ty.\G[0]}$.
\end{proof}

\begin{restatable}[No-Transition Safe Live Contexts Have Only End Locals]{lemma}{notransitionsafelivecontextsonlyendlocals}
  \label{lemma:no-transition-safe-live-only-end}
  For any context $\ctx$ satisfying $\slive(\ctx)$, if $\ctx \ntransi{}$, then for every $\pp \in \dom{\ctx}$ we have
  $\unfold{\ctx(\pp)} = \tend$.
\end{restatable}
\begin{proof}[Proof of Lemma~\ref{lemma:no-transition-safe-live-only-end}]
Suppose for contradiction that $\unfold{\ctx(\pp)} \neq \tend$
for some $\pp \in \dom{\ctx}$.
By Lemma~\ref{lemma:end-empty} (contrapositive),
$\ctx(\pp)$ has at least one local transition,
so $\ctx \transt{\pp : \czeta}$ for some action $\czeta$.
Since $\ctx \ntransi{}$, the only maximal reduction sequence
from $\ctx$ is the empty (stuck) sequence, which is trivially
fair.
However, this sequence is not live: the one-sided action
at $\pp$ is enabled at $\ctx$ but no interaction involving $\pp$
is ever taken (Definition~\ref{def:live-paths}).
This contradicts $\live(\ctx)$.
\end{proof}

\preservationofsafetyandlivenessunderreduction*
\begin{proof}[Proof of Lemma~\ref{lemma:preservation-reduction}]
\emph{Safety.}
By Definition~\ref{def:stuck}, $\safe(\ctx)$ requires that for all
$\ctx \transi{}^{*} \ctxii$, if
$\ctxii \transt{\pr : \qsend{\role{s}}{\obsi}}$ and
$\ctxii \transt{\role{s} : \qrecv{\pr}{\obs'}}$,
then $\ctxii \transi{\interaction{\pr}{\role{s}}{\obsi}}$.
Since $\ctx \transi{} \ctxi$, every $\ctxi \transi{}^{*} \ctxii$
is also $\ctx \transi{}^{*} \ctxii$, so the condition
is inherited by~$\ctxi$.

\emph{Liveness.}
Let $(\ctxi = \ctx[0] \transi{\inter[0]} \ctx[1] \transi{\inter[1]} \cdots)$
be a fair reduction sequence from~$\ctxi$.
Prepending the step $\ctx \transi{\interaction{\pp}{\pq}{\obs}} \ctxi$
gives a maximal reduction sequence from~$\ctx$, and it is fair.
A pair enabled at position~$0$ either is $(\pp, \pq)$ itself, taken
at the prepended step, or is disjoint from $\set{\pp, \pq}$, since the
heads of $\pp$ and $\pq$ each name the other as unique partner; in
the disjoint case both endpoints are unchanged by the step
(Lemma~\ref{lemma:semantics}(3)), so the pair remains enabled at
position~$1$ and fairness of the tail supplies the required step.
At later positions fairness is inherited from the tail by shifting
indices.
By $\live(\ctx)$ (Definition~\ref{def:livecontexts}),
this extended sequence is live.
Hence the original sequence from~$\ctxi$ is live, so $\live(\ctxi)$.
\end{proof}

\begin{lemma}[End Empty]
    \label{lemma:end-empty}
    If $\T \ntransa{}$ then $\unfold{\T} = \tend$.
\end{lemma}

\begin{proof}[Proof of Lemma~\ref{lemma:end-empty}]
By inspection of rules \rulename{LR1}--\rulename{LR4},
every local type in constructor form other than $\tend$
(i.e.\ output, input, selection, or branching)
has at least one enabled transition.
Rule \rulename{LR5} unfolds $\toc{\sfmu}$-types, so if
$\T$ is a $\toc{\sfmu}$-type its transitions coincide with
those of $\unfold{\T}$.
Hence if $\T$ has no transitions, $\unfold{\T}$ is not a
$\toc{\sfmu}$-type and has no transitions, so $\unfold{\T} = \tend$.
\end{proof}

\subsubsection{Backward Closure}

\backwardclosureofsynthcandidate*
\begin{proof}[Proof of Lemma~\ref{lemma:backward-closure-synth-candidate}]
We show that for every $(\G,\T)\in\widehat{R}_{\pp}$, there exists a
rule of $\cofullprojsub{p}$ whose conclusion matches
$\G \cofullprojsub{p} \T$ and whose premises all lie in
$\widehat{R}_{\pp}$.

By definition, every $(\G,\T)\in\widehat{R}_{\pp}$ originates from
some $(\Gi,\Ti)\in R_{\pp}$, meaning
$\Gi = \textsf{synth}(\varnothing,\pp[i],\ctx)$ and $\Ti = \ctx(\pp)$
for some $\slive(\ctx)$ with $\pp\in\dom{\ctx}$, with
$\G \unfeq \Gi$ via \rulename{RUnfEq} (iterated applications
collapse by transitivity of $\unfeq$,
Lemma~\ref{lemma:unfeq-properties}), and $\T$ obtained from~$\Ti$ by
finitely many applications of \rulename{RUnfR}, chaining through
nested $\toc{\sfmu}$-binders.
We organise the proof by the form of~$\G$.

\emph{Case $\T = \toc{\sfmu}\ty'.\Ti$.}
Rule \rulename{SPR10} applies; its premise
$(\G, \Ti\sub{\toc{\sfmu}\ty'.\Ti}{\ty'})$ lies in
$\widehat{R}_{\pp}$ by \rulename{RUnfR}.

\noindent
Henceforth assume $\T$ is not a $\toc{\sfmu}$-type.

\emph{Case $\G = \gend$.}
Then $\unfold{\Gi} = \gend$ by clause~(1) of
Definition~\ref{def:unfolding-equivalence}, so synthesis returned
$\gend$ and $\textsf{stuck}(\ctx)$, since otherwise
Definition~\ref{def:synth-rules}(3) emits a $\toc{\sfmu}$-type
unfolding to a communication head.
By Lemma~\ref{lemma:no-transition-safe-live-only-end},
$\unfold{\ctx(\pp)} = \tend$, so $\T = \tend$.
Rule \rulename{SPR11}.

\emph{Case $\G = \toc{\sfmu}\ty.\G[0]$.}
By Lemma~\ref{lemma:unfeq-properties}, $\pt{\G} = \pt{\Gi}$.
If $\pp \in \pt{\G[0]}$: rule \rulename{SPR9} applies.
Its premise is $(\G[0]\sub{\toc{\sfmu}\ty.\G[0]}{\ty}, \T)$.
Since $\G[0]\sub{\toc{\sfmu}\ty.\G[0]}{\ty} \unfeq \G \unfeq \Gi$
(Lemma~\ref{lemma:unfeq-properties}),
\rulename{RUnfEq} applied to $(\Gi, \Ti) \in R_{\pp}$, followed by
the same chain of \rulename{RUnfR} steps that produced $\T$
from~$\Ti$, places the premise in $\widehat{R}_{\pp}$.

If $\pp \notin \pt{\G[0]}$: since
$\pt{\toc{\sfmu}\ty.\G[0]} = \pt{\G[0]}$ and $\pt{\G} = \pt{\Gi}$,
the contrapositive of
Lemma~\ref{lemma:nonterminated-participant-appears} applied to $\Gi$
gives $\unfold{\ctx(\pp)} = \tend$, so $\T = \tend$.
Rule \rulename{SPR12}.

\emph{Case $\G$ is in communication form.}
Then $\neg\textsf{stuck}(\ctx)$, since otherwise $\Gi = \gend$ and
clause~(1) of Definition~\ref{def:unfolding-equivalence} would force
$\unfold{\G} = \gend$.
Write
$\Gi = \toc{\sfmu}\ty.\G[0] =
\textsf{synth}(\varnothing,\pp[i],\ctx)$:
synthesis hit Definition~\ref{def:synth-rules}(3)
with sender $\pp[j] = \textsf{next-active}(\ctx, \pp[i])$ and
receiver~$\pq$, and for each successor context~$\ctx[\lab]$
reached via $\ctx \transi{\interaction{\pp[j]}{\pq}{\obs}} \ctx[\lab]$,
the synthesis sub-call gives
$\G[\lab] = \textsf{synth}(\set{(\ctx,\pp[j])\!\mapsto\!\ty},\;
\pp[(j{+}1)\bmod n],\;\ctx[\lab])$.
Since $\G$ exposes its head and
$\G \unfeq \Gi \unfeq \G[0]\sub{\Gi}{\ty}$
(Lemma~\ref{lemma:unfeq-properties}), clause~(2) or~(3) of
Definition~\ref{def:unfolding-equivalence} gives that $\G$ carries
the same head interaction as $\G[0]\sub{\Gi}{\ty}$, with
continuations $\Gii[\lab]$ satisfying
$\Gii[\lab] \unfeq \G[\lab]\sub{\Gi}{\ty}$.
By Lemma~\ref{lemma:preservation-reduction},
$\slive(\ctx[\lab])$ for each successor, so
Lemma~\ref{lemma:substitution-for-synth} applies
(noting $\textsf{synth}(\varnothing,\pp[j],\ctx) = \Gi$, since
Definition~\ref{def:synth-rules} reads the priority only through
$\textsf{next-active}$ and $\textsf{next-active}(\ctx, \pp[j]) = \pp[j]$),
and with transitivity of $\unfeq$ yields
\[
\Gii[\lab]
\;\unfeq\;
\textsf{synth}(\varnothing,\,\pp[(j{+}1)\bmod n],\,\ctx[\lab]).
\]
Hence
\[
(\star)\qquad
\bigl(\,\Gii[\lab],\;\ctx[\lab](\pp)\,\bigr)
\;\in\; \widehat{R}_{\pp}
\qquad\text{by \rulename{RCand} and \rulename{RUnfEq},
for each successor } \ctx[\lab].
\]
We case-split on the form of the enabled action and the role of~$\pp$.

\smallskip
\emph{Sub-case~(a): value-passing ($\obs = \B$, single continuation).}
Then $\G = \Gvt{\pp[j]}{\pq}{\B}{\Gii[\lab]}$.
By Lemma~\ref{lemma:semantics}:
$\unfold{\ctx(\pp[j])} = \tout{\pq}{\B}{\ctx[\lab](\pp[j])}$,
$\unfold{\ctx(\pq)} = \tin{\pp[j]}{\B}{\ctx[\lab](\pq)}$,
and $\ctx[\lab](\pr) = \ctx(\pr)$ for $\pr \notin \set{\pp[j],\pq}$.

If $\pp = \pp[j]$:
$\T = \tout{\pq}{\B}{\ctx[\lab](\pp[j])}$.
\rulename{SPR1}; premise in $\widehat{R}_{\pp}$
by~$(\star)$.
\quad
If $\pp = \pq$:
$\T = \tin{\pp[j]}{\B}{\ctx[\lab](\pq)}$.
\rulename{SPR2}; premise in $\widehat{R}_{\pp}$
by~$(\star)$.

If $\pp \notin \set{\pp[j],\pq}$:
$\ctx[\lab](\pp) = \ctx(\pp)$.
\rulename{SPR3}; premise
$(\Gii[\lab], \T)$ lies in $\widehat{R}_{\pp}$:
if $\T = \ctx(\pp)$ by~$(\star)$ directly;
if $\T \neq \ctx(\pp)$ by repeated \rulename{RUnfR}
on~$(\star)$ (since $\ctx[\lab](\pp) = \ctx(\pp)$, the same unfolding
chain applies).

\smallskip
\emph{Sub-case~(b): labelled branching ($L = I$).}
The chosen interaction has $\obs = \lab[k]$ for some~$k$, and
Lemma~\ref{lemma:semantics}(2) gives
$\unfold{\ctx(\pp[j])} = \tselsub{\pq}{\lab[i] : \T[i]}{i \in I}$ and
$\unfold{\ctx(\pq)} = \tbrasub{\pp[j]}{\lab[i] : \Ti[i]}{i \in I_r}$
with $k \in I \cap I_r$.
Hence $\ctx \transt{\pp[j] : \qsend{\pq}{\lab[i]}}$ for every
$i \in I$, and $\ctx \transt{\pq : \qrecv{\pp[j]}{\lab[k]}}$, so
$\safe(\ctx)$ yields
$\ctx \transi{\interaction{\pp[j]}{\pq}{\lab[i]}} \ctx[i]$ for every
$i \in I$.
Conversely, Lemma~\ref{lemma:semantics}(2) forces the label of any
interaction between $\pp[j]$ and $\pq$ to lie in~$I$, so the
successor contexts are exactly $\set{\ctx[i]}_{i \in I}$; thus
$L = I$ and
$\G = \GvtPair{\pp[j]}{\pq}{\lab[i] : \Gii[i]}{i \in I}$.
Applying Lemma~\ref{lemma:semantics} to each of these transitions
gives $i \in I_r$, hence $I \subseteq I_r$, together with
$\ctx[i](\pp[j]) = \T[i]$, $\ctx[i](\pq) = \Ti[i]$, and
$\ctx[i](\pr) = \ctx(\pr)$ for $\pr \notin \set{\pp[j],\pq}$.

If $\pp = \pp[j]$:
$\T = \tselsub{\pq}{\lab[i] : \T[i]}{i \in I}$.
\rulename{SPR4} with $J = I$; each premise
$(\Gii[i], \T[i])$
lies in $\widehat{R}_{\pp}$
by~$(\star)$ (since $\T[i] = \ctx[i](\pp)$).

If $\pp = \pq$:
$\T = \tbrasub{\pp[j]}{\lab[i] : \Ti[i]}{i \in I_r}$.
\rulename{SPR5} with $J = I \subseteq I_r$; each premise
$(\Gii[i], \Ti[i])$
lies in $\widehat{R}_{\pp}$
by~$(\star)$ (since $\Ti[i] = \ctx[i](\pp)$).

If $\pp \notin \set{\pp[j], \pq}$:
$\ctx[i](\pp) = \ctx(\pp)$ for every $i \in I$
(Lemma~\ref{lemma:semantics}(3)).
\rulename{SPR6} applies;
\emph{every premise asks for
$(\Gii[i], \T)$ with the same~$\T$}.
By~$(\star)$, $(\Gii[i], \ctx[i](\pp)) \in \widehat{R}_{\pp}$
for each~$i$, and $\ctx[i](\pp) = \ctx(\pp)$ is the same type across
all branches.
If $\T = \ctx(\pp)$, the premises are directly~$(\star)$;
if $\T \neq \ctx(\pp)$, apply repeated \rulename{RUnfR}
on~$(\star)$ (since $\ctx[i](\pp) = \ctx(\pp)$, the same unfolding
chain applies to every branch).
\end{proof}

\subsubsection{Association Proofs}

\begin{lemma}[Widening Bound Yields a Merge]
\label{lemma:widening-bound-yields-merge}
Let $I$ be finite and nonempty, and let $\T \subtw \Ti[i]$ for every
$i \in I$. Then there exists $\Ti$ with
$\comergefull{\set{\Ti[i]}_{i \in I}}{\Ti}$ and $\T \subtw \Ti$.
\end{lemma}

\begin{proof}
We construct $\Ti$ and verify $\comergefull{\set{\Ti[i]}_{i \in I}}{\Ti}$
and $\T \subtw \Ti$ simultaneously, by corecursion guided by the
constructor form of $\T$; this is productive because recursion in $\T$
is guarded, so $\T$ has finitely many distinct subterms and each step
emits one constructor.
By Lemma~\ref{lemma:unfold-replacement} we may replace every type by
its unfolding, so we assume $\T$ and each $\Ti[i]$ are in constructor
form; by Lemma~\ref{lemma:inversion-subtyping-sub} (dual direction for
$\subtw$) the head of each $\Ti[i]$ is then determined by that of $\T$.
We case-split on $\T$.

\emph{Case $\T = \tend$.}
Each $\Ti[i] = \tend$, so $\set{\Ti[i]}_{i \in I} = \set{\tend}$; set
$\Ti = \tend$.
Rule \rulename{M1}; rule \rulename{S7}.

\emph{Cases $\T = \tout{\pq}{\B}{\T[0]}$ and $\T = \tin{\pq}{\B}{\T[0]}$.}
Each $\Ti[i] = \tout{\pq}{\B}{\Ti[0,i]}$ (resp.\ $\tin{\pq}{\B}{\cdot}$)
with $\T[0] \subtw \Ti[0,i]$.
Corecursively merge $\set{\Ti[0,i]}_{i \in I}$ to $\Ti[0]$ with
$\T[0] \subtw \Ti[0]$, and set $\Ti = \tout{\pq}{\B}{\Ti[0]}$
(resp.\ input).
Rule \rulename{M2} (resp.\ \rulename{M3}); rule \rulename{S1}
(resp.\ \rulename{S2}).

\emph{Case $\T = \tselsub{\pq}{\lab[k] : \T[k]}{k \in K}$.}
Each $\Ti[i] = \tselsub{\pq}{\lab[k] : \Ti[k,i]}{k \in K}$ with the
\emph{same} label set $K$ (restricted \rulename{S3}) and
$\T[k] \subtw \Ti[k,i]$.
For each $k \in K$, corecursively merge $\set{\Ti[k,i]}_{i \in I}$ to
$\Ti[k]$ with $\T[k] \subtw \Ti[k]$, and set
$\Ti = \tselsub{\pq}{\lab[k] : \Ti[k]}{k \in K}$.
Rule \rulename{M4} (equal label sets); rule \rulename{S3} with equal
sets.

\emph{Case $\T = \tbrasub{\pq}{\lab[k] : \T[k]}{k \in K}$.}
Each $\Ti[i] = \tbrasub{\pq}{\lab[k] : \Ti[k,i]}{k \in K_i}$ with
$K_i \subseteq K$ and $\T[k] \subtw \Ti[k,i]$ for $k \in K_i$.
For each $k \in \bigcup_{i} K_i$ the set
$\set{\Ti[k,i] \mid i \in I,\ k \in K_i}$ is nonempty and has the
common widening lower bound $\T[k]$; corecursively merge it to
$\Ti[k]$ with $\T[k] \subtw \Ti[k]$, and set
$\Ti = \tbrasub{\pq}{\lab[k] : \Ti[k]}{k \in \bigcup_{i} K_i}$.
Rule \rulename{M5}; rule \rulename{S4}, since
$\bigcup_{i} K_i \subseteq K$.
\end{proof}

\compositeimpliesassociation*
\begin{proof}[Proof of Lemma~\ref{lemma:composite-implies-association}]
We first prove a decomposition claim, the converse of
Lemma~\ref{lemma:projection-in-combined}.
Whenever $\G \cofullprojsub{p} \T$, there exists $\Ti$ with
$\G \cofullproj{p} \Ti$ and $\T \subtw \Ti$ (the widening subtyping of
Definition~\ref{def:widening-subtyping}); this suffices for
association because ${\subtw} \subseteq {\subt}$.
The proof is by coinduction, constructing $\Ti$ by case-splitting on
the rule of $\cofullprojsub{p}$ and appealing to the coinductive
hypothesis for the premises.

\emph{Cases \rulename{SPR1}/\rulename{SPR2}.}
$\G = \Gvt{\pp}{\pq}{\B}{\G[0]}$ and
$\T = \tout{\pq}{\B}{\T[0]}$
(resp.\ $\tin{\pq}{\B}{\T[0]}$),
with $\G[0] \cofullprojsub{p} \T[0]$.
By coinductive hypothesis, there exists $\Ti[0]$ with
$\G[0] \cofullproj{p} \Ti[0]$ and $\T[0] \subtw \Ti[0]$.
Set $\Ti = \tout{\pq}{\B}{\Ti[0]}$
(resp.\ $\tin{\pq}{\B}{\Ti[0]}$).
Rule \rulename{PR1} (resp.\ \rulename{PR2}) for projection;
rule \rulename{S1} (resp.\ \rulename{S2}) for subtyping.

\emph{Case \rulename{SPR3}.}
$\G = \Gvt{\pq}{\pr}{\B}{\G[0]}$,
$\pp \notin \set{\pq, \pr}$,
$\G[0] \cofullprojsub{p} \T$.
By c.h., $\exists \Ti$: $\G[0] \cofullproj{p} \Ti$ and $\T \subtw \Ti$.
Rule \rulename{PR3} gives $\G \cofullproj{p} \Ti$.
Subtyping inherited from the premise.

\emph{Case \rulename{SPR4}.}
$\G = \GvtPair{\pp}{\pq}{\lab[i] : \G[i]}{i \in I}$,
$\T = \tselsub{\pq}{\lab[i] : \T[i]}{i \in I}$,
with $\G[i] \cofullprojsub{p} \T[i]$ for each $i$.
By c.h., for each $i$: $\G[i] \cofullproj{p} \Ti[i]$ and
$\T[i] \subtw \Ti[i]$.
Set $\Ti = \tselsub{\pq}{\lab[i] : \Ti[i]}{i \in I}$.
Rule \rulename{PR4}; rule \rulename{S3} with equal index sets.

\emph{Case \rulename{SPR5}.}
$\G = \GvtPair{\pq}{\pp}{\lab[j] : \G[j]}{j \in J}$,
$\T = \tbrasub{\pq}{\lab[i] : \T[i]}{i \in I}$ with $J \subseteq I$,
and $\G[j] \cofullprojsub{p} \T[j]$ for each $j \in J$.
By c.h., for each $j \in J$: $\G[j] \cofullproj{p} \Ti[j]$ and
$\T[j] \subtw \Ti[j]$.
Set $\Ti = \tbrasub{\pq}{\lab[j] : \Ti[j]}{j \in J}$.
Rule \rulename{PR5}; rule \rulename{S4} with $J \subseteq I$.

\emph{Case \rulename{SPR6}.}
$\G = \GvtPair{\pq}{\pr}{\lab[i] : \G[i]}{i \in I}$,
$\pp \notin \set{\pq, \pr}$,
$\G[i] \cofullprojsub{p} \T$ for all $i$ (same~$\T$).
By c.h., for each $i$:
$\G[i] \cofullproj{p} \Ti[i]$ and $\T \subtw \Ti[i]$.
The types $\set{\Ti[i]}_{i \in I}$ thus share the common widening
lower bound $\T$, so
Lemma~\ref{lemma:widening-bound-yields-merge} provides $\Ti$ with
$\comergefull{\set{\Ti[i]}_{i \in I}}{\Ti}$ and $\T \subtw \Ti$.
Rule \rulename{PR6} gives $\G \cofullproj{p} \Ti$.

\emph{Case \rulename{SPR9}.}
$\G = \toc{\sfmu} \ty . \G[0]$, $\pp \in \pt{\G[0]}$,
$\G[0] \sub{\toc{\sfmu} \ty . \G[0]}{\ty} \cofullprojsub{p} \T$.
By c.h., $\exists \Ti$:
$\G[0] \sub{\toc{\sfmu} \ty . \G[0]}{\ty} \cofullproj{p} \Ti$
and $\T \subtw \Ti$.
Rule \rulename{PR9} gives
$\toc{\sfmu} \ty . \G[0] \cofullproj{p} \Ti$.
Subtyping inherited.

\emph{Case \rulename{SPR10}.}
$\T = \toc{\sfmu} \ty . \T[0]$,
with $\G \cofullprojsub{p} \T[0] \sub{\toc{\sfmu} \ty . \T[0]}{\ty}$.
By c.h., $\exists \Ti$:
$\G \cofullproj{p} \Ti$ and
$\T[0] \sub{\toc{\sfmu} \ty . \T[0]}{\ty} \subtw \Ti$.
Projection $\G \cofullproj{p} \Ti$ directly; and since
$\unfold{\toc{\sfmu} \ty . \T[0]} =
\unfold{\T[0] \sub{\toc{\sfmu} \ty . \T[0]}{\ty}}$,
Lemma~\ref{lemma:unfold-replacement}(1) gives
$\toc{\sfmu} \ty . \T[0] \subtw \Ti$.

\emph{Cases \rulename{SPR11}/\rulename{SPR12}.}
$\T = \tend$. Set $\Ti = \tend$.
Rules \rulename{PR11}/\rulename{PR12} and rule \rulename{S7}.

\medskip
\noindent
This proves the decomposition claim.
To conclude association, split $\ctx = \ctxact, \ctxterm$, where
$\ctxact$ is the restriction of $\ctx$ to $\pt{\G}$ (recall
$\dom{\ctx} \supseteq \pt{\G}$).
For each $\pp \in \pt{\G}$, the decomposition applied to
$\G \cofullprojsub{p} \ctx(\pp)$ yields $\Ti[\pp]$ with
$\G \cofullproj{p} \Ti[\pp]$ and $\ctx(\pp) \subtw \Ti[\pp]$, hence
$\ctx(\pp) \subt \Ti[\pp]$ since ${\subtw} \subseteq {\subt}$;
defining $\ctxi$ by $\ctxi(\pp) = \Ti[\pp]$, we have
$\dom{\ctxi} = \pt{\G}$,
$\G \cofullproj{p} \ctxi(\pp)$ for each $\pp$, and
$\ctxact \subt \ctxi$ componentwise.
For each $\pp \in \dom{\ctx} \setminus \pt{\G}$,
Lemma~\ref{lemma:completed-participants}(1) gives
$\unfold{\ctx(\pp)} = \tend$, so $\ctxterm$ contains only
terminated endpoints.
By Definition~\ref{def:association}, $\ctx \assoc \G$.
\end{proof}

\subsubsection{Balancedness}

\synthesisedbalanced*
\begin{proof}[Proof of Lemma~\ref{lemma:synthesised-balanced}]
We show that for every subterm $\Gi \in \subterms{\G}$,
$\pt{\Gi} \subseteq \unavoidable(\Gi)$
(the reverse inclusion is immediate from the definition).
Equivalently, we show that every $\pp \in \pt{\Gi}$
appears on every path through the unfolded tree of~$\Gi$;
this path formulation agrees with the inductive definition of
$\unavoidable$ since global types are finitely branching.

Each subterm $\Gi$ of $\G$ corresponds to a synthesis call
$\textsf{synth}(\Sigma, \pp[k], \ctxi)$ for some reachable context~$\ctxi$
(i.e.\ $\ctx \transi{}^{*} \ctxi$).
By repeated application of
Lemma~\ref{lemma:preservation-reduction},
$\slive(\ctxi)$.

Consider any such subterm $\Gi$ with corresponding context $\ctxi$.
Let $\pp \in \pt{\Gi}$ and consider any path $\pi$ through the
unfolded tree of $\Gi$.
We must show $\pp$ appears on~$\pi$.

The path $\pi$ fixes a specific label at each branching node.
Together with the round-robin schedule, this determines a
maximal round-robin reduction sequence from~$\ctxi$.
By Lemma~\ref{lemma:projection-of-round-robin-sequences-is-fair},
this projects onto a \emph{fair} reduction sequence from~$\ctxi$.
Since $\live(\ctxi)$, this fair sequence is also \emph{live}
(Definition~\ref{def:livecontexts}).

We claim that every participant $\pp$ with
$\unfold{\ctxi(\pp)} \neq \tend$ communicates on this path.
By Lemma~\ref{lemma:end-empty} (contrapositive), $\pp$ has a
one-sided action $\czeta$ at $\ctxi$, and $\safe(\ctxi)$ holds by
$\slive(\ctxi)$, so
Lemma~\ref{lemma:fair-live-communication} yields a step of the path
involving $\pp$.
After $\pp$ communicates, its type may change; if $\pp$ still has a
one-sided action at the later position,
Lemma~\ref{lemma:fair-live-communication} applies again, and so on
until $\pp$ reaches~$\tend$.

It remains to connect $\pp \in \pt{\Gi}$ with the hypothesis of the
claim. If $\unfold{\ctxi(\pp)} = \tend$, then
Lemma~\ref{lemma:terminated-participant-not-appears} gives
$\pp \notin \pt{\Gi}$, a contradiction; hence
$\unfold{\ctxi(\pp)} \neq \tend$, the claim applies, and $\pp$
appears on~$\pi$.

Since $\pp$ appears on every path through $\Gi$, we have
$\pp \in \unavoidable(\Gi)$.
As this holds for all $\pp \in \pt{\Gi}$ and all
$\Gi \in \subterms{\G}$, the type $\G$ is balanced.
\end{proof}

\subsubsection{Main Theorems}

\completenessofassociation*

The proof goes through the combined relation.
We prove the key completeness lemma for the combined relation,
from which Lemma~\ref{lemma:completeness-association} follows via
Lemma~\ref{lemma:composite-implies-association}.

\begin{lemma}[Completeness of the Combined Relation]
\label{lemma:completeness-combined}
Let $\G$ be balanced, let $\inter = \interaction{\pp}{\pq}{\obs}$ with
$\pp, \pq \in \pt{\G}$, and let $\ctx$ be a typing context with
$\dom{\ctx} \supseteq \pt{\G}$ and
$\G \cofullprojsub{x} \ctx(\role{x})$ for every
$\role{x} \in \pt{\G}$. Suppose:
\begin{enumerate}[leftmargin=*, label=(\roman*)]
\item $\obs = \B$, with
  $\unfold{\ctx(\pp)} = \tout{\pq}{\B}{\T[\pp]}$ and
  $\unfold{\ctx(\pq)} = \tin{\pp}{\B}{\T[\pq]}$; or
\item $\obs = \lab[j]$, with
  $\unfold{\ctx(\pp)} = \tselsub{\pq}{\lab[i] : \T[\pp,i]}{i \in I_s}$,
  $j \in I_s$, and
  $\unfold{\ctx(\pq)} = \tbrasub{\pp}{\lab[i] : \T[\pq,i]}{i \in I_r}$.
\end{enumerate}
Then, in case~(ii), $I_s \subseteq I_r$; and, writing
$\T[\pp] = \T[\pp,j]$ and $\T[\pq] = \T[\pq,j]$
in case~(ii), there exists a balanced $\Gi$ with
$\G \transg{\inter} \Gi$, $\pt{\Gi} \subseteq \pt{\G}$, and
\[
\Gi \cofullprojsub{p} \T[\pp],
\qquad
\Gi \cofullprojsub{q} \T[\pq],
\qquad
\Gi \cofullprojsub{x} \ctx(\role{x})
\ \text{ for every } \role{x} \in \pt{\G} \setminus \set{\pp,\pq}.
\]
\end{lemma}
\begin{proof}
Throughout we reason about $\unfold{\G}$ and reintroduce leading
$\toc{\sfmu}$-binders in the constructed transition via (iterated)
rule \rulename{GR3}; note that $\pt{}$, $\subterms{}$, balancedness,
and all hypotheses are invariant under unfolding.

\emph{Measure.}
Since $\G$ is balanced, $\pp \in \pt{\Gii} = \unavoidable(\Gii)$ for every
$\Gii \in \subterms{\G}$ with $\pp \in \pt{\Gii}$, and $\unavoidable$ is a
least fixed point computed over the finitely many subterms of $\G$
(Definition~\ref{def:subterms-unavoidable}).
For such $\Gii$, let $\mathrm{depth}_{\pp}(\Gii)$ be the least stage of this
fixed-point computation at which $\pp$ enters $\unavoidable(\Gii)$.
By construction, $\mathrm{depth}_{\pp}(\Gii) = 1$ iff $\pp$ occurs in the
head of $\unfold{\Gii}$;
if $\unfold{\Gii} = \Gvt{\pr}{\role{s}}{\Bi}{\Gii[0]}$ with
$\pp \notin \set{\pr,\role{s}}$ then
$\mathrm{depth}_{\pp}(\Gii) = \mathrm{depth}_{\pp}(\Gii[0]) + 1$;
and if $\unfold{\Gii} = \GvtPair{\pr}{\role{s}}{\lab[i] : \Gii[i]}{i \in I}$
with $\pp \notin \set{\pr,\role{s}}$ then
$\mathrm{depth}_{\pp}(\Gii) = 1 + \max_{i \in I} \mathrm{depth}_{\pp}(\Gii[i])$.
We proceed by strong induction on $\mathrm{depth}_{\pp}(\G)$.

\emph{Any head involving $\pp$ or $\pq$ is the $(\pp,\pq)$-head.}
Suppose the head of $\unfold{\G}$ involves $\pp$.
If $\pp$ were the receiver, or the sender towards some
$\role{r} \neq \pq$, then
Lemma~\ref{lemma:inversion-combined-relation}(1)--(4) applied to
$\G \cofullprojsub{p} \ctx(\pp)$ would force $\unfold{\ctx(\pp)}$
to be an input, or an output directed at $\role{r}$, contradicting
the shape hypothesised in (i)/(ii); for the same reason the kind of
the head (value or choice) matches the hypothesis case.
Symmetrically, a head involving $\pq$ forces, via
$\G \cofullprojsub{q} \ctx(\pq)$ and the input shape of
$\unfold{\ctx(\pq)}$, the partner to be $\pp$.
Consequently, if $\mathrm{depth}_{\pp}(\G) \geq 2$ then the head of
$\unfold{\G}$ involves neither $\pp$ nor $\pq$.

\emph{Base case $\mathrm{depth}_{\pp}(\G) = 1$.}
In case~(i), $\unfold{\G} = \Gvt{\pp}{\pq}{\B}{\G[0]}$.
Lemma~\ref{lemma:inversion-combined-relation}(1) at $\pp$ and (2) at
$\pq$ give (matching the hypothesised shapes, which determine the
continuations uniquely)
$\G[0] \cofullprojsub{p} \T[\pp]$ and
$\G[0] \cofullprojsub{q} \T[\pq]$, while clause~(5),
instantiated at each $\role{x} \in \pt{\G} \setminus \set{\pp,\pq}$,
gives $\G[0] \cofullprojsub{x} \ctx(\role{x})$.
Rules \rulename{GR3}/\rulename{GR1} give
$\G \transg{\inter} \G[0]$; set $\Gi = \G[0]$.
Since $\Gi \in \subterms{\G}$, it is balanced and
$\pt{\Gi} \subseteq \pt{\G}$.

In case~(ii),
$\unfold{\G} = \GvtPair{\pp}{\pq}{\lab[i] : \G[i]}{i \in I^{\dagger}}$.
Clause~(3) at $\pp$ forces $I^{\dagger} = I_s$ and gives
$\G[i] \cofullprojsub{p} \T[\pp,i]$ for all $i$; clause~(4) at $\pq$
gives $I_r \supseteq I^{\dagger} = I_s$ (in particular $j \in I_r$,
proving the first claim) and
$\G[i] \cofullprojsub{q} \T[\pq,i]$ for $i \in I^{\dagger}$;
clause~(6) gives
$\G[i] \cofullprojsub{x} \ctx(\role{x})$ for all $i$ and all
$\role{x} \in \pt{\G} \setminus \set{\pp,\pq}$.
Rules \rulename{GR3}/\rulename{GR4} (with $j \in I^{\dagger}$) give
$\G \transg{\inter} \G[j]$; set $\Gi = \G[j]$, a balanced subterm
with $\pt{\Gi} \subseteq \pt{\G}$, and instantiate the displayed
relations at $i = j$.

\emph{Inductive case, foreign value head.}
$\unfold{\G} = \Gvt{\pr}{\role{s}}{\Bi}{\G[0]}$ with
$\set{\pr,\role{s}} \cap \set{\pp,\pq} = \emptyset$ (by the head
analysis above).
Clauses (1), (2), and (5) of
Lemma~\ref{lemma:inversion-combined-relation} give
$\unfold{\ctx(\pr)} = \tout{\role{s}}{\Bi}{\T[\pr]}$ with
$\G[0] \cofullprojsub{r} \T[\pr]$, symmetrically
$\unfold{\ctx(\role{s})} = \tin{\pr}{\Bi}{\T[\role{s}]}$ with
$\G[0] \cofullprojsub{s} \T[\role{s}]$, and
$\G[0] \cofullprojsub{x} \ctx(\role{x})$ for every other
$\role{x} \in \pt{\G}$.
Since
$\pt{\G} = \unavoidable(\unfold{\G}) =
\set{\pr,\role{s}} \cup \unavoidable(\G[0])$ and
$\pp, \pq \notin \set{\pr, \role{s}}$, we get
$\pp, \pq \in \unavoidable(\G[0]) \subseteq \pt{\G[0]}$; moreover
$\G[0] \in \subterms{\G}$ is balanced,
$\pt{\G[0]} \subseteq \pt{\G}$, and
$\mathrm{depth}_{\pp}(\G[0]) = \mathrm{depth}_{\pp}(\G) - 1$.
The induction hypothesis applies to $\G[0]$ with the context
$\ctx[0]$ defined by $\ctx[0](\pr) = \T[\pr]$,
$\ctx[0](\role{s}) = \T[\role{s}]$, and
$\ctx[0](\role{x}) = \ctx(\role{x})$ otherwise
(the hypotheses at $\pp, \pq$ are unchanged), yielding a balanced
$\Gi[0]$ with $\G[0] \transg{\inter} \Gi[0]$,
$\pt{\Gi[0]} \subseteq \pt{\G[0]}$, and the residual relations.
Rules \rulename{GR3}/\rulename{GR2}
($\set{\pr,\role{s}} \cap \set{\pp,\pq} = \emptyset$) give
$\G \transg{\inter} \Gi$ for
$\Gi = \Gvt{\pr}{\role{s}}{\Bi}{\Gi[0]}$.

We check the required relations for $\Gi$.
For $\pp$: from $\Gi[0] \cofullprojsub{p} \T[\pp]$, rule
\rulename{SPR3} gives $\Gi \cofullprojsub{p} \T[\pp]$;
likewise for $\pq$, and for every
$\role{x} \in \pt{\G} \setminus \set{\pp,\pq,\pr,\role{s}}$
(such $\role{x}$ lie in $\pt{\G[0]}$, so the induction hypothesis
covers them, with $\ctx[0](\role{x}) = \ctx(\role{x})$).
For $\pr$: if $\pr \in \pt{\G[0]}$, the induction hypothesis gives
$\Gi[0] \cofullprojsub{r} \T[\pr]$; otherwise
Lemma~\ref{lemma:completed-participants}(1) gives
$\unfold{\T[\pr]} = \tend$ and, since
$\pr \notin \pt{\Gi[0]} \subseteq \pt{\G[0]}$,
Lemma~\ref{lemma:completed-participants}(2) gives
$\Gi[0] \cofullprojsub{r} \T[\pr]$ as well.
Rule \rulename{SPR1} then yields
$\Gi \cofullprojsub{r} \tout{\role{s}}{\Bi}{\T[\pr]} = \unfold{\ctx(\pr)}$,
and closure under \rulename{SPR10} yields
$\Gi \cofullprojsub{r} \ctx(\pr)$.
The case of $\role{s}$ is symmetric, via \rulename{SPR2}.
Finally, $\Gi[0]$ is balanced, so
$\unavoidable(\Gi) = \set{\pr,\role{s}} \cup \unavoidable(\Gi[0])
= \set{\pr,\role{s}} \cup \pt{\Gi[0]} = \pt{\Gi}$ and
$\subterms{\Gi} = \set{\Gi} \cup \subterms{\Gi[0]}$; hence $\Gi$ is
balanced, and
$\pt{\Gi} = \set{\pr,\role{s}} \cup \pt{\Gi[0]} \subseteq \pt{\G}$.

\emph{Inductive case, foreign choice head.}
$\unfold{\G} = \GvtPair{\pr}{\role{s}}{\lab[i] : \G[i]}{i \in I_0}$
with $\set{\pr,\role{s}} \cap \set{\pp,\pq} = \emptyset$.
Clauses (3), (4), and (6) of
Lemma~\ref{lemma:inversion-combined-relation} give
$\unfold{\ctx(\pr)} = \tselsub{\role{s}}{\lab[i] : \T[\pr,i]}{i \in I_0}$
with $\G[i] \cofullprojsub{r} \T[\pr,i]$ for all $i$,
$\unfold{\ctx(\role{s})} = \tbrasub{\pr}{\lab[i] : \T[\role{s},i]}{i \in I_{\role{s}}}$
with $I_0 \subseteq I_{\role{s}}$ and
$\G[i] \cofullprojsub{s} \T[\role{s},i]$ for $i \in I_0$, and
$\G[i] \cofullprojsub{x} \ctx(\role{x})$ for all $i$ and all other
$\role{x} \in \pt{\G}$.
By balancedness,
$\pt{\G} = \set{\pr,\role{s}} \cup \bigcap_{i} \unavoidable(\G[i])$,
so for every branch $i$:
$\pp, \pq \in \unavoidable(\G[i]) \subseteq \pt{\G[i]}$, and moreover
\begin{equation}
\pt{\G} \setminus \set{\pr,\role{s}} \subseteq \pt{\G[i]}
\subseteq \pt{\G}.
\tag{$\ast$}
\end{equation}
Each $\G[i] \in \subterms{\G}$ is balanced and
$\mathrm{depth}_{\pp}(\G[i]) < \mathrm{depth}_{\pp}(\G)$.
The induction hypothesis applies in each branch with the context
$\ctx[i]$ defined by $\ctx[i](\pr) = \T[\pr,i]$,
$\ctx[i](\role{s}) = \T[\role{s},i]$, and
$\ctx[i](\role{x}) = \ctx(\role{x})$ otherwise,
yielding balanced $\Gi[i]$ with $\G[i] \transg{\inter} \Gi[i]$,
$\pt{\Gi[i]} \subseteq \pt{\G[i]}$,
$\Gi[i] \cofullprojsub{p} \T[\pp]$,
$\Gi[i] \cofullprojsub{q} \T[\pq]$, and
$\Gi[i] \cofullprojsub{x} \ctx[i](\role{x})$ for every
$\role{x} \in \pt{\G[i]} \setminus \set{\pp,\pq}$.
Since the \emph{same} interaction $\inter$ fires in every branch,
rules \rulename{GR3}/\rulename{GR5} give
$\G \transg{\inter} \Gi$ for
$\Gi = \GvtPair{\pr}{\role{s}}{\lab[i] : \Gi[i]}{i \in I_0}$.

We check the required relations for $\Gi$.
For $\pp$: the relations
$\Gi[i] \cofullprojsub{p} \T[\pp]$ hold for every branch
\emph{with the same local type}, so rule \rulename{SPR6} gives
$\Gi \cofullprojsub{p} \T[\pp]$; likewise for $\pq$, and,
using $\ctx[i](\role{x}) = \ctx(\role{x})$ together
with~($\ast$), for every
$\role{x} \in \pt{\G} \setminus \set{\pp,\pq,\pr,\role{s}}$.
For $\pr$: in each branch, either $\pr \in \pt{\G[i]}$ and the
induction hypothesis gives $\Gi[i] \cofullprojsub{r} \T[\pr,i]$, or
$\pr \notin \pt{\G[i]}$ and
Lemma~\ref{lemma:completed-participants}(1)--(2) gives the same, as
in the previous case.
Rule \rulename{SPR4} then yields
$\Gi \cofullprojsub{r}
\tselsub{\role{s}}{\lab[i] : \T[\pr,i]}{i \in I_0} = \unfold{\ctx(\pr)}$,
and \rulename{SPR10} closes to $\ctx(\pr)$.
For $\role{s}$, rule \rulename{SPR5} applies with
$I_0 \subseteq I_{\role{s}}$, again closed by \rulename{SPR10}.

Balancedness of $\Gi$: each $\Gi[i]$ is balanced, so it suffices to
check the root, i.e.\ that every
$\role{x} \in \pt{\Gi[i]} \setminus \set{\pr,\role{s}}$ lies in
$\pt{\Gi[k]}$ for every $k \in I_0$.
Let $\T[\role{x}]$ be $\T[\pp]$, $\T[\pq]$,
or $\ctx(\role{x})$ according to whether $\role{x}$ is $\pp$,
$\pq$, or another participant; by~($\ast$) and the above, the
relation $\Gi[k] \cofullprojsub{x} \T[\role{x}]$ holds in
\emph{every} branch $k$, with the same $\T[\role{x}]$.
Lemma~\ref{lemma:completed-participants}(3) in branch $i$ gives
$\unfold{\T[\role{x}]} \neq \tend$, and then
Lemma~\ref{lemma:completed-participants}(1) in branch $k$ gives
$\role{x} \in \pt{\Gi[k]}$.
Hence
$\unavoidable(\Gi) = \set{\pr,\role{s}} \cup \bigcap_k \pt{\Gi[k]}
= \set{\pr,\role{s}} \cup \bigcup_k \pt{\Gi[k]} = \pt{\Gi}$,
so $\Gi$ is balanced; and
$\pt{\Gi} \subseteq \set{\pr,\role{s}} \cup \bigcup_k \pt{\G[k]}
\subseteq \pt{\G}$.
\end{proof}

\begin{proof}[Proof of Lemma~\ref{lemma:completeness-association}]
Let $\inter = \interaction{\pp}{\pq}{\obs}$.
Since $\ctx \transi{\inter} \ctxi$, Lemma~\ref{lemma:semantics}
gives the head shapes:
if $\obs = \B$, then
$\unfold{\ctx(\pp)} = \tout{\pq}{\B}{\ctxi(\pp)}$ and
$\unfold{\ctx(\pq)} = \tin{\pp}{\B}{\ctxi(\pq)}$;
if $\obs = \lab[j]$, then
$\unfold{\ctx(\pp)} = \tselsub{\pq}{\lab[i] : \T[i]}{i \in I}$ and
$\unfold{\ctx(\pq)} = \tbrasub{\pp}{\lab[i] : \Ti[i]}{i \in J}$ with
$j \in I \cap J$, $\ctxi(\pp) = \T[j]$, and $\ctxi(\pq) = \Ti[j]$;
in both cases $\ctxi(\role{x}) = \ctx(\role{x})$ for
$\role{x} \notin \set{\pp,\pq}$.
In particular
$\unfold{\ctx(\pp)} \neq \tend \neq \unfold{\ctx(\pq)}$, so
$\pp$ and $\pq$ cannot lie in the $\ctxterm$ part of the split of
$\ctx$; hence $\pp, \pq \in \pt{\G}$, and clause~(1) of
Definition~\ref{def:association} provides a witness $\ctx[0]$ with
$\dom{\ctx[0]} = \pt{\G}$,
$\G \cofullproj{x} \ctx[0](\role{x})$, and
$\ctx(\role{x}) \subt \ctx[0](\role{x})$ for all
$\role{x} \in \pt{\G}$.

By Lemma~\ref{lemma:inversion-subtyping-sub} applied to
$\ctx(\pp) \subt \ctx[0](\pp)$ and $\ctx(\pq) \subt \ctx[0](\pq)$:
if $\obs = \B$, then
$\unfold{\ctx[0](\pp)} = \tout{\pq}{\B}{\T[\pp]}$ with
$\ctxi(\pp) \subt \T[\pp]$, and
$\unfold{\ctx[0](\pq)} = \tin{\pp}{\B}{\T[\pq]}$ with
$\ctxi(\pq) \subt \T[\pq]$;
if $\obs = \lab[j]$, then
$\unfold{\ctx[0](\pp)} =
\tselsub{\pq}{\lab[i] : \T[\pp,i]}{i \in I_s}$ with
$I \subseteq I_s$ (so $j \in I_s$) and
$\ctxi(\pp) = \T[j] \subt \T[\pp,j]$, and
$\unfold{\ctx[0](\pq)} =
\tbrasub{\pp}{\lab[i] : \T[\pq,i]}{i \in I_r}$ with
$I_r \subseteq J$ and $\Ti[i] \subt \T[\pq,i]$ for all $i \in I_r$.

By Lemma~\ref{lemma:projection-in-combined},
$\G \cofullprojsub{x} \ctx[0](\role{x})$ for every
$\role{x} \in \pt{\G}$, so
Lemma~\ref{lemma:completeness-combined} applies to $\ctx[0]$.
In the choice case it gives $I_s \subseteq I_r$, so $j \in I_r$ and
$\ctxi(\pq) = \Ti[j] \subt \T[\pq,j] = \T[\pq]$.
In both cases it yields a balanced $\Gi$ with
$\G \transg{\inter} \Gi$, $\pt{\Gi} \subseteq \pt{\G}$,
$\Gi \cofullprojsub{p} \T[\pp]$,
$\Gi \cofullprojsub{q} \T[\pq]$, and
$\Gi \cofullprojsub{x} \ctx[0](\role{x})$ for
$\role{x} \in \pt{\G} \setminus \set{\pp,\pq}$.

Define $\ctx[0]'$ with $\dom{\ctx[0]'} = \pt{\G}$ by
$\ctx[0]'(\pp) = \T[\pp]$,
$\ctx[0]'(\pq) = \T[\pq]$, and
$\ctx[0]'(\role{x}) = \ctx[0](\role{x})$ otherwise.
The relations above give $\Gi \cofullprojsubctx \ctx[0]'$
(note $\dom{\ctx[0]'} = \pt{\G} \supseteq \pt{\Gi}$), so
Lemma~\ref{lemma:composite-implies-association} yields
$\ctx[0]' \assoc \Gi$.
Moreover, $\ctxi(\role{x}) \subt \ctx[0]'(\role{x})$ for every
$\role{x} \in \pt{\G}$: at $\pp$ and $\pq$ this was established
above, and elsewhere
$\ctxi(\role{x}) = \ctx(\role{x}) \subt \ctx[0](\role{x})$.

We conclude $\ctxi \assoc \Gi$, taking the split of $\ctxi$ whose
first part is the restriction of $\ctxi$ to $\pt{\Gi}$ (note
$\dom{\ctxi} = \dom{\ctx} \supseteq \pt{\G} \supseteq \pt{\Gi}$).
For clause~(2), consider
$\role{x} \in \dom{\ctxi} \setminus \pt{\Gi}$.
If $\role{x} \notin \pt{\G}$, then $\role{x} \notin \set{\pp,\pq}$ and
$\unfold{\ctxi(\role{x})} = \unfold{\ctx(\role{x})} = \tend$ by
clause~(2) for $\ctx \assoc \G$;
if $\role{x} \in \pt{\G} \setminus \pt{\Gi}$, then
clause~(2) for $\ctx[0]' \assoc \Gi$ gives
$\unfold{\ctx[0]'(\role{x})} = \tend$, and from
$\ctxi(\role{x}) \subt \ctx[0]'(\role{x})$,
Lemma~\ref{lemma:unfold-replacement} and
Lemma~\ref{lemma:inversion-subtyping}(5) force
$\unfold{\ctxi(\role{x})} = \tend$.
For clause~(1), clause~(1) of $\ctx[0]' \assoc \Gi$ provides
$\ctx[0]''$ with
$\dom{\ctx[0]''} = \pt{\Gi}$,
$\Gi \cofullproj{x} \ctx[0]''(\role{x})$, and
$\ctx[0]'(\role{x}) \subt \ctx[0]''(\role{x})$ for
$\role{x} \in \pt{\Gi}$.
By Lemma~\ref{lemma:subtyping-transitive},
$\ctxi(\role{x}) \subt \ctx[0]''(\role{x})$ for
$\role{x} \in \pt{\Gi}$, so the same witness $\ctx[0]''$ serves for
$\ctxi$.
\end{proof}

\soundnessofassociation*
\begin{proof}[Proof of Lemma~\ref{lemma:soundness-association}]
Since $\G \transg{\inter} \Gi$, the unfolding of $\G$ has a
communication head, whose participants lie in $\pt{\G}$; in
particular $\pt{\G} \neq \emptyset$.
By clause~(1) of Definition~\ref{def:association}, there exists
$\ctx[0]$ such
that $\dom{\ctx[0]} = \pt{\G}$, $\G \cofullproj{p} \ctx[0](\pp)$
for all $\pp \in \pt{\G}$, and
$\ctx(\pp) \subt \ctx[0](\pp)$ for all $\pp \in \pt{\G}$.
By Lemma~\ref{lemma:projection-in-combined},
$\G \cofullprojsub{p} \ctx[0](\pp)$ for all $\pp \in \pt{\G}$;
below we invert these combined-relation judgements.

We proceed by induction on the derivation of
$\G \transg{\inter} \Gi$ to construct $\interi$ and $\ctxi$ with
$\ctx \transi{\interi} \ctxi$.
Once such $\interi, \ctxi$ are obtained, the global-side witness
$\Gii$ follows from Completeness of Association
(Lemma~\ref{lemma:completeness-association}) applied to
$\ctx \assoc \G$ (with $\G$ balanced) and
$\ctx \transi{\interi} \ctxi$:
there exists a balanced $\Gii$ with $\G \transg{\interi} \Gii$ and
$\ctxi \assoc \Gii$.

\emph{Case \rulename{GR1}.}
$\G = \Gvt{\pp}{\pq}{\B}{\G[0]}$, $\inter = \interaction{\pp}{\pq}{\B}$,
and $\Gi = \G[0]$.
Lemma~\ref{lemma:inversion-combined-relation} applied to
$\G \cofullprojsub{p} \ctx[0](\pp)$ and
$\G \cofullprojsub{q} \ctx[0](\pq)$ gives
$\unfold{\ctx[0](\pp)} = \tout{\pq}{\B}{\T[\pp]}$ and
$\unfold{\ctx[0](\pq)} = \tin{\pp}{\B}{\T[\pq]}$.
Since $\ctx(\pp) \subt \ctx[0](\pp)$ and
$\ctx(\pq) \subt \ctx[0](\pq)$,
Lemma~\ref{lemma:inversion-subtyping}\,(1)--(2) yields
$\unfold{\ctx(\pp)} = \tout{\pq}{\B}{\T'[\pp]}$ and
$\unfold{\ctx(\pq)} = \tin{\pp}{\B}{\T'[\pq]}$.
Rules \rulename{LR1}, \rulename{LR2}, \rulename{LTag}, and
\rulename{LEnv} give
$\ctx \transi{\interaction{\pp}{\pq}{\B}} \ctxi$ for some $\ctxi$.
Take $\interi = \inter$.

\emph{Case \rulename{GR4}.}
$\G = \GvtPair{\pp}{\pq}{\lab[i] : \G[i]}{i \in I}$,
$\inter = \interaction{\pp}{\pq}{\lab[j]}$ with $j \in I$,
and $\Gi = \G[j]$.
Lemma~\ref{lemma:inversion-combined-relation} gives
$\unfold{\ctx[0](\pp)} = \tselsub{\pq}{\lab[i] : \T[\pp,i]}{i \in I}$
and
$\unfold{\ctx[0](\pq)} = \tbrasub{\pp}{\lab[i] : \T[\pq,i]}{i \in I'}$
with $I \subseteq I'$.
Since $\ctx(\pp) \subt \ctx[0](\pp)$ and
$\ctx(\pq) \subt \ctx[0](\pq)$,
Lemma~\ref{lemma:inversion-subtyping}\,(3)--(4) yields
$\unfold{\ctx(\pp)} = \tselsub{\pq}{\lab[i] : \T'[\pp,i]}{i \in I_s}$
with $I_s \subseteq I$, and
$\unfold{\ctx(\pq)} = \tbrasub{\pp}{\lab[i] : \T'[\pq,i]}{i \in I_r}$
with $I_r \supseteq I'$.
By well-formedness of local types, $I_s \neq \emptyset$;
pick any $k \in I_s$.
Since $I_s \subseteq I \subseteq I' \subseteq I_r$, also $k \in I_r$,
so rules \rulename{LR3}, \rulename{LR4}, \rulename{LTag}, and
\rulename{LEnv} give
$\ctx \transi{\interaction{\pp}{\pq}{\lab[k]}} \ctxi$ for some $\ctxi$.
Take $\interi = \interaction{\pp}{\pq}{\lab[k]}$ (which need not
equal $\inter$ when $k \neq j$).

\emph{Case \rulename{GR2}.}
$\G = \Gvt{\pr}{\role{s}}{\B}{\G[0]}$ with
$\set{\pr,\role{s}}$ disjoint from the participants of $\inter$,
$\G[0] \transg{\inter} \Gi[0]$,
and $\Gi = \Gvt{\pr}{\role{s}}{\B}{\Gi[0]}$.
The head of $\unfold{\G}$ is the same $\pr,\role{s}$ pair, so the
argument of case \rulename{GR1} applies verbatim to
$\unfold{\ctx(\pr)}$ and $\unfold{\ctx(\role{s})}$:
$\ctx \transi{\interaction{\pr}{\role{s}}{\B}} \ctxi$
for some $\ctxi$.
Take $\interi = \interaction{\pr}{\role{s}}{\B}$.
We do not need the inductive hypothesis on
$\G[0] \transg{\inter} \Gi[0]$: $\ctx$ is lifted via the head
communication of $\G$, not the deeper step.

\emph{Case \rulename{GR5}.}
Symmetric to \rulename{GR2}, using
$\unfold{\G} = \GvtPair{\pr}{\role{s}}{\lab[i]:\G[i]}{i \in I}$ and
the argument of case \rulename{GR4} on the head choice.

\emph{Case \rulename{GR3}.}
$\G = \toc{\sfmu}\ty.\G[1]$ and the premise is
$\G[1]\sub{\toc{\sfmu}\ty.\G[1]}{\ty} \transg{\inter} \Gi$.
Since
$\unfold{\G} = \unfold{\G[1]\sub{\toc{\sfmu}\ty.\G[1]}{\ty}}$,
$\ctx \assoc \G$ also gives
$\ctx \assoc \G[1]\sub{\toc{\sfmu}\ty.\G[1]}{\ty}$ (the split of
$\ctx$ and the witness $\ctx[0]$ are the same, and $\pt{}$ and
balancedness are invariant under unfolding).
The inductive hypothesis on the premise yields the desired
$\interi$ and $\ctxi$.

In each case we obtain $\ctx \transi{\interi} \ctxi$.
By Lemma~\ref{lemma:completeness-association}, there exists a
balanced $\Gii$
with $\G \transg{\interi} \Gii$ and $\ctxi \assoc \Gii$.
\end{proof}

\associationimpliestraces*
\begin{proof}[Proof of Lemma~\ref{lemma:association-implies-traces}]
Let $\inter[0]\,\inter[1]\,\cdots \in \trace(\ctx)$, witnessed by
a maximal reduction sequence
$\ctx = \ctx[0] \transi{\inter[0]} \ctx[1] \transi{\inter[1]} \cdots$.
Set $\G[0] = \G$; then $\ctx[0] \assoc \G[0]$ with $\G[0]$ balanced.
By Lemma~\ref{lemma:completeness-association}, applied iteratively,
each step $\ctx[k] \transi{\inter[k]} \ctx[k{+}1]$ yields a balanced
$\G[k{+}1]$ with $\G[k] \transg{\inter[k]} \G[k{+}1]$
and $\ctx[k{+}1] \assoc \G[k{+}1]$,
giving $\G[0] \transg{\inter[0]} \G[1] \transg{\inter[1]} \cdots$
with the same labels.
It remains to check that the constructed sequence is maximal.
If the given sequence is infinite, so is the constructed one.
Otherwise it ends in a stuck $\ctx[n]$ with
$\ctx[n] \assoc \G[n]$;
were $\G[n]$ to have a transition,
Lemma~\ref{lemma:soundness-association} would yield a transition of
$\ctx[n]$, contradicting stuckness.
Hence the constructed sequence is maximal and
$\inter[0]\,\inter[1]\,\cdots \in \trace(\G)$.
\end{proof}

\soundnessimpliestraces*
\begin{proof}[Proof of Lemma~\ref{lemma:soundness-implies-traces}]
Let $\inter[0]\,\inter[1]\,\cdots \in \trace(\G)$, witnessed by
a maximal reduction sequence
$\G = \G[0] \transg{\inter[0]} \G[1] \transg{\inter[1]} \cdots$.
Set $\ctx[0] = \ctx$; then $\G[0] \cofullprojsubctx \ctx[0]$.
By Lemma~\ref{lemma:soundness-combined-relation}, applied iteratively,
there exists $\ctx[k{+}1]$ with $\ctx[k] \transi{\inter[k]} \ctx[k{+}1]$
and $\G[k{+}1] \cofullprojsubctx \ctx[k{+}1]$,
yielding $\ctx[0] \transi{\inter[0]} \ctx[1] \transi{\inter[1]} \cdots$.
For maximality, suppose the given sequence ends in a stuck $\G[n]$.
A global type whose unfolding has a communication head always has a
transition (by \rulename{GR1} or \rulename{GR4}, under
\rulename{GR3}), so $\unfold{\G[n]} = \gend$ and
$\pt{\G[n]} = \emptyset$; by
Lemma~\ref{lemma:completed-participants}, every $\ctx[n](\pp)$
unfolds to $\tend$, so $\ctx[n]$ is stuck.
Hence $\inter[0]\,\inter[1]\,\cdots \in \trace(\ctx)$.
\end{proof}

\soundnessofcombinedrelation*
\begin{proof}[Proof of Lemma~\ref{lemma:soundness-combined-relation}]
Assume $\G \cofullprojsubctx \ctx$ and $\G \transg{\interaction{\pp}{\pq}{\obs}} \Gi$.
We construct $\ctxi$ with $\ctx \transi{\interaction{\pp}{\pq}{\obs}} \ctxi$ and $\Gi \cofullprojsubctx \ctxi$.
Note that $\dom{\ctxi} = \dom{\ctx} \supseteq \pt{\G} \supseteq \pt{\Gi}$
in every case below, and that the hypothesis provides
$\G \cofullprojsub{t} \ctx(\role{t})$ for \emph{every}
$\role{t} \in \dom{\ctx}$, so it suffices to establish
$\Gi \cofullprojsub{t} \ctxi(\role{t})$ for every
$\role{t} \in \dom{\ctx}$.
The proof is by induction on the derivation of
$\G \transg{\interaction{\pp}{\pq}{\obs}} \Gi$.
By Lemma~\ref{lemma:inversion-global-reduction}, one of four cases holds.

\medskip
\emph{Cases~1 and~2 (head interaction fires).}

\emph{Case~1}: $\obs = \B$,
$\unfold{\G} = \Gvt{\pp}{\pq}{\B}{\G[0]}$, and $\Gi = \G[0]$.

For the sender $\pp$: from $\G \cofullprojsub{p} \ctx(\pp)$,
Lemma~\ref{lemma:inversion-combined-relation}(1) gives
$\unfold{\ctx(\pp)} = \tout{\pq}{\B}{\T[\pp]}$ with $\G[0] \cofullprojsub{p} \T[\pp]$.
Hence $\ctx(\pp) \transa{\qsend{\pq}{\B}} \T[\pp]$.

For the receiver $\pq$: from $\G \cofullprojsub{q} \ctx(\pq)$,
Lemma~\ref{lemma:inversion-combined-relation}(2) gives
$\unfold{\ctx(\pq)} = \tin{\pp}{\B}{\T[\pq]}$ with $\G[0] \cofullprojsub{q} \T[\pq]$.
Hence $\ctx(\pq) \transa{\qrecv{\pp}{\B}} \T[\pq]$.

Combining via rules \rulename{LTag} and \rulename{LEnv},
$\ctx \transi{\interaction{\pp}{\pq}{\B}} \ctxi$
where $\ctxi(\pp) = \T[\pp]$, $\ctxi(\pq) = \T[\pq]$, and $\ctxi(\pr) = \ctx(\pr)$ for $\pr \notin \set{\pp, \pq}$.

For each uninvolved $\pr \notin \set{\pp, \pq}$:
Lemma~\ref{lemma:inversion-combined-relation}(5) gives
$\G[0] \cofullprojsub{r} \ctx(\pr) = \ctxi(\pr)$.
Hence $\Gi \cofullprojsubctx \ctxi$.

\emph{Case~2}: $\obs = \lab[j]$,
$\unfold{\G} = \GvtPair{\pp}{\pq}{\lab[i] : \G[i]}{i \in I}$,
$j \in I$, $\Gi = \G[j]$.

For the sender $\pp$:
Lemma~\ref{lemma:inversion-combined-relation}(3) gives
$\unfold{\ctx(\pp)} = \tselsub{\pq}{\lab[i] : \T[i]}{i \in I}$ with
$\G[i] \cofullprojsub{p} \T[i]$ for all $i \in I$.
Since $j \in I$:
$\ctx(\pp) \transa{\qsend{\pq}{\lab[j]}} \T[j]$
and $\G[j] \cofullprojsub{p} \T[j]$.

For the receiver $\pq$:
Lemma~\ref{lemma:inversion-combined-relation}(4) gives
$J' \supseteq I$ with
$\unfold{\ctx(\pq)} = \tbrasub{\pp}{\lab[i] : \Ti[i]}{i \in J'}$ and
$\G[j] \cofullprojsub{q} \Ti[j]$ (since $j \in I \subseteq J'$).
Hence $\ctx(\pq) \transa{\qrecv{\pp}{\lab[j]}} \Ti[j]$.

Combining via \rulename{LTag} and \rulename{LEnv},
$\ctx \transi{\interaction{\pp}{\pq}{\lab[j]}} \ctxi$
where $\ctxi(\pp) = \T[j]$, $\ctxi(\pq) = \Ti[j]$, $\ctxi(\pr) = \ctx(\pr)$ for $\pr \notin \set{\pp, \pq}$.

For uninvolved $\pr$:
Lemma~\ref{lemma:inversion-combined-relation}(6) gives
$\G[i] \cofullprojsub{r} \ctx(\pr)$ for all $i$;
in particular $\G[j] \cofullprojsub{r} \ctx(\pr) = \ctxi(\pr)$.
Hence $\Gi \cofullprojsubctx \ctxi$.

\medskip
\emph{Case~3 (base-type prefix pass-through).}
$\unfold{\G} = \Gvt{\pr}{\role{s}}{\B}{\G[0]}$ with
$\set{\pp,\pq} \cap \set{\pr,\role{s}} = \emptyset$,
$\G[0] \transg{\interaction{\pp}{\pq}{\obs}} \Gi[0]$,
and $\Gi = \Gvt{\pr}{\role{s}}{\B}{\Gi[0]}$.

For each participant $\role{t} \in \dom{\ctx}$, apply
Lemma~\ref{lemma:inversion-combined-relation} to
$\G \cofullprojsub{t} \ctx(\role{t})$:
\begin{itemize}
  \item $\role{t} = \pr$: case~(1) gives
    $\unfold{\ctx(\pr)} = \tout{\role{s}}{\B}{\T[\pr]}$ with
    $\G[0] \cofullprojsub{r} \T[\pr]$.
  \item $\role{t} = \role{s}$: case~(2) gives
    $\unfold{\ctx(\role{s})} = \tin{\pr}{\B}{\T[\role{s}]}$ with
    $\G[0] \cofullprojsub{s} \T[\role{s}]$.
  \item $\role{t} \notin \set{\pr, \role{s}}$: case~(5) gives
    $\G[0] \cofullprojsub{t} \ctx(\role{t})$.
\end{itemize}
Define $\ctx_0$ by $\ctx_0(\pr) = \T[\pr]$,
$\ctx_0(\role{s}) = \T[\role{s}]$, and
$\ctx_0(\role{t}) = \ctx(\role{t})$ for $\role{t} \notin \set{\pr, \role{s}}$.
Then $\G[0] \cofullprojsubctx \ctx_0$.
Since $\set{\pp,\pq} \cap \set{\pr,\role{s}} = \emptyset$,
we have $\ctx_0(\pp) = \ctx(\pp)$ and $\ctx_0(\pq) = \ctx(\pq)$.

By the induction hypothesis (the derivation of
$\G[0] \transg{\interaction{\pp}{\pq}{\obs}} \Gi[0]$ is strictly smaller),
there exists $\ctxi_0$ with
$\ctx_0 \transi{\interaction{\pp}{\pq}{\obs}} \ctxi_0$ and
$\Gi[0] \cofullprojsubctx \ctxi_0$.

The step $\ctx_0 \transi{\interaction{\pp}{\pq}{\obs}} \ctxi_0$
only modifies $\pp$ and $\pq$.
Since $\ctx_0(\pp) = \ctx(\pp)$ and $\ctx_0(\pq) = \ctx(\pq)$,
the same one-sided transitions exist in~$\ctx$, giving
$\ctx \transi{\interaction{\pp}{\pq}{\obs}} \ctxi$
where $\ctxi(\pp) = \ctxi_0(\pp)$, $\ctxi(\pq) = \ctxi_0(\pq)$,
and $\ctxi(\role{t}) = \ctx(\role{t})$ for $\role{t} \notin \set{\pp,\pq}$.

It remains to show $\Gi \cofullprojsubctx \ctxi$
where $\Gi = \Gvt{\pr}{\role{s}}{\B}{\Gi[0]}$.
\begin{itemize}
  \item For $\role{t} \notin \set{\pr,\role{s}}$:
    by \rulename{SPR3},
    $\Gi[0] \cofullprojsub{t} \ctxi(\role{t})$ implies
    $\Gi \cofullprojsub{t} \ctxi(\role{t})$.
    From $\Gi[0] \cofullprojsubctx \ctxi_0$:
    if $\role{t} \in \set{\pp,\pq}$, then
    $\Gi[0] \cofullprojsub{t} \ctxi_0(\role{t}) = \ctxi(\role{t})$;
    if $\role{t} \notin \set{\pp,\pq}$, then
    $\ctxi(\role{t}) = \ctx(\role{t}) = \ctx_0(\role{t}) = \ctxi_0(\role{t})$,
    and $\Gi[0] \cofullprojsub{t} \ctxi_0(\role{t}) = \ctxi(\role{t})$.
  \item For $\role{t} = \pr$:
    since $\pr \notin \set{\pp,\pq}$, we have
    $\ctxi(\pr) = \ctx(\pr)$ and
    $\ctxi_0(\pr) = \ctx_0(\pr) = \T[\pr]$.
    From $\Gi[0] \cofullprojsubctx \ctxi_0$:
    $\Gi[0] \cofullprojsub{r} \T[\pr]$.
    By \rulename{SPR1},
    $\Gvt{\pr}{\role{s}}{\B}{\Gi[0]} \cofullprojsub{r} \tout{\role{s}}{\B}{\T[\pr]}$.
    Since $\unfold{\ctx(\pr)} = \tout{\role{s}}{\B}{\T[\pr]}$, by closure under
    \rulename{SPR10}:
    $\Gi \cofullprojsub{r} \ctx(\pr) = \ctxi(\pr)$.
  \item For $\role{t} = \role{s}$: symmetric, using \rulename{SPR2} and \rulename{SPR10}.
\end{itemize}

\medskip
\emph{Case~4 (branching prefix pass-through).}
$\unfold{\G} = \GvtPair{\pr}{\role{s}}{\lab[i] : \G[i]}{i \in I}$ with
$\set{\pp,\pq} \cap \set{\pr,\role{s}} = \emptyset$,
$\G[i] \transg{\interaction{\pp}{\pq}{\obs}} \Gi[i]$ for all $i$,
and $\Gi = \GvtPair{\pr}{\role{s}}{\lab[i] : \Gi[i]}{i \in I}$.
The argument is analogous to Case~3, using
Lemma~\ref{lemma:inversion-combined-relation} cases~(3), (4), and~(6)
instead of~(1), (2), and~(5),
and rules \rulename{SPR4}/\rulename{SPR5}/\rulename{SPR6}
instead of \rulename{SPR1}/\rulename{SPR2}/\rulename{SPR3}.
The induction hypothesis is applied to each sub-derivation
$\G[i] \transg{\interaction{\pp}{\pq}{\obs}} \Gi[i]$,
yielding compatible context steps that agree on $\pp$ and $\pq$
(since the combined relation determines the local types of $\pp$ and $\pq$
independently of the branch index~$i$).
\end{proof}

\subsection{Proof of Subject Reduction for the Top-Down System}
\label{app:typing-sr}

\begin{proof}[Proof of Theorem~\ref{thm:topdown}(1)]
The proof uses Theorem~\ref{thm:topdown}(2), which does not depend on
subject reduction.
Assume $\provestop \M : \ctx$ with
$\M = \prod_{i \in I}\proctag{\pp[i]}{\PP[i]}$.
By inversion of \rulename{T-Sess}, we have
$\ctx = \set{\ptag{\pp[i]}{\T[i]}}_{i \in I}$ where
$\vdash \PP[i] : \T[i]$ and $\G \cofullproj{\pp[i]} \T[i]$ for all
$i \in I$, for some balanced $\G$ with
$\pt{\G} \subseteq \set{\pp[i] \mid i \in I}$.

We first observe that $\ctx \assoc \G$.
Split $\ctx = \ctxact, \ctxterm$, where $\ctxact$ is the restriction
of $\ctx$ to $\pt{\G}$.
Clause~(1) of Definition~\ref{def:association} holds with witness
$\ctx[0] = \ctxact$, whose entries are the projections themselves,
using reflexivity of $\subt$.
Clause~(2) holds because $\G \cofullproj{\pp[i]} \T[i]$ with
$\pp[i] \notin \pt{\G}$ forces $\unfold{\T[i]} = \tend$, since
$\G \cofullprojsub{\pp[i]} \T[i]$ by
Lemma~\ref{lemma:projection-in-combined} and
Lemma~\ref{lemma:completed-participants}(1) applies.

By Theorem~\ref{thm:topdown}(2a), $\slive(\ctx)$, hence
$\safe(\ctx)$ and $\live(\ctx)$.
The premises $\vdash \PP[i] : \T[i]$ together with $\safe(\ctx)$
allow rule \rulename{B-Sess}, so $\provesbot \M : \ctx$.
By Theorem~\ref{thm:bottomup}(1), $\M \red^\ast \Mi$ yields
$\provesbot \Mi : \ctxi$ with $\ctx \transi{}^{*} \ctxi$ and
$\live(\ctxi)$, where
$\Mi = \prod_{i \in I}\proctag{\pp[i]}{\PPi[i]}$.
By induction on the length of $\ctx \transi{}^{*} \ctxi$, applying
Lemma~\ref{lemma:completeness-association} at each step, there
exists a balanced $\Gi$ with $\ctxi \assoc \Gi$.

It remains to retype $\Mi$ top-down, following the construction in
the proof of Theorem~\ref{thm:proc:completeness}(2).
By inversion of \rulename{B-Sess},
$\vdash \PPi[i] : \ctxi(\pp[i])$ for all $i \in I$.
By Definition~\ref{def:association}, $\ctxi$ splits into
$\ctxact, \ctxterm$, and clause~(1) provides $\ctx[0]$ with
$\dom{\ctx[0]} = \pt{\Gi}$,
$\Gi \cofullproj{p} \ctx[0](\pp)$ for all $\pp \in \pt{\Gi}$, and
$\ctxact \subt \ctx[0]$.
Define $\ctxii$ with $\dom{\ctxii} = \dom{\ctxi}$ by
$\ctxii(\pp) = \ctx[0](\pp)$ for $\pp \in \pt{\Gi}$ and
$\ctxii(\pp) = \tend$ otherwise.
Then $\ctxi \subt \ctxii$, since on $\pt{\Gi}$ this is
$\ctxact \subt \ctx[0]$, and elsewhere
$\unfold{\ctxi(\pp)} = \tend$ gives $\ctxi(\pp) \subt \tend$
(rules \rulename{S5} and \rulename{S7}).
Rule \rulename{Sub} lifts each typing to
$\vdash \PPi[i] : \ctxii(\pp[i])$.
Moreover $\Gi \cofullproj{\pp[i]} \ctxii(\pp[i])$ holds for each
$i$, by the witness property for $\pp[i] \in \pt{\Gi}$ and by
Lemma~\ref{lemma:completed-participants}(4) for
$\pp[i] \notin \pt{\Gi}$, and $\Gi$ is balanced.
Since $\pt{\Gi} \subseteq \dom{\ctxi}$, rule \rulename{T-Sess}
yields $\provestop \Mi : \ctxii$ with
$\ctx \transi{}^{*} \ctxi \subt \ctxii$, as required.
\end{proof}

\section{Benchmark Global Type Table}
\label{app:example-global-types}

Rows (d), (e), and (u) show the MapReduce, Independent Workers, and
Ring families schematically; rows (m)--(t) reproduce the remaining
uncopied Haskell catalogue entries, including the concrete Coinductive
Full instances.
Entries are labelled (a)--(u) for reference from
Table~\ref{tab:benchmarks}.

{\scriptsize
\setlength{\tabcolsep}{3pt}
\setlength{\LTleft}{0pt}
\setlength{\LTright}{0pt plus 1fill}
\setlength{\LTcapwidth}{\textwidth}
\newcommand{\ccfstep}[1]{\role{p}\gar\role{q}\;\goc{\{}\labname{a}\gc #1\goc{\}}}
\begin{longtable}{@{} p{0.60\textwidth} p{0.36\textwidth} @{}}
\caption{Benchmark global type table used in the evaluation.}
\label{tab:example-global-types}\\
\toprule
\endfirsthead
\toprule
\multicolumn{2}{@{}l}{\small\emph{Benchmark Global Type Table (continued)}} \\
\midrule
\endhead
\midrule
\multicolumn{2}{r@{}}{\footnotesize Continued on next page} \\
\endfoot
\bottomrule
\endlastfoot
\multicolumn{2}{@{}l}{\small\textbf{(a) Odd-Even}~\cite[Ex.~2.1]{Li2023}} \\*[2pt]
$\G[\mathit{oe}] = \pp\gar\pq\;
\color{gtypecolor}\left\{\color{black}
\begin{array}{@{}l@{}}
\labname{o} \gc \pq\gar\pr\;\goc{\{}\labname{o} \gc
  \toc{\sfmu}\ty[1].\,\pp\gar\pq\;
  \color{gtypecolor}\left\{\color{black}
  \begin{array}{@{}l@{}}
  \labname{o} \gc \GvtPairs{\pq}{\pr}{\labname{o} \gc
    \GvtPairs{\pq}{\pr}{\labname{o} \gc \ty[1]}} \\
  \labname{b} \gc \GvtPairs{\pq}{\pr}{\labname{b} \gc
    \GvtPairs{\pr}{\pp}{\labname{o} \gc \gend}}\goc{\}}
  \end{array}
  \color{gtypecolor}\right\}\color{black} \\[4pt]
\labname{m} \gc \toc{\sfmu}\ty[2].\,\pp\gar\pq\;
  \color{gtypecolor}\left\{\color{black}
  \begin{array}{@{}l@{}}
  \labname{o} \gc \GvtPairs{\pq}{\pr}{\labname{o} \gc
    \GvtPairs{\pq}{\pr}{\labname{o} \gc \ty[2]}} \\
  \labname{b} \gc \GvtPairs{\pq}{\pr}{\labname{b} \gc
    \GvtPairs{\pr}{\pp}{\labname{m} \gc \gend}}
  \end{array}
  \color{gtypecolor}\right\}\color{black}
\end{array}
\color{gtypecolor}\right\}$
&
\parbox{\linewidth}{\footnotesize\raggedright
Three participants $\pp$, $\pq$, $\pr$.
The outer choice ($\labname{o}$/$\labname{m}$) selects between
two recursive modes.
In each mode, $\pq$ relays to~$\pr$ and the
protocol either continues ($\labname{o}$) or
terminates ($\labname{b}$).
Requires coinductive full-merging projection.}
\\
\midrule
\multicolumn{2}{@{}l}{\small\textbf{(b) OAuth}~\cite[Ex.~1]{Scalas2019}} \\*[2pt]
$\G[\mathit{oa}] =
\role{s}\gar\role{c}\;
\color{gtypecolor}\left\{\color{black}
\begin{array}{@{}l@{}}
\labname{login} \gc \role{c}\gar\role{a}\;\goc{\big\{}\labname{password} \gc \\
\quad \Gvt{\role{c}}{\role{a}}{\tstring}
  \role{a}\gar\role{s}\;\goc{\{}\labname{auth} \gc \\
\qquad \Gvt{\role{a}}{\role{s}}{\tbool} \gend\goc{\}}\goc{\big\}} \\
\labname{auth} \gc \GvtPairs{\role{c}}{\role{a}}{\labname{quit} \gc \gend}
\end{array}
\color{gtypecolor}\right\}$
&
\parbox{\linewidth}{\footnotesize\raggedright
Three participants: $\role{s}$ (server),
$\role{c}$ (client), $\role{a}$ (auth.\ provider).
Server offers login or direct auth.
On login, client sends password to auth provider,
who replies with a boolean to the server.}
\\
\midrule
\multicolumn{2}{@{}l}{\small\textbf{(c) Recursive Two-Buyer}~\cite[Ex.~2]{Scalas2019}} \\*[2pt]
$\begin{aligned}
\G[\mathit{tb}] ={} &
\role{a}\gar\role{s}\;\goc{\big\{}\labname{query} \gc
\Gvt{\role{a}}{\role{s}}{\tstring} \\
& \quad \role{s}\gar\role{a}\;\goc{\{}\labname{price} \gc
\Gvt{\role{s}}{\role{a}}{\tint} {} \\
& \qquad \toc{\sfmu}\ty.\;\role{a}\gar\role{b}\;
\color{gtypecolor}\left\{\color{black}
\begin{array}{@{}l@{}}
\labname{split} \gc \role{b}\gar\role{a}\;
  \color{gtypecolor}\left\{\color{black}
  \begin{array}{@{}l@{}}
  \labname{yes} \gc \GvtPairs{\role{a}}{\role{s}}{\labname{buy} \gc \gend} \\
  \labname{no} \gc \ty
  \end{array}
  \color{gtypecolor}\right\}\color{black} \\[3pt]
\labname{cancel} \gc \GvtPairs{\role{a}}{\role{s}}{\labname{no} \gc \gend}
\end{array}
\color{gtypecolor}\right\}\goc{\}}\goc{\big\}}
\end{aligned}$
&
\parbox{\linewidth}{\footnotesize\raggedright
Three participants: $\role{a}$ (buyer~A),
$\role{b}$ (buyer~B), $\role{s}$ (seller).
A queries the seller, who quotes a price;
A asks B to split. B accepts or declines (recurse).
A may also cancel outright.
Not balanced, but projectable.}
\\
\midrule
\multicolumn{2}{@{}l}{\small\textbf{(d) MapReduce}~\cite[Ex.~3]{Scalas2019}} \\*[2pt]
$\begin{aligned}
\G[\mathit{mr}] ={} & \toc{\sfmu}\ty.\;
\role{m}\gar\role{w_1}\;\goc{\{}\labname{datum} \gc
  \Gvt{\role{m}}{\role{w_1}}{\tint} \cdots \\
& \role{m}\gar\role{w_n}\;\goc{\{}\labname{datum} \gc
  \Gvt{\role{m}}{\role{w_n}}{\tint} \\
& \Gvt{\role{w_1}}{\role{r}}{\tint} \cdots\;
  \Gvt{\role{w_n}}{\role{r}}{\tint} \\
& \role{r}\gar\role{m}\;
\color{gtypecolor}\left\{\color{black}
\begin{array}{@{}l@{}}
\labname{continue} \gc \Gvt{\role{r}}{\role{m}}{\tint} \ty \\
\labname{stop} \gc \role{m}\gar\role{w_1}\;\goc{\{}\labname{stop} \gc \cdots \\
\quad \role{m}\gar\role{w_n}\;\goc{\{}\labname{stop} \gc \gend\goc{\}\cdots\}}\goc{\}\cdots\}}
\end{array}
\color{gtypecolor}\right\}
\end{aligned}$
&
\parbox{\linewidth}{\footnotesize\raggedright
$n{+}2$ participants: $\role{m}$ (manager),
$\role{w_1},\ldots,\role{w_n}$ (workers), $\role{r}$ (reducer).
Manager dispatches data, workers return results,
reducer continues or stops.
Running example (\SEC{sec:overview});
the catalogue instances are MapReduce($5$)--($7$).}
\\
\midrule
\multicolumn{2}{@{}l}{\small\textbf{(e) Independent Workers}~\cite[Ex.~4]{Scalas2019}} \\*[2pt]
$\begin{aligned}
\G[\mathit{iw}] ={} &
\role{s}\gar\role{wa_1}\;\goc{\{}\labname{datum} \gc
\Gvt{\role{s}}{\role{wa_1}}{\tint} \cdots \\
& \role{s}\gar\role{wa_n}\;\goc{\{}\labname{datum} \gc
\Gvt{\role{s}}{\role{wa_n}}{\tint} {} \\
& \toc{\sfmu}\ty.\;
\role{wa_i}\gar\role{wb_i}\;
\goc{\{}
\labname{datum} \gc \Gvt{\role{wa_i}}{\role{wb_i}}{\tint} \\
& \quad \role{wb_i}\gar\role{wc_i}\;\goc{\{}\labname{datum} \gc
  \Gvt{\role{wb_i}}{\role{wc_i}}{\tint} \\
& \qquad \Gvt{\role{wc_i}}{\role{wa_i}}{\tint}\,\ty \goc{\}}\text{,} \\
& \quad \labname{stop} \gc \GvtPairs{\role{wb_i}}{\role{wc_i}}{\labname{stop} \gc \gend}
\goc{\}}\goc{\}\cdots\}}
\end{aligned}$
&
\parbox{\linewidth}{\footnotesize\raggedright
$3n{+}1$ participants: $\role{s}$ (server) and
$n$ independent pipelines, each with
$\role{wa_i}$, $\role{wb_i}$, $\role{wc_i}$.
Server dispatches integers; each pipeline
loops with datum/stop until termination.
Shown schematically: the syntax has no parallel operator, so the
catalogue interleaves the pipeline-local choices; the
instances are Independent Workers($7$) and ($10$).}
\\
\midrule
\multicolumn{2}{@{}l}{\small\textbf{(f) Semantic Recursion}~\cite{Tirore2023}} \\*[2pt]
$\begin{aligned}
\G[\mathit{itp}] ={} & \toc{\sfmu}\ty.\;
\Gvt{\role{a}}{\role{b}}{\tstring} {} \\
& \toc{\sfmu}\tyi.\;
\role{c}\gar\role{d}\;
\color{gtypecolor}\left\{\color{black}
\begin{array}{@{}l@{}}
\labname{left} \gc \ty \\
\labname{right} \gc \Gvt{\role{a}}{\role{b}}{\tstring} \tyi
\end{array}
\color{gtypecolor}\right\}
\end{aligned}$
&
\parbox{\linewidth}{\footnotesize\raggedright
Four participants: $\role{a}$, $\role{b}$, $\role{c}$, $\role{d}$.
Two nested recursions: the inner loop
($\tyi$) between $\role{c}$ and $\role{d}$
interleaves with sends from $\role{a}$ to~$\role{b}$.
Choosing $\labname{left}$ restarts the outer loop.
Requires coinductive projection.}
\\
\midrule
\multicolumn{2}{@{}l}{\small\textbf{(g) Simple Travel Agency}~\cite[Fig.~1(a)]{YoshidaGheri2020}} \\*[2pt]
$\begin{aligned}
\G[\mathit{sta}] ={} &
\Gvt{\role{c}}{\role{a}}{\tstring}
\Gvt{\role{a}}{\role{c}}{\tint} {} \\
& \role{c}\gar\role{a}\;
\color{gtypecolor}\left\{\color{black}
\begin{array}{@{}l@{}}
\labname{accept} \gc \Gvt{\role{c}}{\role{a}}{\tstring}
  \Gvt{\role{a}}{\role{c}}{\tint}\,\gend \\
\labname{reject} \gc \gend
\end{array}
\color{gtypecolor}\right\}
\end{aligned}$
&
\parbox{\linewidth}{\footnotesize\raggedright
Two participants: $\role{c}$ (client) and $\role{a}$ (agency).
The client requests a quote and receives a reply,
then either accepts with booking details or rejects.}
\\
\midrule
\multicolumn{2}{@{}l}{\small\textbf{(h) Better Travel Agency}~\cite[Fig.~1(b)]{YoshidaGheri2020}} \\*[2pt]
$\begin{aligned}
\G[\mathit{bta}] ={} & \toc{\sfmu}\ty.\;
\Gvt{\role{c}}{\role{a}}{\tstring}
\Gvt{\role{a}}{\role{c}}{\tint} {} \\
& \role{c}\gar\role{a}\;
\color{gtypecolor}\left\{\color{black}
\begin{array}{@{}l@{}}
\labname{accept} \gc \Gvt{\role{c}}{\role{a}}{\tstring}
  \Gvt{\role{a}}{\role{c}}{\tint}\,\gend \\
\labname{retry} \gc \ty \\
\labname{reject} \gc \gend
\end{array}
\color{gtypecolor}\right\}
\end{aligned}$
&
\parbox{\linewidth}{\footnotesize\raggedright
The recursive variant of (g).
After the initial request and quote, the client may accept,
reject, or retry and return to the start of the protocol.}
\\
\midrule
\multicolumn{2}{@{}l}{\small\textbf{(i) Inductive Plain}~\cite[Ex.~4.8]{thien-nobuko-popl-25}} \\*[2pt]
$\begin{aligned}
\G[\mathit{ip}] ={} & \toc{\sfmu}\ty.\;
\role{q}\gar\role{r}\;
\color{gtypecolor}\left\{\color{black}
\begin{array}{@{}l@{}}
\lab[1] \gc \role{r}\gar\role{p}\;\goc{\{}\lab[1] \gc \ty \goc{\}} \\
\lab[2] \gc \role{r}\gar\role{p}\;\goc{\{}\lab[1] \gc \ty \goc{\}}
\end{array}
\color{gtypecolor}\right\}
\end{aligned}$
&
\parbox{\linewidth}{\footnotesize\raggedright
Three participants: $\role{q}$ chooses and $\role{r}$ relays to $\role{p}$.
Both branches present the same label to $\role{p}$,
so plain merge already succeeds.}
\\
\midrule
\multicolumn{2}{@{}l}{\small\textbf{(j) Inductive Full}~\cite[Ex.~4.8]{thien-nobuko-popl-25}} \\*[2pt]
$\begin{aligned}
\G[\mathit{if}] ={} & \toc{\sfmu}\ty.\;
\role{q}\gar\role{r}\;
\color{gtypecolor}\left\{\color{black}
\begin{array}{@{}l@{}}
\lab[1] \gc \role{r}\gar\role{p}\;\goc{\{}\lab[1] \gc \ty \goc{\}} \\
\lab[2] \gc \role{r}\gar\role{p}\;\goc{\{}\lab[2] \gc \gend \goc{\}}
\end{array}
\color{gtypecolor}\right\}
\end{aligned}$
&
\parbox{\linewidth}{\footnotesize\raggedright
The same overall shape as (i), but the second branch relays
a different label to $\role{p}$.
This is the smallest benchmark in the catalogue where full merge matters.}
\\
\midrule
\multicolumn{2}{@{}l}{\small\textbf{(k) Ring}~\cite{Castro2026}} \\*[2pt]
$\begin{aligned}
\G[\mathit{ring}] ={} &
\role{a}\gar\role{b}\;
\color{gtypecolor}\left\{\color{black}
\begin{array}{@{}l@{}}
\labname{AppThenGet} \gc \role{b}\gar\role{c}\;
  \color{gtypecolor}\left\{\color{black}
  \begin{array}{@{}l@{}}
  \labname{AppThenGet} \gc \role{c}\gar\role{a}\;\goc{\{}\labname{Val} \gc \gend \goc{\}}
  \end{array}
  \color{gtypecolor}\right\}\color{black} \\[3pt]
\labname{App} \gc \role{b}\gar\role{c}\;
  \color{gtypecolor}\left\{\color{black}
  \begin{array}{@{}l@{}}
  \labname{App} \gc \role{a}\gar\role{c}\;\goc{\big\{}\labname{Get} \gc
    \role{c}\gar\role{a}\;\goc{\{}\labname{Val} \gc \gend \goc{\}}\goc{\big\}}
  \end{array}
  \color{gtypecolor}\right\}\color{black}
\end{array}
\color{gtypecolor}\right\}
\end{aligned}$
&
\parbox{\linewidth}{\footnotesize\raggedright
Three participants arranged in a ring.
One branch sends an application then retrieves a value,
while the other performs the application and retrieval in a different order.}
\\
\midrule
\multicolumn{2}{@{}l}{\small\textbf{(l) Adder}~\cite[Fig.~1(a)]{FASE16EndpointAPI}} \\*[2pt]
$\begin{aligned}
\G[\mathit{ad}] ={} & \toc{\sfmu}\ty.\;
\role{c}\gar\role{s}\;
\color{gtypecolor}\left\{\color{black}
\begin{array}{@{}l@{}}
\labname{Add} \gc \Gvt{\role{c}}{\role{s}}{\tint}
  \Gvt{\role{c}}{\role{s}}{\tint} \\
\quad \role{s}\gar\role{c}\;\goc{\{}\labname{Res} \gc
    \Gvt{\role{s}}{\role{c}}{\tint}\,\ty \goc{\}} \\
\labname{Bye} \gc \role{s}\gar\role{c}\;\goc{\{}\labname{Bye} \gc \gend \goc{\}}
\end{array}
\color{gtypecolor}\right\}
\end{aligned}$
&
\parbox{\linewidth}{\footnotesize\raggedright
Two participants: client $\role{c}$ and server $\role{s}$.
The client either submits two integers and receives their sum,
or terminates the session.}
\\
\midrule
\multicolumn{2}{@{}l}{\small\textbf{(m) Company Communication}~\cite[Fig.~9]{DBLP:conf/ecoop/GheriLSTY22}} \\*[2pt]
$\begin{aligned}
\G[\mathit{cc}] ={} &
\role{d}\gar\role{ad}\;\goc{\Bigg\{}\labname{prod}\gc
\role{d}\gar\role{s}\;\goc{\Big\{}\labname{prod}\gc \\
& \quad
\role{d}\gar\role{f_1}\;\goc{\big\{}\labname{prod}\gc
\role{f_1}\gar\role{f_2}\;\goc{\{}\labname{prod}\gc
\tfix{\ty}\; \\
& \qquad
\role{f_2}\gar\role{f_1}\;
\color{gtypecolor}\left\{\color{black}
\begin{array}{@{}l@{}}
\labname{price}\gc
  \role{f_1}\gar\role{d}\;\goc{\big\{}\labname{ok}\gc
  \GvtPairs{\role{d}}{\role{ad}}{\labname{go}\gc \\
\quad
  \GvtPairs{\role{f_1}}{\role{s}}{\labname{price}\gc
  \GvtPairs{\role{s}}{\role{w}}{\labname{publish}\gc \gend}}}\goc{\big\}} \\
\labname{wait}\gc
  \role{f_1}\gar\role{d}\;\goc{\big\{}\labname{wait}\gc
  \GvtPairs{\role{d}}{\role{ad}}{\labname{wait}\gc \\
\quad
  \GvtPairs{\role{f_1}}{\role{s}}{\labname{wait}\gc
  \GvtPairs{\role{s}}{\role{w}}{\labname{wait}\gc \ty}}}\goc{\big\}}
\end{array}
\color{gtypecolor}\right\}\goc{\}}\goc{\big\}}\goc{\Big\}}\goc{\Bigg\}}
\end{aligned}$
&
\parbox{\linewidth}{\footnotesize\raggedright
Six participants: $\role{d}$ (director), $\role{ad}$ (admin), $\role{s}$ (seller),
$\role{f_1},\role{f_2}$ (factories), $\role{w}$ (warehouse).
Loops on \texttt{wait} until a price is accepted.}
\\
\midrule
\multicolumn{2}{@{}l}{\small\textbf{(n) Online Wallet}~\cite[Fig.~1]{DBLP:conf/rv/NeykovaYH13}} \\*[2pt]
$\begin{aligned}
\G[\mathit{ow}] ={} &
\role{c}\gar\role{a}\;\goc{\big\{}\labname{login}\gc
\Gvt{\role{c}}{\role{a}}{\tstring}
\Gvt{\role{c}}{\role{a}}{\tstring} \\
& \quad
\role{a}\gar\role{c}\;
\color{gtypecolor}\left\{\color{black}
\begin{array}{@{}l@{}}
\labtext{login\_ok}\gc
  \role{a}\gar\role{s}\;\goc{\{}\labtext{login\_ok}\gc
  \tfix{\ty}\; \\
\quad
  \GvtPairs{\role{s}}{\role{c}}{\labname{account}\gc
  \Gvt{\role{s}}{\role{c}}{\tint}
  \Gvt{\role{s}}{\role{c}}{\tint} \\
\quad
  \role{c}\gar\role{s}\;
  \color{gtypecolor}\left\{\color{black}
  \begin{array}{@{}l@{}}
  \labname{pay}\gc
    \Gvt{\role{c}}{\role{s}}{\tstring}
    \Gvt{\role{c}}{\role{s}}{\tint}\ty \\
  \labname{quit}\gc \gend
  \end{array}
  \color{gtypecolor}\right\}\color{black}}\goc{\}} \\
\labtext{login\_fail}\gc
  \Gvt{\role{a}}{\role{c}}{\tstring} \\
\quad
  \GvtPairs{\role{a}}{\role{s}}{\labtext{login\_fail}\gc
  \Gvt{\role{a}}{\role{s}}{\tstring}\gend}
\end{array}
\color{gtypecolor}\right\}\goc{\big\}}
\end{aligned}$
&
\parbox{\linewidth}{\footnotesize\raggedright
Three participants: $\role{c}$ (client),
$\role{a}$ (authenticator), $\role{s}$ (service).
After login, the service reports account
state and accepts payments until quit;
login failure terminates.}
\\
\midrule
\multicolumn{2}{@{}l}{\small\textbf{(o) Distributed Logging}~\cite[Fig.~16]{DBLP:conf/ecoop/LagaillardieNY22}} \\*[2pt]
$\begin{aligned}
\G[\mathit{dl}] ={} &
\role{c}\gar\role{s}\;\goc{\big\{}\labname{Start}\gc
\Gvt{\role{c}}{\role{s}}{\tint}
\tfix{\ty}\; \\
& \quad
\role{s}\gar\role{c}\;
\color{gtypecolor}\left\{\color{black}
\begin{array}{@{}l@{}}
\labname{Success}\gc \Gvt{\role{s}}{\role{c}}{\tint}\ty \\
\labname{Failure}\gc
  \Gvt{\role{s}}{\role{c}}{\tint}
  \role{c}\gar\role{s}\;
  \color{gtypecolor}\left\{\color{black}
  \begin{array}{@{}l@{}}
  \labname{Restart}\gc \Gvt{\role{c}}{\role{s}}{\tint}\ty \\
  \labname{Stop}\gc \Gvt{\role{c}}{\role{s}}{\tint}\gend
  \end{array}
  \color{gtypecolor}\right\}\color{black}
\end{array}
\color{gtypecolor}\right\}\goc{\big\}}
\end{aligned}$
&
\parbox{\linewidth}{\footnotesize\raggedright
Two participants: $\role{c}$ (client) and $\role{s}$ (server).
After starting, the server reports success
(recur) or failure (client restarts or stops).}
\\
\midrule
\multicolumn{2}{@{}l}{\small\textbf{(p) E-Voting}~\cite{DBLP:conf/ecoop/LagaillardieNY22}} \\*[2pt]
$\begin{aligned}
\G[\mathit{ev}] ={} &
\role{v}\gar\role{s}\;\goc{\big\{}\labname{Authenticate}\gc
\Gvt{\role{v}}{\role{s}}{\tstring} \\
& \quad
\role{s}\gar\role{v}\;
\color{gtypecolor}\left\{\color{black}
\begin{array}{@{}l@{}}
\labname{Ok}\gc
  \Gvt{\role{s}}{\role{v}}{\tstring} \\
\quad
  \role{v}\gar\role{s}\;
  \color{gtypecolor}\left\{\color{black}
  \begin{array}{@{}l@{}}
  \labname{Yes}\gc
    \Gvt{\role{v}}{\role{s}}{\tstring} \\
  \quad
    \GvtPairs{\role{s}}{\role{v}}{\labname{Result}\gc
    \Gvt{\role{s}}{\role{v}}{\tint}\gend} \\
  \labname{No}\gc
    \Gvt{\role{v}}{\role{s}}{\tstring} \\
  \quad
    \GvtPairs{\role{s}}{\role{v}}{\labname{Result}\gc
    \Gvt{\role{s}}{\role{v}}{\tint}\gend}
  \end{array}
  \color{gtypecolor}\right\}\color{black} \\
\labname{Reject}\gc \Gvt{\role{s}}{\role{v}}{\tstring}\gend
\end{array}
\color{gtypecolor}\right\}\goc{\big\}}
\end{aligned}$
&
\parbox{\linewidth}{\footnotesize\raggedright
Two participants: $\role{v}$ (voter) and $\role{s}$ (server).
The voter authenticates; on success
casts a Yes/No vote and receives
the result. On reject, terminates.}
\\
\midrule
\multicolumn{2}{@{}l}{\small\textbf{(q)--(t) Coinductive Full($n$)}~\cite[Thm.~4.24]{thien-nobuko-popl-25},\; $n \in \{2,3,4,5\}$} \\*[2pt]
\multicolumn{2}{@{}l}{\footnotesize
Let $A_1(\ty)=\ccfstep{\ty}$ and
$A_{m+1}(\ty)=\ccfstep{A_m(\ty)}$.} \\*[2pt]
$\begin{aligned}
\G[\mathit{cf}_n] ={} &
\GvtPair{\role{p}}{\role{r}}{\lab[k]\gc \G[k]}{1 \le k \le n} \\
\G[k] ={} &
\tfix{\ty}\;
\role{p}\gar\role{q}\;
\color{gtypecolor}\left\{\color{black}
\begin{array}{@{}l@{}}
\labname{a}\gc A_{d_k}(\ty) \\
\labname{b}\gc \GvtPairs{\role{p}}{\role{q}}{\lab[k]\gc \gend}
\end{array}
\color{gtypecolor}\right\}
\end{aligned}$
&
\parbox{\linewidth}{\footnotesize\raggedright
Three participants: $\role{p}$, $\role{q}$, $\role{r}$.
Branch depths $(d_1,\ldots,d_n)$ are the
first $n$ primes: $(2,3)$, $(2,3,5)$,
$(2,3,5,7)$, $(2,3,5,7,11)$.}
\\
\midrule
\multicolumn{2}{@{}l}{\small\textbf{(u) Ring($n$)},\; $n \in \{5,8,12,15\}$} \\*[2pt]
$\begin{aligned}
\G[\mathit{ring}_n] ={} & \tfix{\ty}\;
\Gvt{\role{p_1}}{\role{p_2}}{\tint}
\Gvt{\role{p_2}}{\role{p_3}}{\tint} \\
& \quad \cdots\;
\Gvt{\role{p_n}}{\role{p_1}}{\tint}\,\ty
\end{aligned}$
&
\parbox{\linewidth}{\footnotesize\raggedright
$n$ participants $\role{p_1},\ldots,\role{p_n}$ arranged in a cycle:
each round passes an integer once around the ring and repeats.
Projectable by all four EPP algorithms.}
\\
\end{longtable}
}

\begin{example}
\label{ex:binary-counter}
The Binary Counter family exhibits the unavoidable exponential growth
of the synthesised global type with the number of participants.
Let $\ctx[b] = \set{\ptag{\pp[i]}{\T[\pp[i]]}}_{i \in P}$ where
$P = \{0, 1, \ldots, n-1\}$ and each participant $\pp[i]$ has local
type:
\[
\T[\pp[i]] = \toc{\sfmu} \ty[0].\;
\pp[i{-}1]\mkern2mu\toc{\&}\mkern2mu
\left\{
\begin{array}{@{}l@{}}
\lab[0] : \tsel{\pp[i{+}1]}{ \lab[0] : \ty[0] } \\[3pt]
\lab[1] : \tsel{\pp[i{+}1]}{ \lab[0] : \toc{\sfmu} \ty[1].\;
  \pp[i{-}1]\mkern2mu\toc{\&}\mkern2mu
  \left\{
  \begin{array}{@{}l@{}}
  \lab[0] : \tsel{\pp[i{+}1]}{ \lab[0] : \ty[1] } \\[3pt]
  \lab[1] : \tsel{\pp[i{+}1]}{ \lab[1] : \ty[0] }
  \end{array}
  \right\}
  }
\end{array}
\right\}
\]
Each participant acts as one bit of a ripple-carry counter,
receiving from $\pp[i-1]$ and sending to $\pp[i+1]$. Labels
$\lab[0]$ and $\lab[1]$ encode ``no carry'' and ``carry'', while
$\ty[0]$ and $\ty[1]$ record whether the current bit stores $0$ or
$1$; receiving $\lab[1]$ flips the bit, and only state $\ty[1]$
propagates the carry onward.
\end{example}

}{}

\end{document}